\newif\ifdraft
\newif\iffullversion
\newif\iftsc

\draftfalse
\fullversiontrue
\tscfalse

\RequirePackage{amsmath}

\iftsc
\documentclass[journal=tosc,submission]{iacrtrans}
\else
\documentclass[a4paper,11pt]{article}
\fi

\iffullversion
\usepackage[margin=1in]{geometry}
\fi

\usepackage[T2A, T1]{fontenc}
\usepackage[utf8]{inputenc}
\usepackage[english]{babel}
\usepackage{amsthm,amsfonts,amssymb}
\usepackage{graphicx}
\usepackage{bbm}
\usepackage{color}
\usepackage{xcolor}
\usepackage{xspace}
\usepackage{array}
\usepackage[skins,theorems]{tcolorbox}
\usepackage[ruled,vlined,linesnumbered, procnumbered]{algorithm2e}
\usepackage{declmath}
\usepackage[giveninits=false,backend=biber,style=alphabetic,maxnames=1000,sorting=nyt,hyperref,backref,backrefstyle=none]{biblatex}
\usepackage[pdfborderstyle={/S/U/W 0.3}]{hyperref}
\usepackage{etoolbox,tikz,longtable}
\usepackage{csquotes}
\usepackage{makeidx}\makeindex
\usepackage{authblk}

\iffullversion
\usepackage{textcomp}
\usepackage{lmodern}
\usepackage{tgpagella}
\usepackage{substitutefont}
\substitutefont{T2A}{\rmdefault}{fcm}
\fi

\let\oldnl\nl
\newcommand{\nonl}{\renewcommand{\nl}{\let\nl\oldnl}}

\newtheorem{thm}{Theorem}
\newtheorem{lemm}[thm]{Lemma}
\newtheorem{defi}[thm]{Definition}
\newtheorem{nrclaim}[thm]{Claim}

\newenvironment{subalgorithm}[1][htbp]
{
	\begin{algorithm}[#1]%
	}{\end{algorithm}}

\iffullversion

\renewenvironment{proof}[1][]{\medskip\noindent
	\textit{Proof\ifx\proof#1\proof\else\ #1\fi. }}{\qed\medskip}
\fi

\allowdisplaybreaks[1]

\newcommand\tocsectionstar[1]{%
	\section*{#1}%
	\addcontentsline{toc}{section}{#1}%
}

\newcommand{\parag}[1]{\medskip\noindent\textbf{#1 }}

\declare{macro=\ket, argspec=[1], code=\lvert#1\rangle, nopage=true, label = ket,  placeholder=\ket\psi,  description={A quantum state, a normalized vector in a Hilbert space}}
\declare{macro=\ketun, argspec=[1], code=\lvert#1 ), noindex=true, nopage=true,  placeholder=\ketun{\psi},  description={Vector in a Hilbert space, a not normalized quantum state}}
\declare{macro=\bra, argspec=[1], code=\langle#1\rvert, noindex=true, nopage=yes, placeholder=\bra\Psi,  description={Conjugate transpose of $\ket\Psi$}}
\newcommand{\braket}[2]{\langle #1 | #2 \rangle }
\newcommand{\proj}[1]{\ket{#1}\bra{#1} }

\DeclareMathOperator{\Hil}{\mathcal{H}}
\declare{macro=\Uni, code={\mathsf{U}},  noindex=yes, nopage=yes,  placeholder={\Uni},  description={A unitary}}
\declare{macro=\Puni, code={\mathsf{P}}, noindex=true, nopage=true, placeholder={\Puni}, label = permutation,  description={The permutation unitary}}
\declare{macro=\Juni, code={\mathsf{J}}, placeholder={\Juni_R}, label = projector,  description={Projector on relation $ R $.}}
\declare{macro=\HT, code={\mathsf{HT}},  placeholder={\HT_n}, label = hadamard,  description={The Hadamard transform}}
\declare{macro=\QFT, code={\mathsf{QFT}}, placeholder={\QFT_N}, label = qft,  description={The Quantum Fourier Transform}}
\declare{macro=\id, code={\mathbbm{1}}, noindex=yes, placeholder=\id,  description={Identity on $n$ qubits, or on subset $n$.}}
\declare{macro=\V, code={\mathsf{V}}, placeholder={\V_R}, label = v,  description={The unitary outputting $ D\in R $.}}

\definecolor{darkgreen}{rgb}{0.01, 0.75, 0.24}
\definecolor{byzantium}{rgb}{0.8, 0.0, 0.8}
\definecolor{brown}{rgb}{0.5, 0.27, 0.11}
\definecolor{banana}{rgb}{1.0, 0.88, 0.21}
\definecolor{dgreen}{rgb}{.1,.5,.1}

\newcommand{\GAME}[1]{\textnormal{\textbf{Game #1}}}
\newcommand{\strings}{\lbrace 0,1\rbrace}
\declare{macro=\cO, argspec=[1], code=O\left( #1\right), nopage=yes, noindex=true,
placeholder=\cO{n}, label=obig,  description={Complexity class \href{https://en.wikipedia.org/wiki/Big_O_notation}{"big O"}}}
\declare{macro=\om, argspec=[1], code=\omega\left( #1\right), noindex=true, nopage=yes, placeholder=\om(n),  description={Complexity class "small $ \omega $"}}
\newcommand{\PR}{\mathbb{P}}
\newcommand{\PRover}[1]{\underset{#1}{\PR}}
\newcommand{\pr}[1]{\mathbb{P}\left[\vphantom{a^b}#1\right]}
\newcommand{\sampuni}{\overset{\$}{\gets}}
\DeclareMathOperator{\eps}{\varepsilon}
\newcommand{\norm}[1]{\left\lVert #1 \right\rVert}

\newcommand{\cboxed}[2]{{\color{#1}\boxed{#2}}}

\declare{macro=\StO, code={\mathsf{StO}}, label = sto,   placeholder=\StO,  description={Standard Oracle}}
\declare{macro=\FO, code={\mathsf{FO}}, label = fo,  placeholder=\FO,  description={Fourier Oracle, $ \QFT^{YF}_N\circ\StO\circ\QFT^{\dagger YF}_N $}}
\declare{macro=\PhO, code={\mathsf{PhO}}, label = pho,  placeholder=\PhO,  description={Phase Oracle, $ \QFT^{Y}_N\circ\StO\circ\QFT^{\dagger Y}_N $}}
\declare{macro=\CStO, code={\mathsf{CStO}}, label = csto,  placeholder={\CStO_{\distrD}, \CStO_{\mathcal{Y}}},  description={Compressed Standard Oracle, for distribution $\distrD$ and for a conditionally uniform distribution over $\mathcal{Y}$}}
\declare{macro=\CFO, code={\mathsf{CFO}}, label = cfo,  placeholder={\CFO_{\distrD}},  description={Compressed Fourier Oracle for distribution $ \distrD $}}
\declare{macro=\Dec, code={\mathsf{Dec}}, label = dec,  placeholder={\Dec_{\distrD}},  description={Decompression procedure}}
\declare{macro=\CPhO, code={\mathsf{CPhO}}, label = cpho,  placeholder=\CPhO_{\distrU},  description={Compressed Phase Oracle}}
\declare{macro=\Ho, code={\mathsf{H}}, label = h, placeholder={\Ho, \G},  description={Compressed Oracle}}
\declare{macro=\F, code={\mathsf{F}}, noindex=true, nopage=true,  placeholder={\F},  description={Compressed Oracles}}
\declare{macro=\G, code={\mathsf{G}}, noindex=true, nopage=true,  placeholder={\G},  description={Compressed Oracle}}
\declare{macro=\sysC, code={\mathsf{C}}, noindex=true, nopage=true,  placeholder={\C},  description={Quantum construction}}
\declare{macro=\Sim, code={\mathsf{S}}, label = sim,  placeholder={\Sim},  description={Classical and quantum simulators.}}
\declare{macro=\Samp, code={\mathsf{Samp}}, label = samp,  placeholder={\Samp_{\distrD}(\mathcal{S})},  description={Algorithm preparing a superposition of samples of outputs of $f\gets\distrD$ on inputs from $ \mathcal{S} $.}}
\declare{macro=\Plu, code={\mathsf{Plu}}, noindex=true, nopage=true,  placeholder={\Plu},  description={Modular addition}}
\declare{macro=\Count, code={\mathsf{Count}}, noindex=true, nopage=true,  placeholder={\Count_{\distrD}},  description={Count the number of queries made to $\CFO_{\distrD}$}}
\declare{macro=\REM, code={\textnormal{REM}}, noindex=true,nopage=true, placeholder={\REM},  description={"Remove" tag}}
\declare{macro=\UPD, code={\textnormal{UPD}}, noindex=true,nopage=true, placeholder={\UPD},  description={"Update" tag}}
\declare{macro=\ADD, code={\textnormal{ADD}}, noindex=true,nopage=true, placeholder={\ADD},  description={"Add" tag}}
\declare{macro=\NOT, code={\textnormal{NOT}}, noindex=true,nopage=true, placeholder={\NOT},  description={"Do nothing" tag}}
\declare{macro=\Find, code={\textnormal{Find}}, label = find,  placeholder={\Find},  description={Event of measurement of the relation $R$ returning $1$}}
\declare{macro=\Bad, code={\textnormal{Bad}},  placeholder={\Bad}, label = bad,  description={A "bad" event in a game.}}
\declare{macro=\funAdd, code={\mathsf{Add}}, label = add,  placeholder={\funAdd},  description={Function adding $x$ to the compressed database}}
\declare{macro=\funSub, code={\mathsf{Sub}},noindex=true, nopage=true,   placeholder={\funSub},  description={Subtracting mod $N$ function}}
\declare{macro=\funRem, code={\mathsf{Rem}}, label = rem,  placeholder={\funRem},  description={Removing $\ize=0$ from the database}}
\declare{macro=\funUpd, code={\mathsf{Upd}}, label = upd,  placeholder={\funUpd},  description={Updating $\eta$ in the database}}
\declare{macro=\funLocate, code={\mathsf{Locate}}, label = locate,  placeholder={\funLocate},  description={Locate the position of $x$ in the database}}
\declare{macro=\funLarger, code={\mathsf{Larger}}, label = larger,  placeholder={\funLarger},  description={A unitary for comparing two bit-strings}}
\declare{macro=\funPath, code={\mathsf{SpPath}}, label = path,  placeholder={\funPath(s,G)},  description={Function constructing an input to $ \sponge $ leading to a given node}}
\declare{macro=\funClean, code={\mathsf{Clean}}, label = clean,  placeholder={\funClean},  description={Clean up function for auxiliary register}}
\declare{macro=\Queries, code={\mathsf{Queries}}, label = queries,  noindex=true, nopage=true,  placeholder={\Queries},  description={Decide which $ x $'s were queried.}}
\declare{macro=\coll, code={\textnormal{coll}}, noindex=true, nopage=true, placeholder={f_{\coll}},  description={Bound on the probability of a collision}}
\declare{macro=\preim, code={\textnormal{preim}}, noindex=true, nopage=true, placeholder={R_{\preim}},  description={The preimage relation}}
\declare{macro=\fun, code={\textnormal{fun}}, noindex=true, nopage=true, placeholder={R_{\fun}},  description={Relation of collisions in inputs of $ \phif $}}
\declare{macro=\opt, code={\textnormal{opt}}, noindex=true, nopage=true, placeholder={\sigma_{\opt}},  description={The optimal way to look at a given permutation.}}
\declare{macro=\sort, code={\textnormal{sort}}, noindex=true, nopage=true, placeholder={\Uni_{\sort}},  description={Sorting unitary.}}
\declare{macro=\Good, code={\textnormal{Good}}, noindex=true, nopage=true, placeholder={\ket{\Psi_q^{\Good}}},  description={The state imitating interaction with $ \Per $.}}

\declare{macro=\De, code={\mbox{\usefont{T2A}{\rmdefault}{m}{n}\CYRD}},  placeholder={D,\Delta,\De}, label = database,  description={Prepared database in the standard basis (and the database register), prepared database in the Fourier basis, and the unprepared databse}}
\declare{macro=\ize, code={\mbox{\usefont{T2A}{\rmdefault}{m}{n}\cyri}},  placeholder={y,\eta,\ize}, label = yvalues,  description={Values in the $Y$ register of a database in different bases}}

\makeatletter
\newcommand\xleftrightarrow[2][]{%
	\ext@arrow 9999{\longleftrightarrowfill@}{#1}{#2}}
\newcommand\longleftrightarrowfill@{%
	\arrowfill@\leftarrow\relbar\rightarrow}
\makeatother

\DeclareMathOperator{\pub}{\textnormal{pub}}
\DeclareMathOperator{\priv}{\textnormal{priv}}
\DeclareMathOperator{\sysR}{\mathsf{R}}
\DeclareMathOperator{\sysT}{\mathsf{T}}
\declare{macro=\advD, code={\mathsf{D}}, label = dist,  placeholder={\advD},  description={The distinguisher}}
\declare{macro=\advA, code={\mathsf{A}}, label = adversary,  placeholder={\advA, \advB},  description={An adversary, a classical or quantum algorithm}}
\newcommand{\advB}{\mathsf{B}}

\DeclareMathOperator{\Env}{\mathsf{Env}}

\declare{macro=\Collapse, argspec=[1], code={\textnormal{\textbf{Collapse #1}}}, placeholder={\Collapse{1}}, label=collapse,  description={Collapsing game}}
\declare{macro=\meas, code={\mathsf{M}}, placeholder={\meas}, label=measurement, noindex=true, nopage=true,  description={Measurement in the computational basis}}
\declare{macro=\bits, argspec=[1], code={\bit^{#1}}, noindex=true, nopage=true,  placeholder=\bits n,  description={Bitstrings of length $n$}}
\newcommand{\bit}{\{0,1\}}
\declare{macro=\xor, code={\oplus}, label = xor,  description={Bitwise XOR}}
\declare{macro=\abs, argspec=[1], code={\left\lvert#1\right\rvert}, label = abs, nopage=true,  placeholder={\abs x},    description={Cardinality of a set $x$ / length of a string $x$/ absolute value}}

\declare{macro=\sponge, code=\textsc{Sponge}, label = sponge,  placeholder={\sponge_{\phif}[\pad,r,c] } ,  description={Sponge construction with the internal function $\phif$, capacity $c$, and rate $r$}}
\declare{macro=\pad, code=\textsc{pad}, label = pad,  description={Padding function}}
\declare{macro=\phif, code={\varphi}, label = phi, description={The map between states in $\sponge$.}}
\declare{macro=\phifbar, code=\bar{\phif}, label = phibar, description={The map between states with its output limited to the first $r$ bits}}
\declare{macro=\phifhat, code=\hat{\phif}, label = phihat, description={The map between states with its output limited to the last $c$ bits}}
\declare{macro=\setc, code=\mathcal{C}, label = cinner,  description={The set of inner states, generalization of $\strings^c$, inner part of $ s\in\seta\times\setc $ denoted by $ \hat{s} $}}
\declare{macro=\seta, code=\mathcal{A}, label = alpha,  description={The alphabet set of outer states, generalization of $ \bits{r} $, outer part of $ s\in\seta\times\setc $ denoted by $ \bar{s} $}}
\declare{macro=\setr, code=\mathcal{R}, label = root,  description={The set of rooted supernodes}}
\declare{macro=\sete, code=\mathcal{U}, label = uset,  description={The set of supernodes with outgoing edges}}
\declare{macro=\setz, code=\mathcal{Z}, noindex=true, nopage=true, description={The set of outputs of a function}} 
\declare{macro=\sets, code=\mathcal{S}, noindex=true, nopage=true, description={The forbidden set}}
\declare{macro=\setp, code=\mathcal{P}, noindex=true, nopage=true, description={The set of previous queries}}
\declare{macro=\sett, code=\mathcal{T}, noindex=true, nopage=true, description={The set of previous queries and $ 0 $}}
\declare{macro=\setb, code=\mathcal{B}, noindex=true, nopage=true, description={The set of previous queries and $ 0 $}}
\declare{macro=\vertices, code=\mathcal{V}, label = vertices,  description={The set of vertices of a sponge graph}}
\declare{macro=\edges, code=\mathcal{E}, label = edges,  description={The set of edges of a sponge graph}}
\declare{macro=\distrD, code=\mathfrak{D}, label = distr,  description={A distribution.}} 
\declare{macro=\distrI, code=\mathfrak{I}, noindex=yes, nopage=true, description={A distribution.}} 
\declare{macro=\distrU, code=\mathfrak{U}, label = uniformdistr,  description={The uniform distribution.}}

\title{
	Quantum Lazy Sampling and Game-Playing Proofs for Quantum Indifferentiability 
}

\iffullversion
\author[1]{Jan Czajkowski\thanks{j.czajkowski@uva.nl}}
\author[2]{Christian Majenz\thanks{christian.majenz@gmail.com}}
\author[1]{Christian Schaffner\thanks{c.schaffner@uva.nl}}
\author[2]{Sebastian Zur\thanks{sebastian.zur@cwi.nl}}
\affil[1]{QuSoft, University of Amsterdam}
\affil[2]{QuSoft, CWI}
\else
\author{}
\institute{}
\fi

\begin{document}
\maketitle

\ifdraft
\begin{center}
	{\color{red} {\Huge DRAFT} }
\end{center}
\fi

\iftsc
\keywords{post-quantum security \and QROM \and quantum indifferentiability \and sponge construction \and compressed oracles}
\fi

\begin{abstract}
Game-playing proofs constitute a powerful framework for non-quantum cryptographic security arguments, most notably applied in the context of indifferentiability. An essential ingredient in such proofs is lazy sampling of random primitives.
We develop a quantum game-playing proof framework by generalizing two recently developed proof techniques. First, we describe how Zhandry's compressed quantum oracles~(Crypto'19) can be used to do quantum lazy sampling of a class of non-uniform function distributions. Second, we observe how Unruh's one-way-to-hiding lemma~(Eurocrypt'14) can also be applied to compressed oracles, providing a quantum counterpart to the fundamental lemma of game-playing. 
Subsequently, we use our game-playing framework to prove quantum indifferentiability of the sponge construction, assuming a random internal function. 
\end{abstract}

\iffullversion
\newpage
\tableofcontents
\newpage
\fi

\section{Introduction}
The modern approach to cryptography relies on mathematical rigor: Trust in a given cryptosystem is mainly established by proving that, given a set of assumptions, it fulfills a security definition formalizing real-world security needs. Apart from the definition of security, the mentioned assumptions include the threat model, specifying the type of adversaries we want to be protected against. One way of formalizing the above notions is via \emph{games}, i.e.\ programs interacting with the adversaries and outputting a result signifying whether there has been a breach of security or not. Adversaries in this picture are also modeled as programs, or more formally Turing machines.

The framework of game-playing proofs introduced by Bellare and Rogaway in \cite{bellare2004code}---modeling security arguments as games, played by the adversaries---is especially useful because it makes proofs easier to verify. Probabilistic considerations might become quite involved when talking about complex systems and their interactions; the structure imposed by games, however, simplifies them. In the game-playing framework, randomness can be, for example, considered to be sampled on the fly, making conditional events easier to analyze. A great example of that technique is given in the proof of the PRP/PRF switching lemma in \cite{bellare2004code}.

In this work we focus on idealized security notions; In the Random Oracle Model (ROM) one assumes that the publicly accessible hash functions are in fact random \cite{bellare1993random}. This is a very useful assumption as it simplifies proofs, but also cryptographic constructions designed with the ROM in mind are more efficient.

We are interested in the post-quantum threat model, which is motivated by the present worldwide efforts to build a quantum computer. It has been shown that quantum computers can efficiently solve problems that are considered hard for classical machines. Hardness of the factoring and discrete-logarithm problems is, e.g., important for public-key cryptography, but these problems can be solved efficiently on a quantum computer using Shor's algorithm \cite{shor1994algorithms}. The obvious formalization of the threat model is to include adversaries operating a fault-tolerant quantum computer, which is in particular capable of running the mentioned attacks. This model is the basis of the field of post-quantum cryptography \cite{bernstein2009post}.

While the attacks based on Shor's algorithm are the most well-known ones, public-key cryptography may not be the  only area with quantum vulnerabilities. Many cryptographic hash functions are based on publicly available compression functions \cite{merkle1989certified, damgaard1989design, bertoni2007sponge} and as such they could be run on a quantum machine. This fact motivates us to analyze adversaries that have quantum access to the public building blocks of the cryptosystem. Therefore, the quantum threat model takes us from the Random-Oracle Model \cite{bellare1993random}---often used in the context of hash functions---to the Quantum Random-Oracle Model~\cite{boneh2011random} (QROM), where the random oracle can be accessed in superposition.

Having highlighted a desirable proof structure---fitting the clear and easy-to-verify game-playing framework---and the need of including fully quantum adversaries with quantum access to random oracles into the threat model, we encounter an obvious challenge: defining a \emph{quantum} game-playing framework. In this article, we resolve that challenge and apply the resulting framework to the setting of hash functions. In the following paragraphs we describe our results and the main proof techniques we used to achieve them.

\medskip
\noindent\textbf{Our Results.}
We devise a quantum game-playing framework for security proofs that involve fully quantum adversaries. Our framework is based on a combination of two recently developed proof techniques: compressed quantum random oracles by Zhandry \cite{zhandry2018record} and the One-Way to Hiding (O2H) lemma by Unruh \cite{unruh2015revocable, ambainis2018quantum}. The former provides a way to lazy-sample a quantum-accessible random oracle, and the latter is a quantum counterpart of the Fundamental Game-Playing lemma---a key ingredient in the original game-playing framework. As our first main result we obtain a clean and powerful tool for proofs in post-quantum cryptography. The main advantage of the framework is the fact that it allows the translation of certain classical security proofs to the quantum setting, in a way that is  arguably more straight-forward than for previously available proof techniques.

On the technical side, we begin by re-formalizing Zhandry's compressed-oracle technique, which, as a by-product, makes a generalization to some non-uniform distributions of oracles relatively straightforward. In particular, we generalize the compressed-oracle technique of \cite{zhandry2018record} to a class of non-uniform distributions over functions, allowing a more general form of (quantum) lazy sampling. Our result allows to treat distributions with outputs that are independent for distinct inputs. Subsequently, we observe that the techniques of ``puncturing oracles'' proposed in~\cite{ambainis2018quantum} can also be applied to compressed oracles, yielding a more general version of the O2H lemma which forms the quantum counterpart of the fundamental game-playing lemma.

There are already some examples in the literature where generalized compressed oracles for non-uniform distributions have been used, e.g.~\cite{Alagic2018} (superposition oracle without compression that outputs 1 with probability $\epsilon$, we define the sampling procedure for such distribution in Appendix~\ref{sec:cfo-details}) and \cite{hamoudi2020quantum} (a generalization similar to ours but presented after our paper was posted online). We believe that the generalized formalism developed here will continue to be useful.

Punctured oracles are quantum oracles measured after every adversarial query. An important lemma that we prove is a bound on the probability that any of these measurements succeeds. We provide two proofs, one making heavy use of the results from \cite{CFHL21}, and one that has a potential of being more general but is considerably more complicated\footnote{This second proof is presented in Appendix~\ref{sec:full-proof}.}.
The bound on the probability of any of the measurements in a punctured oracle succeeding, together with the O2H lemma for compressed oracles provides a bound on the distinguishing advantage between a regular compressed oracle and a punctured one.
In Lemma~9 in \cite{zhandry2018record} indistinguishability of a compressed oracle and a punctured compressed oracle is also proven. The method, however, is different from ours and much fewer details are shown. A crucial difference though is that there are two nontrivial technical claims left implicit. According to \cite{Zhandry-emails}, however, there is a proof that maintains the claimed bound. As that proof is not publicly available at this point, we state and prove our indistinguishability bound for punctured oracles with almost the same bound. As far as we can tell, our bounds seem tight.

We go on to apply our quantum game-playing framework by proving quantum indifferentiability of the sponge construction~\cite{bertoni2007sponge} used in SHA3. More precisely, we show that the sponge construction is indifferentiable from a random oracle in case the internal function is a random function. We leave it as an interesting open question to extend our results to the setting of SHA3 which uses a permutation as internal function.
A reader mostly interested in the main result of this paper can go directly to section~\ref{sec:q-indifferentiability}. In the introduction of that section we give a high level explanation of the main concepts used in the proof of quantum indifferentiability.

\medskip
\parag{Related Work.} 
Indifferentiability is a security notion developed by Maurer, Renner, and Holenstein \cite{maurer2004indifferentiability} commonly used for hash-function domain-extension schemes \cite{coron2005merkle, bertoni2008indifferentiability}. Here, it captures the adversary's access to both the construction and the internal function.

The subject of quantum indifferentiability, addressed in our work, has been recently analyzed in two articles. Carstens, Ebrahimi, Tabia, and Unruh make a case in  \cite{carstens2018quantum} against the possibility of fulfilling the definition of indifferentiability for quantum adversaries. Assuming a technical conjecture, they prove a theorem stating that if two systems are perfectly (with zero advantage) quantumly indifferentiable then there is a stateless classical indifferentiability simulator. In the last part of their work they show that there cannot be a stateless simulator for domain-decreasing constructions---i.e.\ most constructions for hash functions.
Zhandry on the other hand \cite{zhandry2018record} develops a technique that allows to prove indifferentiability for the Merkle-Damg{\aa}rd construction. His result does not contradict the result of~\cite{carstens2018quantum}, as it handles the \emph{imperfect case}, albeit with a negligible error. The technique of that paper, compressed quantum oracles, is one of the two main ingredients of our framework. 
Recent work by Unruh and by Ambainis, Hamburg, and Unruh \cite{unruh2015revocable, ambainis2018quantum} form the second main ingredient of our result. They show the One-Way to Hiding (O2H) Lemma, which is the quantum counterpart of the Fundamental Game-Playing lemma---a key ingredient in the original game-playing framework. The O2H lemma provides a way to ``reprogram'' quantum accessible oracles on some set of inputs, formalized as ''punctured'' oracles in the latter paper.

The quantum security of domain-extension schemes has been the topic of several recent works. \cite{song17nmac,CHS19} study domain extension for message authentication codes and pseudorandom functions. For random inner function, \cite{zhandry2018record} has proven indifferentiability of the Merkle-Damg\aa rd construction which hence has strong security in the QROM. For hash functions in the standard model, quantum generalizations of collision resistance were defined in \cite{Unruh2016a,Alagic2018}. For one of them, collapsingness,  some domain-extension schemes including the Merkle-Damg\aa rd and sponge constructions, have been shown secure~\cite{czaj18sponge,fehr2018,Unruh2016}.

In a recent article~\cite{unruh2019logic} Unruh developed quantum Relational Hoare Logic for computer verification of proofs in (post-)quantum cryptography. There he also uses the approach of game-playing, but in general focuses on formal definitions of quantum programs and predicates. To investigate the relation between \cite{unruh2019logic} and our work in more detail one would have to express our results in the language of the new logic. We leave it as an interesting direction for the future.
The proof techniques of \cite{zhandry2018record} and \cite{ambainis2018quantum} have been recently used to show security of the 4-Round Feistel construction in \cite{hosoyamada2019tight} and of generic key-encapsulation mechanisms in \cite{jiang2019tighter} respectively. In \cite{chevalier2020security} the authors use compressed oracles for randomness in an encryption scheme using a random tweakable permutation (that is given to the algorithm externally). In \cite{CFHL21} quantum query complexity results are proven using the compressed oracles technique and provide a framework that simplifies such tasks.

\medskip
\noindent\textbf{Note.} A previous version of this paper contained an additional set of results about quantum lazy-sampling of random permutations and indifferentiability of SHA-3. Unfortunately there was a flaw in the argument and the technique for quantum lazy sampling random permutations presented there does not work as claimed. The difficulty lies in the fact that that permutations do not have independent outputs, which seems to require a completely different approach.

\medskip
\noindent\textbf{Organization.}
In Section~\ref{sec:preliminaries} we introduce the crucial classical notions we use. We provide the necessary definitions of the classical game-playing framework and indifferentiability needed in the remainder of the paper.
In Section~\ref{sec:oracles} we generalize the compressed-oracle technique of \cite{zhandry2018record} to non-uniform distributions over functions.
In Section~\ref{sec:comp-o2h} we prove a generalization of the O2H lemma of \cite{unruh2015revocable}, adapted to the use with compressed oracles for non-uniform distributions.
The quantum game-playing framework is defined via the general compressed quantum oracles that appear in security games, and we derive an upper bound on the probability of the \Find\ event for the case of puncturing a uniform oracle on collisions. 
In Section~\ref{sec:q-indifferentiability} we use these results to prove quantum indifferentiability of the sponge construction.

\section{Preliminaries}\label{sec:preliminaries}
We write $[N]:=\{0,1,\dots,N-1\}$ for the set of size $N$. We denote the Euclidean norm of a vector $\ket\psi\in\mathbb C^d$ by $\norm{\ket\psi}$. By $ x\gets\advA $ we denote sampling $ x $ from a distribution or getting the output of a randomized algorithm.
A summary of symbols used throughout the paper can be found in the \hyperref[sec:notation]{Symbol Index}.

\subsection{Classical Game-Playing Proofs}
Many proofs of security in cryptography follow the Game-Playing framework, proposed in~\cite{bellare2004code}. It is a very powerful technique as cryptographic security proofs tend to be simpler to follow and formulate in this framework. The central idea of this approach are \emph{identical-until-bad} games. Say games $\G$ and $\Ho$ are two programs that are syntactically identical except for code that follows after setting a flag $\Bad\symbolindexmark{\Bad}$ to one, then we call those games identical-until-bad. Usually in cryptographic proofs $\G$ and $\Ho$ will represent two functions that an adversary $\advA\symbolindexmark{\advA}$ will have oracle access to. In the following we denote the situation when $\advA$ interacts with $\Ho$ by $\advA^{\Ho}$. Then we can say the following about the adversary's view.
\begin{lemm}[Fundamental lemma of game-playing, Lemma~2 of \cite{bellare2004code}]\label{lem:clas-game-play}
Let $\G$ and $\Ho$ be identical-until-bad games and let $\advA$ be an adversary that outputs a bit $b$. Then
\begin{align}
\abs{\PR[b=1:b\gets \advA^{\Ho}]-\PR[b=1:b\gets \advA^{\G}]}\leq \PR[\Bad=1:\advA^{\G}].
\end{align}
\end{lemm}

\subsection{Indifferentiability}
In the Random-Oracle Model (ROM) we assume the hash function used in a cryptosystem to be a random function~\cite{bellare1993random}. This model is very useful in cryptographic proofs but might not be applicable if the discussed hash function is constructed using some internal function. The ROM can still be used in this setting but by assuming the internal function is random. The notion of security is then \emph{indistinguishability} of the constructed functions from a random oracle.
In most constructions however (such as in SHA-2~\cite{sha2} and SHA-3~\cite{sha3}), the internal function is publicly known, rendering the security notion of indistinguishability too weak. A notion of security dealing with this issue is \emph{indifferentiability} introduced by Maurer, Renner, and Holenstein \cite{maurer2004indifferentiability}.

Access to the publicly known internal function and the hash function constructed from it is handled by \emph{interfaces}. An interface to a system is an access structure defined by the format of inputs and expected outputs. Let us illustrate this definition by an example, let the system $\sysC$ under consideration be a hash function $\Ho_f:\bits{*}\to\bits{n}$, constructed using a function $f: \{0,1\}^n\to \{0,1\}^n$. Then the \emph{private} interface of the system accepts finite-length strings as inputs and outputs $n$-bit long strings. Outputs from the private interface are generated by the hash function, so we can write (slightly abusing notation) $\sysC^{\priv}=\Ho_f$. The public interface accepts $n$-bit long strings and outputs $n$-bit strings as well. We have that $\sysC^{\pub}=f$.
Often we consider one of the analyzed systems, $\sysR$, to be a random oracle. Then both interfaces are the same and output random outputs of appropriate given length.

The following definitions and Theorem~\ref{thm:class-compose} are the rephrased versions of definitions and theorems from \cite{maurer2004indifferentiability, coron2005merkle}. We also make explicit the fact that the definitions are independent of the threat model we consider---whether it is the classical model or the quantum model. To expose those two cases we write ``classical or quantum'' next to algorithms that can be classical or quantum machines; Communication between algorithms (systems, adversaries, and environments) can also be of two types, where quantum communication will involve quantum states (consisting of superpositions of inputs)---explained in more detail in the remainder of the paper.
\begin{defi}[Indifferentiability \cite{maurer2004indifferentiability}]\label{def:indifferentiability}
A cryptographic (classical or quantum) system $\sysC$ is $(q,\eps)$-\textnormal{indifferentiable} from $\sysR$, if there is an efficient (classical or quantum) \textnormal{simulator} $\Sim$ and a negligible function $\eps$ such that for any efficient (classical or quantum) \textnormal{distinguisher} $\advD$\symbolindexmark{\advD} with binary output (0 or 1) the advantage
\begin{equation}
\left| \PR\left[b=1: b\gets \advD[\sysC_k^{\priv}[\sysC_k^{\pub}], \sysC_k^{\pub}] \right] -\PR\left[b=1: b\gets \advD[\sysR_k^{\priv}, \Sim[\sysR_k^{\pub}]] \right]  \right| \leq \eps(k) \, ,
\end{equation}
where $k$ is the security parameter.
The distinguisher makes at most $q$ (classical or quantum) queries to $\sysC$.
\end{defi}
It is important to note that if $ \sysR $ is the random oracle (which is often the case), then both interfaces are the same.
By efficient we mean with runtime that is polynomial in the security parameter $k$. The definitions are still valid and the theorem below holds also if we interpret efficiency in terms of queries made by the algorithms. Note that then we can allow the algorithms to be unbounded with respect to runtime, the distinction between quantum and classical queries is still of crucial importance though.
By square brackets we denote (classical or quantum) oracle access to some algorithm, we also use $\advA^{\Ho}$ if the oracle is denoted by a more confined symbol.
\iffullversion
In Fig.~\ref{fig:indiff} we present a a scheme of the situation captured by Def.~\ref{def:indifferentiability}.
\begin{figure}
	\begin{center}
		\includegraphics[scale=1]{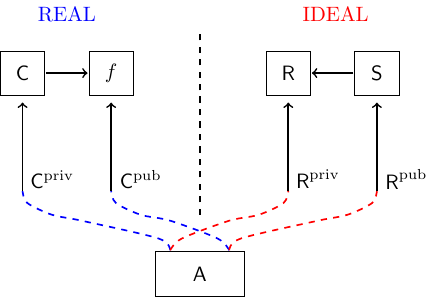}
	\end{center}
	\caption{A schematic representation of the notion of indifferentiability, Def.~\ref{def:indifferentiability}. Arrows denote ''access to'' the pointed system. \label{fig:indiff}}
\end{figure}
\fi

\begin{defi}[As secure as \cite{maurer2004indifferentiability}]
A cryptographic (classical or quantum) system $\sysC$ is said to be \textnormal{as secure as} $\sysC'$ if for all efficient (classical or quantum) environments $\Env$ the following holds: For any efficient (classical or quantum) attacker $\advA$ accessing $\sysC$ there exists another (classical or quantum) attacker $\advA'$ accessing $\sysC'$ such that the difference between the probability distributions of the binary outputs of $\Env[\sysC,\advA]$ and $\Env[\sysC',\advA']$ is negligible, i.e. 
\begin{equation}
\left| \PR\left[b=1:b\gets\Env[\sysC,\advA] \right] -\PR\left[ b=1:b\gets\Env[\sysC',\advA'] \right]  \right| \leq \eps(k) \, ,
\end{equation}
where $\eps$ is a negligible function.
\end{defi}

Indifferentiability is a strong notion of security mainly because if fulfilled it guarantees composability of the secure cryptosystem.
In the following we say that a cryptosystem $\sysT$ is \emph{compatible} with $\sysC$ if the interfaces for interacting of $\sysT$ with $\sysC$ are matching.
\begin{thm}[Composability \cite{maurer2004indifferentiability}]\label{thm:class-compose}
Let $\sysT$ range over (classical or quantum) cryptosystems compatible with $\sysC$ and $\sysR$, then $\sysC$ is $(q,\eps)$-indifferentiable from $\sysR$ if and only if for all $\sysT$, $\sysT[\sysC]$ is as secure as $\sysT[\sysR]$.
\end{thm}
Note that composability that is guaranteed by the above theorem holds only for \emph{single-stage} games \cite{ristenpart2011careful}.

Indifferentiability is a strong security notion guaranteeing that a lower-level function (e.g. a random permutation) can be used to construct a higher-level object (e.g.\ a variable input-length random function) that is ''equivalent'' to the ideal one---in the sense of Thm.~\ref{thm:class-compose}. Here, an adversary's complexity is measured in terms of the number of queries to the oracles only, not in terms of their time complexity.
In quantum indifferentiability adversaries are allowed to access the oracles in superposition. This is necessary in the post-quantum setting, as the building blocks of many hash functions---like e.g those of SHA3~\cite{sha3}---are publicly specified and can be implemented on a quantum computer.

\subsection{Quantum Computing}
The model of quantum adversaries we use is quantum algorithms making $q$ queries to an oracle. Each query is intertwined by a unitary operation acting on the adversary's state and all her auxiliary states. A general introduction to quantum computing can be found in \cite{nielsen2002quantum}. Here we will only introduce specific operations important to understand the paper.

Let us define the \emph{Quantum Fourier Transform} (QFT), a unitary change of basis that we will make heavy use of. For $N \in \mathbb{N}_{>0}$ and $x,\xi\in[N]=\mathbb{Z}_N$ the transform is defined as
\begin{equation}
\QFT_N\ket{x} := \frac{1}{\sqrt{N}}\sum_{\xi\in[N]} \omega_N^{\xi\cdot x}\ket{\xi}, \label{eq:qft-def}
\end{equation}\symbolindexmark{\QFT}
where $\omega_N := e^{\frac{2\pi i}{N}}$ is the $N$-th root of unity.
An important identity for some calculations is
\begin{equation}
\sum_{\xi\in[N]} \omega_N^{x\cdot \xi} \cdot \bar{\omega}_N^{x'\cdot \xi} = N \delta_{x,x'},
\end{equation}
where $\bar{\omega}_N = e^{-\frac{2\pi i}{N}}$ is the complex conjugate of $\omega_N$ and $\delta_{x,x'}$ is the Kronecker delta function.

If we talk about $n$ qubits the identity on their Hilbert space is denoted by $\id_n$. We write $\Uni^A$ to denote that we act with $\Uni$ on register $A$.

\section{Quantum-Accessible Oracles}\label{sec:oracles}
In the Quantum-Random-Oracle Model (QROM)~\cite{boneh2011random}, one assumes that the random oracle can be accessed in superposition. Quantum-accessible random oracles are motivated by the possibility of running an actual instantiation of the oracle as function on a quantum computer, which would allow for superposition access. In this section, oracles implement a function $f:\mathcal{X}\to\mathcal{Y}$ distributed according to some probability distribution $\mathfrak{D}$ on the set $\mathcal{F}$ of functions from $\mathcal{X}$ to $\mathcal{Y}$. Without loss of generality we set $\mathcal{X}=\mathbb{Z}_M$ and $\mathcal{Y}=\mathbb{Z}_N$ for some integers $M,N>0$.

In this section we give a formal treatment of quantum accessible oracles. We explain with special care the compressed-oracle technique of Zhandry \cite{zhandry2018record}. A quantum oracle can be viewed as a purification (extension to a higher-dimensional Hilbert space) of the adversary’s quantum state. The simplest purification extends the state to include a superposition of all full function tables from the set $ \mathcal{F} $. Note that the oracle gives access to a random function from the set $ \mathcal{F} $. The purification we talk about is called the oracle register. A quantum algorithm could simulate the access to the quantum oracle by preparing the oracle register and performing the correct update procedures every time the adversary makes a query. Such a simulator would not be efficient though, as the oracle register we just defined holds $ M $ entries (so one for each element of the domain) of the table of values in $ [N] $. The brilliant idea of Zhandry was to propose a procedure to lazy-sample a uniformly random function. By lazy-sampling we mean here to store just the queries asked by the adversary, not the whole function table. By doing that we limit the number of entries held by the simulator to $ q $ (the bound on the number of queries performed by the adversary). Our result in this section is generalizing Zhandry’s technique to independent distributions on functions: Such that outputs are distributed independently for any distinct inputs.

Classically, an oracle for a function $f$ is modeled via a tape with the queried input $x$ written on it, the tape is then overwritten with $f(x)$. The usual way of translating this functionality to the quantum circuit model is by introducing a special gate that implements the unitary $\Uni_f\ket{x,y}=\ket{x,y+f(x)}$. In the literature $+$ is usually the bitwise addition modulo $2$, but in general it can be any group operation. We are going to use addition in $\mathbb Z_N$.\footnote{Note that introducing the formalism using the group $\mathbb Z_N$ for some $N\in\mathbb N$ is quite general in the following sense: Any finite Abelian group $G$ is isomorphic to a product of cyclic groups, and the (quantum) Fourier transform with respect to such a group is the tensor product of the Fourier transforms on the cyclic groups, given the natural tensor product structure of $\mathbb C^G$. We use this formalism to define the most general compressed oracle technique. A reader that focuses on bitstrings can just consider $ + $ to be the bitwise XOR and $ \cdot $ the inner product of bitstrings.} 

In the case where the function $f$ is a random variable, so is the unitary $\Uni_f$. Sometimes this is not, however, the best way to think of a quantum random oracle, as the randomness of $f$ is accounted for using classical probability theory, yielding a hybrid description. To capture the adversary's point of view more explicitly, it is necessary to switch to the \emph{mixed-state formalism}. A mixed quantum state, or \emph{density matrix}, is obtained by considering the projector onto the one-dimensional subspace spanned by a pure state, and then taking the expectation over any classical randomness. Say that the adversary sends the query state $\ket{\Psi_0}=\sum_{x,y}\alpha_{x,y}\ket{x,y}$ to the oracle, the output state is then
\begin{align}
&\sum_{f}\PR[f:f\gets\mathfrak{D}] \; \Uni_f\proj{\Psi_0}\Uni^{\dagger}_f \; \otimes \; \proj{f}_{F} \nonumber \\
&= \sum_{f}\PR[f:f\gets\mathfrak{D}] \sum_{x,x',y,y'}\alpha_{x,y}\bar{\alpha}_{x',y'} \ket{x,y+f(x)}\bra{x',y'+f(x')}\otimes\proj{f}_{F},
\end{align}
where by $\bar{\alpha}$ we denote the complex conjugate of $\alpha$ and we have recorded the random function choice in a classical register $F$ holding the full function table of $f$.

In quantum information science, a general recipe for simplifying the picture and to gain additional insight is to \emph{purify} mixed states, i.e.\ to consider a pure quantum state on a system augmented by an additional register $E$, such that discarding $E$ recovers the original mixed state. In \cite{zhandry2018record} Zhandry applies this recipe to this quantum-random-oracle formalism.

In the resulting representation of a random oracle, the classical register $F$ is replaced by a quantum register holding a superposition of functions from $\mathfrak{D}$. The joint state before an adversary makes the first query with a state $\ket{\Psi_0}_{XY}$ is $\ket{\Psi_0}_{XY}\sum_{f\in\mathcal{F}} \sqrt{\PR[f:f\gets\mathfrak{D}]} \,\ket{f}_F$. 
The unitary that corresponds to $\Uni_f$ after purification will be called the \emph{Standard Oracle $\StO$} and works by reading the appropriate output of $f$ from $F$ and adding it to the algorithm's output register,$ \symbolindexmark{\StO} $
\begin{equation}
\StO \ket{x,y}_{XY}  \ket{f}_F :=  \ket{x,y + f(x)}_{XY} \ket{f}_F. \label{eq:sto-def}
\end{equation} 
Applied to a superposition of functions as intended,  $\StO$ will entangle the adversary's registers $XY$ with the oracle register $F$.

The main observation of \cite{zhandry2018record} is that if we change the basis of the initial state of the oracle register $F$, the redundancy of this initial state becomes apparent. 
If we are interested in, e.g., an oracle for a uniformly random function, the Fourier transform changes the initial oracle state $\sum_{f}\frac{1}{\sqrt{|\mathcal{F}|}}\ket{f}$ to a state holding only zeros $\ket{0^{M}}$, where $0\in\mathcal{Y}$. The uniform case is treated in great detail in \cite{Unruh-forthcoming}, there the case of random (invertible) permutations is also analyzed.

Let us start by presenting the interaction of the adversary viewed in the same basis, called the Fourier basis. The unitary operation acting in the Fourier basis is called the \emph{Fourier Oracle $\FO$}. Another important insight from \cite{zhandry2018record} is that the Fourier Oracle, instead of adding the output of the oracle to the adversary's output register, does the opposite: It adds the value of the adversary's \emph{output} register to the (Fourier-)transformed truth table $ \symbolindexmark{\FO} $
\begin{equation}
\FO \ket{x,\eta}_{XY} \ket{\phi}_F := \ket{x,\eta}_{XY} \ket{\phi - \chi_{x,\eta}}_F, \label{eq:fo-def}
\end{equation}
where $\phi$ is the transformed truth table $f$ and $\chi_{x,\eta}:=(0,\dots,0,\eta,0,\dots,0)$ is a transformed truth table equal to $0$ in all rows except for row $x$, where it has the value $\eta$. Note that we subtract $\chi_{x,\eta}$ so that the reverse of QFT returns addition of $f(x)$.

Classically, a (uniformly) random oracle can be ``compressed'' by lazy-sampling the responses, i.e.\ by answering with previous answers if there are any, and with a fresh random value otherwise. Is lazy-sampling possible for quantum accessible oracles?
Surprisingly, the answer is yes. Thanks to the groundbreaking ideas presented in \cite{zhandry2018record} we know that there exists a representation of a quantum random oracle that is efficiently implementable.

In the remainder of this section we present an efficient representation of oracles for functions $f$ sampled from product distributions. In the first part we introduce a general structure of quantum-accessible oracles.
In the second part we generalize the idea of compressed random oracles to deal with non-uniform distributions of functions. 
In Appendix~\ref{sec:cfo-details}, we provide additional details on the implementation of the procedures introduced in this section and step-by-step calculations of important identities and facts concerning compressed oracles. 
In Appendix~\ref{sec:uniform-oracles} we recall in detail the compressed oracle introduced in \cite{zhandry2018record}, where the distribution of functions is uniform and the functions map bitstrings to bitstrings. We show the oracle in different bases and present calculations that might be useful for developing intuition for working with the new view on quantum random oracles. 

\subsection{General Structure of the Oracles}
In this subsection we describe the general structure of quantum-accessible oracles that will give us a high-level description of all the oracles we define in this paper. A quantum-accessible random oracle consists of 
\begin{enumerate}
\item Hilbert spaces for the input $\Hil_{\mathcal{X}}$, output $\Hil_{\mathcal{Y}}$, and state registers $\Hil_{\mathcal{F}}$,
\item a procedure $\Samp_{\distrD}$ that, on input a subset of the input space of the functions in $\distrD$, prepares a superposition of partial functions on that subset of inputs with weights according 
 to the respective marginal of the distribution $\distrD$, 
\item an update unitary $\FO_{\distrD}$ that might depend on $\distrD$ (in the case of compressed oracles) or not (in the case of full oracles, Eq.~\eqref{eq:fo-def}).
\end{enumerate}
First of all, let us note that we use the Fourier picture of the oracle as the basis for our discussion. This picture, even though less intuitive at first sight, is simpler to handle mathematically.
The distribution of the functions we model by the quantum oracle are implicitly given by the procedure $\Samp_{\distrD}$ that when acting on the $\ket{0}$ state generates a superposition of values consistent with outputs of a function $f$ sampled from $\distrD$.

In the above structure the way we implement the oracle---in a compressed way, or acting on full function tables---depends on the way we define $\FO_{\distrD}$.

The definition of $\Samp_{\distrD}$ is such that $\Samp_{\distrD}(\mathcal{X})\ket{0^M}=\sum_{f\in\mathcal{F}}\sqrt{\PR[f\gets\distrD]}\ket{f}$ and is a unitary operator.

Quantum-accessible oracles work as follows. First the oracle state is prepared in an all-zero state. Then at every query by the adversary we run $\FO_{\distrD}$ which updates the state of the database. Further details are provided in the following sections.

\subsection{Non-uniform Oracles}
One of the main results of this paper is generalizing the idea of purification and compression of quantum random oracles to a class of non-uniform function distributions. We show that the compressed-oracle technique can be used to deal with distributions over functions with outputs independent of any prior interactions. Examples of such functions are random Boolean functions that output one with a given probability. 

We want to compress the following oracle
\begin{align}
&\StO \ket{x,y}_{XY}\sum_{f\in\mathcal{F}} \sqrt{\PR[f:f\gets\mathfrak{D}]} \, \ket{f}_F \nonumber\\
&= \sum_{f\in\mathcal{F}} \sqrt{\PR[f:f\gets\mathfrak{D}]} \, \ket{x,y+f(x)\mod N}_{XY} \, \ket{f}_F,\label{eq:sto-def-gen}
\end{align}
where $\distrD\symbolindexmark{\distrD}$ is a distribution on the set of functions $\mathcal{F}=\{f:\mathcal{X}\to\mathcal{Y}\}$.
The first ingredient we need is an operation that prepares the superposition of function truth tables according to the given distribution. More formally, we know a unitary that for all $\mathcal{S}\subseteq\mathcal{X}\symbolindexmark{\Samp}$
\begin{equation}
\Samp_{\distrD}(\mathcal{S})\ket{0^{\abs{\mathcal{S}}}}_{F(\mathcal{S})} =
	\bigotimes_{x\in \mathcal{S}} \sum_{y_x\in \mathcal{Y}} \sqrt{\PR[y_x=f(x):f\gets \distrD]} \ket{y_x}_{F(x)},\label{eq:def-add}
\end{equation}
where by $ F(x) $ we denote the register corresponding to $ x $. Later we give explicit examples of $\Samp_{\distrD}$ for different $\distrD$.
Applying QFT to the adversary's register gives us the \emph{Phase Oracle $ \PhO $} that changes the phase of the state according to the output value $ f(x) $. This picture is commonly used in the context of bitstrings but is not very useful in our context. Additionally transforming the oracle register brings us to the Fourier Oracle, that we will focus on. This series of transformations can be depicted as a chain of oracles:
\begin{align}
\StO \xleftrightarrow{\QFT_N^Y} \PhO \xleftrightarrow{\QFT_N^F} \FO,\label{eq:chainQFT}
\end{align}
going ``to the right'' is done by applying $ \QFT_N $ and ``to the left'' by applying the adjoint. Also note that since register $ Y $ holds a single value in $ \mathcal{Y} $ and register $ F $ holds values in $ \mathcal{Y}^M $, the transform above is an appropriate tensor product of $ \QFT_N $.
The non-uniform Fourier Oracle is defined as $\FO = \QFT^{YF}_N \circ\StO\circ\QFT^{\dagger\; YF}_N$, as a consequence of that definition we have
\begin{align}
\begin{split} \label{eq:fo-def-gen}
&\FO \ket{x,\eta}_{XY} \sum_{\phi}\frac{1}{\sqrt{N^M}}\sum_{f\in\mathcal{F}} \sqrt{\PR[f:f\gets\mathfrak{D}]} \; \omega_N^{\phi \cdot f} \; \ket{\phi}_F \\
&= \ket{x,\eta}_{XY}\sum_{\phi}\frac{1}{\sqrt{N^M}}\sum_{f\in\mathcal{F}} \sqrt{\PR[f:f\gets\mathfrak{D}]} \; \omega_N^{\phi \cdot f} \; \ket{\phi -\chi_{x,\eta}\mod N }_F.
\end{split}
\end{align}
The main difference between uniform oracles and non-uniform oracles is that in the latter, the initial state of the oracle in the Fourier basis is not necessarily an all-zero state.
That is because the unitary $\Samp_{\distrD}$---that is used to prepare the initial state---is not the adjoint of the transformation between oracle pictures, like it is the case for the uniform distribution.

Before we give all details of $\Samp_{\distrD}$ let us discuss the two bases: the Fourier basis and the prepared basis.
To deal with the difference between the initial $0$ state and the initial Fourier basis truth tables we use yet another alphabet and define $\De$ (pronounced as [d$\varepsilon$]) which denotes the unprepared database. We call it like that because the initial state of $\De$ is the all-zero state. Moreover only by applying $\QFT^D_N\circ\Samp^D_{\distrD}$ we transform it to $\Delta$, i.e the Fourier basis database. As we will see, operations on $\De$ are more intuitive and easier to define.
We denote an unprepared database by $\ket{\De}_D=\ket{x_1,\ize_1}_{D_1}\ket{x_2,\ize_2}_{D_2}\cdots \ket{x_q,\ize_q}_{D_q} \symbolindexmark{\ize}$ (where the Cyrillic letter $\ize$ is pronounced as [i]). By $\Delta^Y(x)$ we denote the $\eta$ value corresponding to the pair in $\Delta$ containing $x$ and by $ \De^X $ we denote the $ x $ values in $ \De $.
The intuition behind the preparation procedure is to initialize the truth table of the correct distribution in the correct basis. 
This notion is not visible in the uniform-distribution case, because there the sampling procedure for the uniform distribution $ \distrU $ is the Fourier transform: $\Samp_{\distrU}=\QFT_{N}^{\dagger}$, and the database pictures $\Delta$ and $\De$ are equivalent.
The following chain of databases similar to Eq.~\eqref{eq:chainQFT} represents different pictures, i.e.\ bases, in which the compressed database can be viewed $ \symbolindexmark{\De} $
\begin{align}
	\ket{\De} \xleftrightarrow{\Samp_{\distrD}} \ket{D} \xleftrightarrow{\QFT_N^{D^Y}} \ket{\Delta}.\label{eq:chain-database}
\end{align}

Before defining compressed oracles for non-uniform function distributions, let us take a step back and think about classical lazy sampling for such a distribution. Let $f$ be a random function from a distribution $\mathfrak D$. In principle, lazy sampling is always possible as follows. When the first input $x_1$ is queried, just sample from the marginal distribution for $f(x_1)$. Say the outcome is $y_1$ for the next query with $x_2$, we sample from the \emph{conditional distribution} of $f(x_2)$ given that $f(x_1)=y_1$, etc.

Whether actual lazy sampling is feasible depends on the complexity of sampling from the conditional distributions of function values given that a polynomial number of other function values are already fixed.

The method for quantum lazy sampling that we generalize in this paper is applicable only to a certain class of distributions. The distributions that we analyze must be independent for every input. By $f(\mathcal{S})$ we denote the part of the full truth table of $f$ corresponding to inputs from $\mathcal{S}$. Below we provide a definition of \emph{product} distributions:
\begin{defi}[Product distribution]\label{def:local-distr}
	A distribution $ \distrD $ on a set of functions $ \mathcal{F}\subseteq\{f:\mathcal{X}\to\mathcal{Y}\} $ is called \emph{product} if for all disjoint $ \mathcal{S}_1,\mathcal{S}_2\subseteq\mathcal{X} $, $f(\mathcal{S}_1)$ and $f(\mathcal{S}_2)$ are independently distributed when $f\gets\distrD$.
\end{defi}

The situation when constructing compressed superposition oracles for non-uniformly distributed random functions is very similar. In this case we need the operations $\Samp_{\distrD}(\mathcal{S})$ to be efficiently implementable for the compressed oracle to be efficient. Here, $ \mathcal{S}\subseteq \mathcal{X}$. By inputting a set to $\Samp_{\distrD}$ we mean that the operation will prepare a superposition of outputs to elements of the set.

Let us now come back to Definition~\ref{def:local-distr}, we want to translate the constraint on distributions to constraints on the quantum sampling procedure. The definition requires that the distribution is independent for any $ \mathcal{S}_1 $ and $ \mathcal{S}_2 $, this leads to the following requirement on sampling procedures:
\begin{align}\label{eq:local-add}
	\forall \mathcal{S}_1,\mathcal{S}_2 \subseteq\mathcal{X}:\;  \Samp_{\distrD}(\mathcal{S}_1\cup \mathcal{S}_2)= \Samp_{\distrD}(\mathcal{S}_1) \circ\Samp_{\distrD}(\mathcal{S}_2).
\end{align}

Let us present a detailed definition of sampling procedures for product distributions.
\begin{defi}[Sampling procedure for a product $ \distrD $]\label{def:local-samp}
	A sampling procedure $ \Samp_{\distrD} $ for a product distribution $ \distrD $ (as defined in Def.~\ref{def:local-distr}) is a family of unitary operators
	\begin{align}
		\left\{ \Samp_{\distrD}(\mathcal{S}_1) : \mathcal{S}_1\subseteq\mathcal{X} \right\},
	\end{align}
	where each operator fulfills the following conditions:
	\begin{itemize}
		\item[(i)] It is efficiently implementable in the number of inputs $ \abs{\mathcal{S}_1} $.
		\item[(ii)] It prepares the appropriate superposition on the zero state:
		\begin{align}
			\Samp_{\distrD}(\mathcal{S}_1)\ket{0^{\abs{\mathcal{S}_1}}}_{F_1}= \sum_{\vec{y}_1\in\mathcal{Y}^{|\mathcal{S}_1|}}\sqrt{\PRover{f\gets\distrD}\left[ f(\mathcal{S}_1)=\vec{y}_1 \right]} \ket{\vec{y}_1}_{F_1}.
		\end{align}
		\item[(iii)] The operators are independent, so for $\mathcal{S}_1,\mathcal{S}_2\subseteq\mathcal{X} $ such that $ \mathcal{S}_1\cap\mathcal{S}_2=\emptyset $ we have:
		\begin{align}
			&\Samp_{\distrD}^{F_1F_2}(\mathcal{S}_1\cup \mathcal{S}_2) =\Samp_{\distrD}^{F_1}(\mathcal{S}_1) \circ\Samp_{\distrD}^{F_2}(\mathcal{S}_2) 
		\end{align}
	and $F_1$ and $ F_2 $ are different quantum registers.
	\end{itemize}
\end{defi}

Note that for $ \Samp_{\distrD}(\mathcal{S})$ to be efficient, it is not sufficient that the probability distributions $\distrD$ are classically efficiently samplable. This is because running a reversible circuit obtained from a classical sampling algorithm on a superposition of random inputs will, in general, entangle the sample with the garbage output of the reversible circuit. The problem of efficiently creating a superposition with amplitudes $\sqrt{p(x)}$ for some probability distribution $p$ has appeared in other contexts, e.g. in classical-client quantum fully homomorphic encryption \cite{Mahadev2018}.
 
An interesting example of a distribution that is not product but which we can quantumly lazy-sample is the following: It is uniform for inputs in $ \bits{n}\setminus \{x\} $ for any $ x $ and is fully determined on the ``last'' input: $ f(x)=\bigoplus_{x'\neq x}f(x') $.
 
Before we state the algorithm that realizes the general \emph{Compressed Fourier Oracle $\CFO_{\distrD}$} we provide a high-level description of the procedure. The oracle $\CFO_{\distrD}$ is a unitary algorithm that performs quantum lazy sampling, maintaining a compressed database of the adversary's queries. For the algorithm to be correct---indistinguishable for all adversaries from the full oracle---it has to respect the following invariants of the database: The full oracle is oblivious to the \emph{order} in which a set of inputs is queried. Hence the same property has to hold for the compressed oracle, i.e.\ we cannot keep entries $(x,\eta)$ in the order of queries. We ensure this property by keeping the database \emph{sorted} according to $x$.

The second issue concerns the danger of storing too much information. If after the query we save $(x,\eta)$ in the database but the resulting entry mapped to $(x,0)$ in the \emph{unprepared} basis, i.e.\ the basis before applying $ \Samp $, then the  compressed database would entangle itself with the adversary, unlike in the case of the full oracle. Hence the database cannot contain $0$ in the unprepared basis.

In the following we sketch the workings of the quantum algorithm $ \CFO_{\distrD} $ responsible for updating the oracle register. The set of inputs $ \mathcal{X} $ is expanded by the symbol $ \perp $, denoting an empty entry in the quantum database.
\paragraph*{$\CFO_{\distrD}$:}
On input $\ket{x,\eta}$ do the following:
\begin{enumerate}
	\item Find the index $l\in [q]$ of the register holding the first $ x_l $ from the right that is $ x_l < x $, we should insert $(x,\eta)$ into this register.
	\item If $x\neq x_l$: insert $x$ in a register after the last element of the database and shift it to position $l$, moving the intermediate registers backwards.
	\item Apply $ \QFT_N^{D^Y_l} \circ \Samp_{\distrD}^{D_l}(x) $ to change the basis to the Fourier basis (in which the adversary's $\eta$ is encoded) and update register $D_l$ to contain $(x_l,\eta_l-\eta)$, change the basis back to original by applying $ \Samp_{\distrD}^{\dagger D_l}(x) \circ \QFT^{\dagger D^Y_l}_N $. 
	\item Check if register $ l $ contains a pair of the form $(x_l,0)$, if yes subtract $x$ from the first part to yield $ (\perp,0)$ and shift it back to the end of the database. 
	\item Uncompute\footnote{Uncomputing a function means in the context of quantum computing applying the conjugate of the unitary calculating this function.} $ l $.
\end{enumerate}
If after $q$ queries the database has a suffix of $u$ pairs of the form $(\perp,0)$, we say the database has $s=q-u$ non-padding entries.

Using this notation,  Alg.~\ref{alg:generalCFO} defines the procedure of updates of the database of the compressed database. We refer to Appendix~\ref{sec:algorithms} for the fully detailed description of  $\CFO_{\distrD}$.
\begin{algorithm}[hbt]
	\caption{General $\CFO_{\distrD}\symbolindexmark{\CFO}$}\label{alg:generalCFO}
\DontPrintSemicolon
\SetKwInOut{Input}{Input}\SetKwInOut{Output}{Output}
		\Input{Unprepared database and adversary query: $\ket{x,\eta}_{XY}\ket{\De}_D$}
		\Output{$\ket{x,\eta}_{XY}\ket{\De'}_{D}$}
		\BlankLine
		 	Count in register $ S $ the number of non-padding ($ \De^X\neq\perp $) entries $ s $ in $ D $ \label{line:queries1}\;
		\If(\tcp*[f]{add}){$x\not\in \De^X$}{
		 Insert $x$ to $\De^X$ in the right place and add $1$ to $S$ \label{line:x-alwaysin-D}  \tcp*[r]{keeping $ \De^X $ sorted}
	}
		 Apply $\QFT_N^{D^Y(x)}\Samp_{\distrD}^{D^Y(x)}(x)$\tcp*[r]{prepare the database: $\De(x)\mapsto\Delta(x)$}
		 Subtract $\eta$ from $\Delta^Y(x)$ \tcp*[r]{update entry with $x$}
		 Apply $\Samp_{\distrD}^{\dagger D(x)}(x)\QFT^{\dagger D^Y(x)}_N$ \tcp*[r]{unprepare the database: $\Delta(x)\mapsto\De(x)$}
		 In register $ L $ save location $ l $ of $ x $ in $ \De $ \;
		\If(\tcp*[f]{remove or do nothing}){$\De^Y_l=0$\label{line:ifeta0rem}}{
		 Remove $x$ from $D^X_l$ and shift register $D^X_l$ to the back \tcp*[r]{$ \De^X_l\mapsto\perp $}
	}
		\If(\tcp*[f]{$ x $ was removed}){$ \De^X_l\neq x $}{
		 Shift $ D^Y_l $ to the back and subtract $ 1 $ from $ S $\;
	}
		 Uncompute $ l $ from register $ L $\label{line:lazy-uncompL} \tcp*[r]{Algorithm~\ref{alg:lazy-uncompL}}
		 Uncompute $s$ from register $S$\label{line:queries2}\;
		Return $\ket{x,\eta}_{XY}\ket{\De'}_{D}$ \tcp*[r]{$\De'$ is the modified database}
\end{algorithm}

Below in Alg.~\ref{alg:lazy-uncompL} we explain how to uncompute $ l $ in line~\ref{line:lazy-uncompL} of Algorithm~\ref{alg:generalCFO}.
\begin{algorithm}[hbt]
	\caption{Uncompute $ L $ in line~\ref{line:lazy-uncompL} of Alg.~\ref{alg:generalCFO} \label{alg:lazy-uncompL}}
	\DontPrintSemicolon
		 Control on registers $ X $ and $ D^X $\;
		\For{$ i =1\dots,s-1 $}{
		\If{$ \De^X_i=x $}{
		 Subtract $ i $ from $ L $\;
		}
		\ElseIf{$ \De^X_i < x $ and $  x < \De^X_{i+1}  $}{
		 Subtract $ i+1 $ from $ L $\;
	}
		}
\end{algorithm}

In Alg.~\ref{alg:generalCFO} we use the fact that $ \Samp_{\distrD} $ is a local sampling procedure, Def.~\ref{def:local-samp}; Note that we write $ \Samp_{\distrD}^{D(x)}(x) $, so the sampling is independent from all queries that are already in the database.

We would like to stress that to keep the compressed oracle $\CFO_{\distrD}$ a unitary operation we always keep the database of size $q$. This can be easily changed by always appending an empty register $(\perp,0)$ at the beginning of each query of adversary $ \advA $. The current formulation of $\CFO_{\distrD}$ assumes that there is an upper bound on the number of queries made by the adversary, this is not a fundamental requirement.

The interface corresponding to the compressed Fourier oracle $\CFO_{\distrD}$ interprets the adversary's output register in the Fourier basis. When we want to change the basis to the standard one, we apply $ \QFT^{D^Y}_N $ to the database register and $ \QFT_N^{Y} $ to the adversary's output register. These basis changes give rise to the versions of oracle analogous to the full-oracle case: $ \symbolindexmark{\CStO} $
\begin{align}
	\CStO \xleftrightarrow{\QFT_N^{Y}} \CPhO \xleftrightarrow{\QFT_N^{D^Y}} \CFO.\label{eq:comp-chainQFT}
\end{align}
The intermediate oracle is the compressed phase oracle.

The decompression procedure for the general Compressed Fourier Oracle is given by Alg.~\ref{alg:general-dec}.
\begin{algorithm}
	\caption{General Decompression Procedure $\Dec_{\distrD}\symbolindexmark{\Dec}$}\label{alg:general-dec}
\DontPrintSemicolon
\SetKwInOut{Input}{Input}\SetKwInOut{Output}{Output}
\Input{Unprepared database: $\ket{\De}_{D}$}
\Output{Prepared, Fourier-basis truth table: $\ket{\phi(\De)}$}
\BlankLine
Count in register $ S $ the number of non-padding ($ x\neq\perp $) entries $ s $\;
Initialize register $ F $ in the state $\bigotimes_{x\in\mathcal{X}}\ket{0}$\;
\For(\tcp*[f]{Controlled on $ S $}){$i=1,2,\dots, s$}{
Swap register $D_i^Y$ with $F(x_i)$\;
}
\For{$x\in \mathcal{X}$ in descending order}{
\If{$F(x)$ holds a value $ \neq 0 $}{
Subtract $x$ from register $D_s^X$\;
Subtract $1$ from register $S$\;
}
}
Discard $D$ and $S$\;
Apply $\QFT_N^{F}\Samp^{F}_{\distrD}(\mathcal{X})$\tcp*[r]{Prepare the database}
\end{algorithm} 
The output of the decompression procedure $\phi(\De)$ is the state holding the prepared Fourier-basis truth table of the functions from $\distrD$, which by construction is consistent with the adversary's interaction with the compressed oracle. 

The decompression can be informally described as follows.
The first operation coherently counts the number of $ \De^X\neq\perp $ and stores the result in a register $S$.
Next we prepare a fresh all-zero initial state of a function from $\mathcal X$ to $\mathcal Y$, i.e.\ $\mathcal{X}$ registers of dimension $N$, all in the zero state. These registers will hold the final $\FO$ superposition oracle state.
The next step is swapping each $Y$-type register of the $\CFO$-database with the prepared zero state in the $\FO$ at the position indicated by the corresponding $X$-type register in the $\CFO$ database. This FOR loop is controlled on register $ S $. Note that after preparing $ S $ we do not modify $ S $ anymore in this step.
The task left to do is deleting $x$'s from $D$. It is made possible by the fact that the non-padding entries of the $\CFO$ database are nonzero and ordered. That is why we can iterate over the entries of the truth table $F$ and, conditioned on the entry not being $0$, delete the last entry of $D^X$ and reducing $S$ by one to update the number of remaining non-padding entries in the $\CFO$-database. Here the loop range does not depend on the size of the database, just the size of the domain.
Finally, we switch to the correct basis to end up with a full oracle of Fourier type, i.e. a $\FO$.

\begin{thm}[Correctness of $\CFO_{\distrD}$]\label{thm:indist-cfofo-gen}
Say $\distrD$ is a product distribution (Def.~\ref{def:local-distr}) over functions, let $\CFO_{\distrD}$ be defined as in Alg.~\ref{alg:generalCFO} and $\FO$ as in Eq.\eqref{eq:fo-def-gen}. Let $z$ be a random arbitrarily distributed string. Then for any quantum adversary $\advA$ making $q$ quantum queries we have
\begin{equation}
	\abs{ \pr{b=1:b\gets \advA^{\FO}(z)}- \pr{b=1:b\gets \advA^{\CFO}(z)} } =0.
\end{equation}
\end{thm}
\begin{proof}[Proof sketch]
We will show that
\begin{equation}\label{eq:lazy-psifo-psicfo-sketch}
	\ket{\Psi_{\FO}}_{AF} = \Dec^D_{\distrD}\ket{\Psi_{\CFO}}_{AD},
\end{equation}
where $\ket{\Psi_{\FO}}_{AF}$ is the joint state of the adversary and the oracle resulting from the interaction of $\advA$ with $\FO$ and $\ket{\Psi_{\CFO}}_{AD}$ is the state resulting from the interaction of $\advA$ with $\CFO_{\distrD}$. The state $ \ket{\Psi_{\FO}}_{AF} $ is generated by applying $ \prod_{i=1}^q \Uni_i\circ\FO $  to the $ \ket{\psi_{0}}_A\ket{0^M}_F $, where $ \ket{\psi_0}_A $ is the initial state of the adversary. In the case of the compressed oracle, the state $ \ket{\Psi_{\CFO}}_{AD} $ is generated by applying $ \prod_{i=1}^q \Uni_i\circ\CFO $  to  $ \ket{\Psi_{0}}_A\ket{(\perp,0)^q}_D $, where $ (\perp,0)^q $ denotes $ q $ pairs $ (\perp,0) $.

We can focus on the state equality from Eq.~\eqref{eq:lazy-psifo-psicfo-sketch} because if they are indeed equal, then any adversary measurement on $ \ket{\Psi_{\FO}}_{AF} $ will yield the output $ b=1 $ with the same probability as on $ \Dec^D_{\distrD}\ket{\Psi_{\CFO}}_{AD} $.

To prove that Eq.~\eqref{eq:lazy-psifo-psicfo-sketch} indeed holds we calculate a single query made to the compressed oracle. We can perform a detailed calculation of that procedure thanks to the assumption that $ \distrD $ is a product distribution (Def.~\ref{def:local-distr}) and the sampling procedure that constructed accordingly (Def.~\ref{def:local-samp}).

After we calculated the updated compressed database we can easily decompress it and compare with the corresponding updated full oracle register. All the details of this proof can be found in Appendix~\ref{sec:full-correctness}.
\end{proof}

\section{One-way to Hiding Lemma for Compressed Oracles}\label{sec:comp-o2h}
The fundamental game-playing lemma, Lemma~\ref{lem:clas-game-play}, is a very powerful tool in proofs that include a random oracle. A common use of the framework is to reprogram the random oracle in a useful way. The fundamental lemma gives us a simple way of calculating how much the reprogramming costs in terms of the adversary's advantage---the difference between probabilities of $\advA$ outputting 1 when interacting with one game or the other. The lemma that provides a counterpart to Lemma~\ref{lem:clas-game-play} valid for quantum accessible oracles is the \emph{One-Way to Hiding} (O2H) Lemma first introduced by Unruh in \cite{unruh2015revocable}.

In this section we generalize the O2H lemma to work with the compressed oracle technique. The oracle register in this technique is a superposition over databases of input-output pairs. A relation on a database is a specific set of databases that fulfill some requirement, e.g., contains a collision (two entries with distinct inputs and the same output). The O2H lemma, as stated in \cite{ambainis2018quantum}, works with punctured oracles, these are quantum oracles that include a binary measurement after every query. After introducing the notion of relations on databases we bring the concept of punctured oracles to the compressed oracles technique. Punctured compressed oracles involve measurements on the superposition of databases. These measurements allow to analyze adversaries that had access to oracles that e.g.\ never output colliding outputs, this is a very useful situation, considering how often we lazy-sample functions is cryptographic proofs and then want to focus on some transcripts of input-output pairs. Our version of the O2H lemma provides a bound on the distinguishing advantage between an oracles that is not punctured and an oracle that is. The bound in the O2H lemma is stated in terms of the probability of any measurement in the punctured oracle succeeding, i.e., finding a database in the oracle register that fulfills the relation we discuss. The strength of our result lies in how versatile the new O2H lemma is, moreover the proof of the lemma is almost the same as the one in \cite{ambainis2018quantum}.

In the original statement of the O2H lemma, the main idea is that there is a marked subset of inputs to the random oracle $\Ho\symbolindexmark{\Ho}$, and an adversary tries to distinguish the situation in which she interacts with the normal oracle from an interaction with an oracle $\G$ that differs only on this set. The lemma states a bound for the distinguishing advantage which depends on the probability of an external algorithm measuring the input register of the adversary and seeing an element of the marked set. This probability is usually small, for random marked sets.

Recently this technique was generalized by Ambainis, Hamburg, and Unruh in \cite{ambainis2018quantum}. The main technical idea introduced by the generalized O2H lemma is to exchange the oracle $\G$ with a so-called \emph{punctured oracle} that measures the input of the adversary after every query. The bound on the adversary's advantage is given by the probability of this measurement succeeding. This technique forms the link with the classical identical-until-bad games: we perform a binary measurement on the ``bad'' event and bound the advantage by the probability of observing this bad event.

In this work we present a generalization of this lemma that involves the use of compressed oracles. Our idea is to measure the database of the compressed oracle, which makes the lemma more versatile and easier to use for more general quantum oracles.

Below we state our generalized O2H lemmas. Most proofs of \cite{ambainis2018quantum} apply almost word by word so we just describe the differences and refer the reader to the original work.

\subsection{Relations on databases}
The key concept we use are relations on the database of the compressed oracle.
\begin{defi}[Classical relation $R$ on $D$]\label{def:class-relation}
Let $D$ be a database of size at most $q$ pairs $(x,y)\in\mathcal{X}\times\mathcal{Y}$. We call a subset\footnote{Note that $ [q+1]=\{0,1,\dots, q\} $} $R\subseteq \bigcup_{t\in[q+1]} \left(\mathcal{X}\times\mathcal{Y}\right)^t$ a \emph{classical relation $R$ on $D$}.
\end{defi}
An example of such a relation is a collision, namely
\begin{align}\label{eq:def-rel-coll}
	R_{\textnormal{coll}}:=\{ ((x_1,y_1),\cdots,(x_s,y_s))\in \bigcup_{t\in[q+1]}\left(\mathcal{X}\times\mathcal{Y}\right)^t: \exists_{i,j} \; i\neq j, x_i\neq x_j, y_i=y_j \}.
\end{align}
Note however, that it is only reasonable to check if the \emph{non-padding entries} are in $ R $, omitting the $ (\perp,0) $ pairs at the end of $ D $.
If $D$ is held in a quantum register, the classical relation $R$ has a corresponding projective measurement $\Juni_R\symbolindexmark{\Juni}$ such that $\norm{\Juni_R \ket{(x_1,y_1),\cdots,(x_q,y_q)}_D}=1$ if and only if for some $s$ it holds that $\big((x_1,y_1),\cdots,(x_s,y_s)\big)\in R$ and for the remaining $i>s$, the $(x_i,y_i)$ are padding entries.

We also state an explicit algorithm to implement the measurement of a relation $R$, given that membership in $R$ is efficiently decidable. To denote the single-bit membership decision by $ D\in R $, the bit is $ 1 $ if and only if database $ D $ is in $ R $. To measure the relation we define a unitary $\V^{SDJ}_R\symbolindexmark{\V}$ that XORs a bit $ D\in R $ to register $ J $; This unitary is controlled on registers $ S $ and $ D $, the former holds the information about the size of the database and the latter the database itself.
Alg.~\ref{alg:measureR} defines the measurement procedure of measuring $ R $ on quantum databases in the standard basis.
\begin{algorithm}
	\caption{Measurement of a relation $ R $}\label{alg:measureR}
	\DontPrintSemicolon
	\SetKwInOut{Input}{Input}\SetKwInOut{Output}{Output}
	\Input{Database $\ket{D}_{D}$ in the standard basis}
	\Output{Outcome $p $ and post-measurement state $ \ket{D'}_D $}
	\BlankLine
	Count in register $ S $ the number of non-padding ($D^X \neq \perp$) entries $ s $ in $ D $\;\label{line:measure-queries1}
	Initialize a new qubit register $ \ket{0}_J $\;
	Apply $\V^{SJD}_R$ that XORs a bit $ j:= D\in R$ to register $ J $\;
	Uncompute register $ S $, measure register $ J $, output the outcome $ j $\;\label{line:measure-relation}
\end{algorithm} 

An important issue concerning measuring relations is the basis in which we store the quantum database. For the measurement to be meaningful it has to be done in the standard basis, so it is easiest to analyze $ \CStO_{\distrD} $ or $ \CPhO_{\distrD} $, defined by Eq.~\eqref{eq:comp-chainQFT}.

While not directly relevant to our applications, we keep the generality of \cite{ambainis2018quantum} by introducing the notion of \emph{query depth} as the number of sets of parallel queries an algorithm makes. We usually assume quantum algorithms make $q$ quantum queries in total and $d$ (as in ``query depth'') sequentially, but those queries in sequence may involve a number of parallel queries. A parallel query of width $p$ to an oracle $\Ho$ involves $p$ applications of $\Ho$ to $p$ query registers. Note that if $\Ho$ is considered to be a compressed oracle, $p$-parallel queries are processed by sequentially applying the compressed oracle unitary $p$ times.

First we define a compressed oracle $\Ho$ punctured on relation $R$, denoted by $\Ho\setminus R$.
\begin{defi}[Punctured compressed oracle $\Ho\setminus R$]\label{def:punctured}
Let $\Ho$ be a compressed oracle and $R$ a relation on its database. The punctured compressed oracle $\Ho\setminus R$ is equal to $\Ho$, except that $R$ is measured after every query as described in Alg.~\ref{alg:measureR}. By $\Find\symbolindexmark{\Find}$ we denote the event that $R$ outputs $1$ at least once among all queries.
\end{defi}
Full oracles can be punctured as well, the relation is then checked only on the queried entries of the function table---those queried entries need to be identified (like in $ \Dec_{\distrD} $ from Alg.~\ref{alg:general-dec}) prior to the measurement of $ R $.

In many applications of punctured oracles we might want to apply $ \Ho \setminus R $ only if some condition is fulfilled. Moreover, this condition might be quantum---in other words we control $ \Ho \setminus R $ on some quantum register. To avoid the situation of a measurement being performed or not depending on a state of a quantum register---which is not permitted by quantum mechanics---we propose the following solution: We postpone the measurement to the end of the quantum query. Namely, we omit the measurement of register $ J $ in Alg.~\ref{alg:measureR} and perform it at the end of the compressed-oracle algorithm. 
After the measurement we can uncompute the outcome register $ J $. We are not changing notation and implicitly assume the postponement of puncturing---e.g. in Alg.~\ref{alg:sponge-quant-sim-tra}.

\subsection{One-way to Hiding Lemma}
Using the definitions from the previous sections we can prove a theorem similar to Theorem~1 of \cite{ambainis2018quantum}.

Let us also comment on the differences of the O2H lemma in \cite{ambainis2018quantum} and our paper. The main difference is that in our generalization we no longer focus solely (we can recover the original O2H lemma though) on the adversary's inputs but also treat the outputs of the oracle. Function outputs are also important in \cite{ambainis2018quantum}, but the oracle is not lazy sampled, there they pick a subset of the domain such that e.g. the output is 0 and then puncture on inputs in this random set. We use lazy sampled functions and puncture on databases, so functions defined only on the queried inputs. In addition, defining the puncturing operation on the compressed oracle-database is more expressive, as it allows puncturing conditions depending on more than one input-output pair.
\begin{thm}[Compressed oracle O2H]\label{thm:comp-o2h} 
Let $R_1$ and $ R_2 $ be relations on the database of a quantum oracle $\Ho$.
Let $z$ be a random string. $R$ and $z$ may have arbitrary joint distribution.
Let $\advA\symbolindexmark{\advA}$ be an oracle algorithm of query depth $d$, then
\begin{align}
& \left| \PR[b=1:b\gets \advA^{\Ho\setminus R_1}(z)]- \PR[b=1:b\gets \advA^{\Ho\setminus R_1\cup R_2}(z)] \right| \nonumber\\
&\leq \sqrt{(d+1) \PR[\Find_2: \advA^{\Ho\setminus R_1\cup R_2}(z)]}, \\
&\left| \sqrt{\PR[b=1:b\gets \advA^{\Ho\setminus R_1}(z)]}- \sqrt{\PR[b=1:b\gets \advA^{\Ho\setminus R_1\cup R_2}(z)} \right|  \nonumber\\
&\leq \sqrt{(d+1) \PR[\Find_2: \advA^{\Ho\setminus R_1\cup R_2}(z)]},
\end{align}
where $ \Find_2 $ is the event that measuring $ R_2 $ succeeds.
\end{thm}

\begin{proof}[Proof sketch]
The proof works almost the same as the proof of Theorem~1 of \cite{ambainis2018quantum}. Instead of checking register $ X $ for the success of the puncturing measurement we analyze the oracle register. The rest follows exactly the same reasoning. All the details of the full proof can be found in Appendix~\ref{sec:full-o2h}
\end{proof}

We continue by deriving an explicit formula for $\PR[\Find]$. Let $\advA$ be a quantum algorithm with oracle access to $\Ho$, making at most $q$ quantum queries with depth $d$. Let $R$ be a relation on the database of $\Ho$ and $z$ an input to $\advA$. $R$ and $z$ can have any joint distribution. $\Juni_R$ is the projector from the measurement of $R$ on $D$, $\Uni^{\Ho}_i$ is the $i$-th unitary performed by $\advA^{\Ho\setminus R}$ together with a query to $\Ho$, and $\ket{\Psi_0}$ is the initial state of $\advA$. Then we have the formula
\begin{align}
\PR[\Find: \advA^{\Ho\setminus R}(z)] = 1- \norm{\prod_{i=1}^{d} (\mathbbm{1}-\Juni_R) \Uni^{\Ho}_i \ket{\Psi_0} }^2. \label{eq:def-find}
\end{align}$\symbolindexmark{\Find}$

Let us now discuss the notion of ``identical-until-bad'' games in the case of compressed oracles. For random oracles, the notion was introduced in \cite{ambainis2018quantum}. The definition is rather straightforward as $\Ho$ and $\G$ are considered identical until bad if they had the same outputs except for some marked set. When using compressed oracles, the outputs of $\Ho$ and $\G$ are quantum lazy-sampled, making the definition of what it means for two oracles to be identical until bad require more care.
Here we state a definition that captures useful notions of identical-until-bad punctured oracles.
\begin{defi}[Almost identical oracles]\label{def:almost-id-or1}
Let $\Ho$ and $\G$ be compressed oracles and $R_i$, $i=1,2$ relations on their databases. We call the oracles $\Ho\setminus R_1$ and $\G\setminus R_2$ \textnormal{almost identical} if they are equal conditioned on the events $\neg \Find_1$ and $ \neg\Find_2 $ respectively, i.e. for any string $z$ and any quantum algorithm $\advA$ 
\begin{align}
 \PR[b=1:b\gets A^{\Ho\setminus R_1}(z)\mid \neg \Find_1] = \PR[b=1:b\gets A^{\G\setminus R_2}(z)\mid \neg \Find_2].	
\end{align}
\end{defi}
Note that not punctured compressed oracles are a special case of punctured ones (for $R=\emptyset$), so the above definition can be applied to a pair of oracles where one is punctured and one is not.
We can prove the following bound on the adversary's advantage in distinguishing almost identical punctured oracles.
\begin{lemm}[Distinguishing almost identical punctured oracles]\label{lem:almost-id-gms-dist1}
If $\Ho\setminus R_1$ and $\G\setminus R_2$ are almost identical according to Def.\ref{def:almost-id-or1} then for any $ b\in\{0,1\} $
\begin{align}
& \abs{\PR[b\gets A^{\Ho\setminus R_1}(z)]-\PR[b\gets A^{\G\setminus R_2}(z)]} \leq  2 \PR[\Find_1:A^{\Ho\setminus R_1}(z)] +  2 \PR[\Find_2:A^{\G\setminus R_2}(z)].
\end{align}
\end{lemm}
\begin{proof}
	We bound
\begin{align}
 &\abs{\PR[b\gets A^{\Ho\setminus R_1}(z)]-\PR[b\gets A^{\G\setminus R_2}(z)]} \nonumber \\
  & \overset{\textnormal{Def.~\ref{def:almost-id-or1}}}{=} \left\lvert{\PR[b\gets A^{\Ho\setminus R_1}(z)\mid \neg\Find_1]}{\left( \PR[\neg\Find_1:A^{\Ho\setminus R_1}(z)]-\PR[\neg\Find_2:A^{\G\setminus R_2}(z)]\right) } \right.\nonumber\\
& + \PR[b\gets A^{\Ho\setminus R_1}(z)\mid \Find_1] \PR[\Find_1:A^{\Ho\setminus R_1}(z)]\nonumber \\
& \left.- \PR[b\gets A^{\G\setminus R_2}(z)\mid \Find_2] \PR[\Find_2:A^{\G\setminus R_2}(z)]   \right\rvert \\
&\overset{\triangle}{\leq}\left\lvert \underset{\leq 1}{\underbrace{\PR[b\gets A^{\Ho\setminus R_1}(z)\mid \neg\Find_1]}} \underset{=\PR[\Find_2:A^{\G\setminus R_2}(z)]-\PR[\Find_1:A^{\Ho\setminus R_1}(z)]}{\underbrace{\left( \PR[\neg\Find_1:A^{\Ho\setminus R_1}(z)]-\PR[\neg\Find_2:A^{\G\setminus R_2}(z)]\right) }} \right|\nonumber \\
& +\left|\underset{\leq 1}{\underbrace{\PR[b\gets A^{\Ho\setminus R_1}(z)\mid \Find_1]}}\PR[\Find_1:A^{\Ho\setminus R_1}(z)]\right|\nonumber \\
&+\left|\underset{\leq 1}{\underbrace{\PR[b\gets A^{\G\setminus R_2}(z)\mid \Find_2]}}\PR[\Find_2:A^{\G\setminus R_2}(z)]   \right\rvert \\
& \overset{\triangle}{\leq} 2 \PR[\Find_1:A^{\Ho\setminus R_1}(z)] + 2 \PR[\Find_2:A^{\G\setminus R_2}(z)],
\end{align}
where by $\triangle$ we denote the triangle inequality.
\end{proof}

Note that for $R_2=\emptyset$, the above lemma is essentially a special case of the well known Gentle-Measurement Lemma of \cite{winter1999coding}. 

It is a fact of quantum mechanics that measurements disturb the state. Considering that, one might be curious if measuring the database does not disturb it too much. As an example, note that after a measurement of the collision relation, eq.~\eqref{eq:def-rel-coll}, the database does not necessarily consist of only non-Fourier-$ 0 $ entries. Even though this is true, if the disturbance of the oracle is low enough, then the adversary will not notice it. This is exactly the case of the O2H lemma, the disturbance is low enough so the adversary does not notice any  difference in the content of the oracle's output.

\subsection{Calculating $ \Find $ for the Collision and Preimage Relations}\label{sec:pr-find}
We state a lemma giving a bound on the probability of $ \Find $ for the uniform distribution over the set $ \{f:\mathcal{X}\to \mathcal{Y}\} $, and for the union of the collision and preimage relations. The preimage relation is satisfied when the output of the oracle is $ 0 $:
\begin{align}\label{eq:def-rel-preim}
	R_{\preim}:=\{ ((x_1,y_1),\cdots,(x_t,y_t))\in \bigcup_{s\in[q+1]}\left(\mathcal{X}\times\mathcal{Y}\right)^s: \exists i: y_i=0 \}.
\end{align}

In the following we assume $ \mathcal{Y}=[N] $.
\begin{lemm}\label{lem:prfind-preim-coll}
	For any quantum adversary $ \advA $ interacting with a punctured oracle $ \CStO_{\mathcal{Y}}\setminus (R_{\preim}\cup R_{\coll}) $---where $ R_{\coll} $ is defined in Eq.~\eqref{eq:def-rel-coll} and $ R_{\preim} $ in Eq.~\eqref{eq:def-rel-preim}---the probability of $ \Find $ is bounded by:
	\begin{align}
		\PR[\Find:\advA[\CStO_{\mathcal{Y}}\setminus (R_{\preim}\cup R_{\coll}) ]]\leq 
		296\frac{q^2}{N},
	\end{align}
	where $ q $ is the maximal number of queries made by $ \advA $ and $ N=\abs{\mathcal{Y}} $.
\end{lemm}
\begin{proof}
	The proof of this lemma is a fairly simple application of the techniques from \cite{CFHL21}.  Let $\Pi_R$ be the projector onto the subspace of the database register on which $R$ holds, and let $\Pi_{\le \ell}$ be the projector onto the databases of size at most $\ell$. We can express the $ \Find $ probability as
	\begin{align}
		&\PR[\Find:\advA[\CStO_{\mathcal{Y}}\setminus (R_{\preim}\cup R_{\coll}) ]]\nonumber\\
		&=\sum_i\PR[\Find\text{ at query }i:\advA[\CStO_{\mathcal{Y}}\setminus (R_{\preim}\cup R_{\coll}) ]]\\
		&=\sum_{i=1}^q \left\|\Pi_R\CStO_{\mathcal Y}\Uni_i\left(\prod_{j=i-1}^1(\id-\Pi_R)\CPhO_{\mathcal Y}\Uni_j\right)\ket{\Psi_0}\right\|^2\\
		&=\sum_{i=1}^q \left\|\Pi_R\CStO_{\mathcal Y}(\id-\Pi_R)\Pi_{\le i-1}\Uni_i\left(\prod_{j=i-1}^1(\id-\Pi_R)\CStO_{\mathcal Y}\Uni_j\right)\ket{\Psi_0}\right\|^2\\
		&\le\sum_{i=1}^q \left\|\Pi_R\CStO_{\mathcal Y}(\id-\Pi_R)\Pi_{\le i-1}\right\|_\infty^2\label{eq:intermediate},
	\end{align}
where $ \ket{\Psi_0} $ is the initial state of the adversary and the oracle and $ \Uni_i $ are the adversary's unitaries.

We can now apply Corollary 5.26 in \cite{CFHL21}, where $\mathsf{P'}$ are the databases of size at most $i-1$ not contained in $(R_{\preim}\cup R_{\coll})$ and $\mathsf{P}$ are the databases in $(R_{\preim}\cup R_{\coll})$, to obtain    
\begin{equation}\label{eq:CFHL}
	\left\|\Pi_R\CStO_{\mathcal Y}(\id-\Pi_R)\Pi_{\le i-1}\right\|_\infty\le 2e\sqrt{10\left(\frac{1}{N} + \frac{i-1}{N}\right)} = 2e\sqrt{10\left(\frac{i}{N}\right)},
\end{equation}
where $e$ is Euler's number. Here we have used $\frac{1}{N}$ as an upper bound on the probability that the $y$-value of the new database entry is $0$ (resulting in membership of $R_{\preim}$), whereas $\frac{i-1}{N}$ is an upper bound on the probability that this $y$-value matches one of the $i-1$ $y$-values already contained in the database (resulting in membership of $R_{\coll}$).

We can hence continue to bound
	\begin{align}
		&\PR[\Find:\advA[\CStO_{\mathcal{Y}}\setminus (R_{\preim}\cup R_{\coll}) ]]\nonumber\\
		&\leq \sum_{i=1}^q \left\|\Pi_R\CStO_{\mathcal Y}(\id-\Pi_R)\Pi_{\le i-1}\right\|_\infty^2
		\le\frac{296}{N}\sum_{i=1}^q i =\frac{296 q(q+1)}{2N}\le 296\frac{q^2}{N},
	\end{align}
\end{proof}
where in the final inequality we have assumed $q\ge 1$ as for $q=0$ there is nothing to prove.

In Appendix~\ref{sec:full-proof} we give a direct proof with a better constant at the expense of some lower-order terms.
The above bound 
is just the classical collision (and preimage) finding bound up to a constant factor. Intuitively, this is because the coherence needed by the optimal quantum search algorithms (e.g.\ the Grover algorithm \cite{grover1996fast}) is broken by the repeated measurement.

Finally let us provide a clearer explanation for how to use our technique. Whenever we lazy sample a uniform function in the (classical) game-playing framework we have some bad events, for example the newest output collides with some previous one. To translate the proof to the quantum case we reformulate the bad events to the language of relations and use a punctured compressed oracle. Hybrid jumps are bounded with the O2H lemma and $ \PR[\Find] $ with (a version of) Lemma~\ref{lem:prfind-preim-coll}. Note that only this technique allows us to deal with collisions in quantumly lazy sampled functions. The only other paper that considers this problem is \cite{zhandry2018record} but there are some things that are a bit unclear in the proof of the important lemma there.

\section{Quantum Security of the Sponge Construction}\label{sec:q-indifferentiability}
We use our methods to show a detailed proof of quantum indifferentiability of the sponge construction when used with a random function as the internal function. 
In Appendix~\ref{sec:collapsingness} we prove that quantum indifferentiability implies collapsingness. 

After introducing the sponge construction in the next section, we present two proofs of indifferentiability of the sponge construction. The first proves classical security and the second quantum security. We present two proofs to simplify reading the quantum proof, it follows the same reasoning as the classical one. We also want to highlight how similar these proofs are, this is what we consider to be one of the main advantages of our quantum game-playing framework. In our framework all proofs of quantum indifferentiability can follow the same reasoning and very similar steps as the classical version.

Before we proceed let us remind the reader of the main concepts, that are necessary to follow the proof of quantum indifferentiability.
The central object of the proof are punctured oracles, defined in Def.~\ref{def:punctured}. They play the role of subroutines that lazy-sample functions and output ``True'' when a bad event occurs. Readers familiar with the original game-playing framework \cite{bellare2004code} will recognize the crucial subroutines of the classical games. 
Additionally, punctured oracles are objects that allow to condition probabilistic events on some aspects of quantum queries done by the adversary. This useful feature allows us to sometimes use arguments from the classical proof in the quantum one.

A punctured oracle is built using the compressed-oracle framework and formally includes a quantum database register, as described in detail in section~\ref{sec:oracles}. Nonetheless these details are not necessary to follow the contents of this section. The only two things to keep in mind are that in general the adversary can make quantum queries to the primitives and that the responses of queries are saved in the adversary's quantum register $ \ket{s,v} $, where $ s $ is the query and $ v $ is any value in the codomain of the queried function.

The reason we use punctured oracles is that they allow to use the One-way To Hiding (O2H) lemma. This is an extremely useful tool for bounding the distinguishability advantage of two quantum games. We cover this lemma in details in section~\ref{sec:comp-o2h}. Technically the most demanding part of using the O2H lemma is bounding the probability of any puncturing measurement succeeding (we call this event $ \Find $). We compute a bound on $ \pr{\Find} $ useful in the quantum indifferentiability proof for sponges in section~\ref{sec:pr-find}.

The second distinguishability bound that we use is shown in Lemma~\ref{lem:almost-id-gms-dist1}. This is a relatively simple statement, that is true for games that are almost identical (Def.\ref{def:almost-id-or1}).

\subsection{Sponge Construction}
The sponge construction is used to design variable-input-length and variable-output-length functions. It works by applying the \emph{internal function} $ \phif $ multiple times on the \emph{state} of the function.
In Algorithm~\ref{alg:spongegen} we present the definition of the sponge construction, which we denote with $\sponge$~\cite{bertoni2007sponge}.
The internal state\footnote{Our result also holds for arbitrary finite sets $ \seta\times\setc $, where additionally $ \seta $ is an Abelian group.} $s = (\bar{s}, \hat{s}) \in\bits{r}\times\bits{c}$ of $ \sponge $ consists of two parts: the \emph{outer} part $\bar{s} \in \bits{r}$ and the \emph{inner part} $\hat{s} \in \bits{c}$. The logarithm of the number of possible outer parts $r$ is called the \emph{rate} of the sponge, and $c$ is called  \emph{capacity}. Naturally the internal function is a map $\phif:\bits{r}\times\bits{c}\to\bits{r}\times\bits{c}\symbolindexmark{\phif}$. To denote the internal function with output limited to the first $ r $ bits and the last $ c $ bits we use the same notation as for states, $ \phifbar\symbolindexmark{\phifbar} $ and $ \phifhat\symbolindexmark{\phifhat} $ respectively. By $ \left(\bits{r}\right)^* $ we denote the strings consisting of an arbitrary number of $ r $-bit blocks.
By $\pad:\bits{*}\to\left(\bits{r}\right)^* \symbolindexmark{\pad}$ we denote a padding function: an efficiently computable injection such that $ \abs{\pad(m)}\geq r $ and that the last bit of $ \pad(m) $ is never $ 0 $ (this ensures injectivity for inputs of different lengths). By $\abs{p}_r$ we denote the number of $ r $-bit blocks in $p$ and by $ \lfloor p\rfloor_i $ we denote the $ i $-th $ r $-bit block of $ p $.
The function constructed in that way behaves as follows, $ \sponge_{\phif}:\bits{*}\times\mathbb{N}\to\bits{*} $, where $ \bits{*}:=\bigcup_{n=0}^{\infty}\bits{n} $.
\iffullversion
In Fig.~\ref{fig:sponge} we present a scheme of the sponge construction evaluated on input $m$.
\begin{figure}[hbt]
\begin{center}
\includegraphics[scale=0.6]{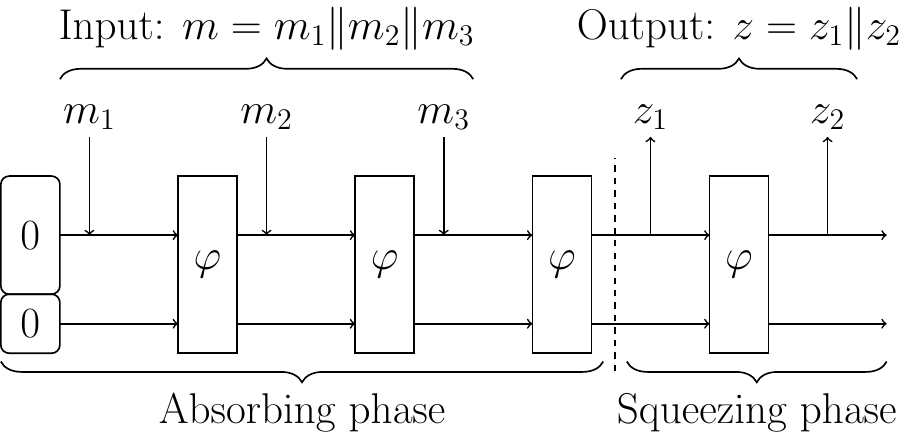}
\end{center}
\caption{A schematic representation of the sponge construction: $\sponge_{\phif}(m_1\|m_2\|m_3)=z_1\|z_2$.\label{fig:sponge}}
\end{figure}
\fi

For a set $ \mathcal{S}\subseteq \bits{r+c} $, by $ \overline{\mathcal{S}} $ we denote the outer part of the set: a set of outer parts of elements of $ \mathcal{S} $. Similarly by $ \widehat{\mathcal{S}} $ we denote the inner part of the set. We use similar notation for quantum registers holding a quantum state in $ \Hil_{\bits{r+c}} $: $ \overline{Y} $ is the part of the register holding the outer parts and $ \widehat{Y} $ holds the inner parts.

\begin{algorithm}[hbt]
\DontPrintSemicolon
\SetKwInOut{Input}{Input}\SetKwInOut{Output}{Output}
\Input{$m \in \bits{*}$, $\ell \geq 0$.}
\Output{$z\in \bits{\ell}$}
\BlankLine
$p:=\pad(m)$\;
$s:=(0,0)\in\bits{r}\times\bits{c}$.\;
\For(\tcp*[f]{Absorbing phase}){$i=1$ \textnormal{to} $\abs{p}_r$}{
$s:= (\bar{s} \xor \lfloor p\rfloor_i ,\hat{s})$\;
$s :=\phif (s)$\;
}
$z:=\bar{s}$\tcp*[r]{Squeezing phase}
\While{$|z| < \ell$}{
$s:= \phif(s)$\;
$z:=z\lVert \bar{s}$\;
}
Output $z$\;
\caption{$\sponge_{\phif}[\pad,r,c]\symbolindexmark{\sponge}$}\label{alg:spongegen}
\end{algorithm}

An important feature of the sponge construction is that one can associate to the internal function $ \phif $ a graph $G=(\vertices,\edges)$ \cite{bertoni2007sponge}. It is called the \emph{sponge graph}; The set of nodes $\vertices:=\bits{r+c}\symbolindexmark{\vertices}$ corresponds to all possible states of the sponge. A directed edge connects any two nodes $ (s,t) $ whenever $ \phif(s)=t $, hence there are $ 2^{r+c} $ edges in $ \edges\symbolindexmark{\edges} $. From each node starts exactly one edge.
We group the nodes with the same inner-part value into \emph{supernodes}, so that we have $2^c$ supernodes and each such supernode consists of $2^r$ nodes. Edges between nodes are also edges between supernodes.

Whenever the adversary queries $ \sponge $, she starts at the $ (0^r,0^c) $ node. This node is called the \emph{root}. Next the first $ r $-bit block $ \lfloor p\rfloor_1 $ in the padded message $ p=\pad(m) $ is added to the outer part of the state and queried to the internal function $ \phif(\lfloor p\rfloor_1,0^c)=s_2 $. The node $ s_2 $ is the node in the edge $ ((\lfloor p\rfloor_1,0^c),s_2) \in\edges$. The same situation repeats for all blocks in $ p $, during the absorbing phase. When $ \sponge $ starts generating output, we no longer modify the state, or just add $ 0^r $ to the outer part.
Note that knowing just $ p $ and $ G $ we can get to the last node traversed by $ \sponge_{\phif}(m) $. This leads us to the definition of a \emph{sponge path}.
\begin{defi}[Sponge Path, Definition~3 in \cite{bertoni2008indifferentiability}]
	First, the empty string is a sponge path to the node $ (0^r,0^c) $. Then, if $ p $ is  a sponge path to node $ s=(\bar{s},\hat{s}) $ and there is an edge $ (\bar{s}\xor a\|\hat{s},t) $ in the sponge graph $ G $, $ p'=p\| a $ is a sponge path to node $ t $.
\end{defi}

Given the above definition, let us say that if $ p $ is a sponge path to $ s $, then we define a function $ \symbolindexmark{\funPath} $
\begin{align}
	\funPath(s,G):=p. \label{eq:def-funPath}
\end{align}
The output of the above function is the input to the construction $ \sponge_{\phif}(.,\ell=r) $ that yields the output $ \bar{s} $.

When we talk about the simulator in a proof of indifferentiability, we define the \emph{simulator graph}. The graph kept by the simulator differs from the sponge graph discussed above by the number of edges in it. As the simulator lazy samples the internal function $ \phif $ the set of edges $ \edges $ grows by at most one edge per one adversary's query. Other than that, all definitions above hold for the simulator graph as well. We refer to the simulator graph $ G $ as just the (sponge) graph whenever it is clear from context.

A supernode is called \emph{rooted} if there is a path (a regular path that is just a set of edges connected by the end-start nodes) leading to it that starts in the root (the $0$-supernode). The set $\setr\symbolindexmark{\setr}$ is the set of all rooted supernodes in $ G $.
By $\sete\symbolindexmark{\sete}$ we denote the set of supernodes with a node with an outgoing edge.

A simulator graph is called \emph{saturated} if $ \setr\cup\sete=\bits{c} $. It means that for every inner state in $\bits{c}$ there is an edge in $G$ that leads to it from $ 0^c $ (the root) or leads from it to another node. Saturation will be important in the proof of indifferentiability as the simulator wants to pick outputs of $ \phif $ without colliding inner parts (so not in $ \setr $) and making the path leading from $ 0^c $ to the output longer by just one edge (so not in $ \sete $).

The simulators defined in the proofs in this section are implicitly stateful. They maintain a classical or quantum state containing a database of the adversary's queries and the simulator's outputs. Using  that database, the simulator can always construct a sponge graph containing all the current knowledge of $ \phif $.

For the proof of indifferentiability we also need an upper bound on the probability of finding a collision in the inner part of outputs of a uniformly random function $ f:\bits{r+c}\to\bits{r+c} $. Considering how $ \sponge $ is defined we want a bound on finding collisions and zero-preimages. 
We define the bound as a function of the number of queries $ q $ to $ f $:
\begin{align} \label{eq:fcoll}
	f_{\coll}(q):= \frac{q(q+1)}{2^{c+1}},
\end{align}
the bound can be derived in the standard way. The probability that any classical algorithm finds a collision or a preimage of zero in $ [N] $ after $ q $ queries is:
\begin{align}
	\pr{\coll\cup\preim\gets\advA}\leq \sum_{i=1}^{q}\frac{i}{N} =\frac{q(q+1)}{2N},
\end{align}
where we use the union bound and note that after $ i $ queries the adversary can either find the preimeage of zero or hit any of the previous outputs, producing a collision. For a more detailed derivation we refer to Appendix~A.4 in \cite{katz2014introduction}.

As the sponge construction is used to design variable-input and variable-output functions we define the random oracle
\begin{align}\label{eq:random-oracle}
	\Ho:\bits{*}\times\mathbb{N}\to\bits{*}
\end{align}
 accordingly. A random oracle grants access to a function sampled from distribution $ \mathfrak{R} $ on functions $ \bits{*}\times\mathbb{N}\to\bits{*} $, that is defined as follows: To sample a function $h\gets \mathfrak{R}$ we
\begin{itemize}
	\item choose $g$ uniformly at random from $\{g:\bits{*}\to\bits{\infty} \}$, where by $\bits{\infty}$ we denote the set of infinitely long bitstrings,
	\item for each $(x,\ell)\in\bits{*}\times\mathbb{N}$ set $h(x,\ell):=\left\lfloor g(x) \right\rfloor_{\ell}$, that is, output the first $\ell$ bits of the output of $g$.
\end{itemize}
In the following section, we omit the second input and we mean that we ask for a single letter $ \Ho(x)=y\in\bits{r} $.

\subsection{Classical Indifferentiability of Sponges with Random Functions}
In the game-playing proofs and Algorithms~\ref{alg:sponge-clas-sim-tra} and~\ref{alg:sponge-quant-sim-tra} described in this section we use the following convention: every version of the algorithm executes the part of the code that is \textbf{not boxed} and among the boxed statements only the part that is inside the box in the color corresponding to the color of the name in the definition.

First we present a slightly modified proof of indifferentiability from \cite{bertoni2008indifferentiability}. We modify the proof to better fit the framework of game-playing proofs. It is not our goal to obtain the tightest bounds nor the simplest (classical) proof. Instead, our classical game-playing proof paves the way to the quantum security proof which is presented in the next section.
\begin{thm}[$ \sponge $ with functions, classical indifferentiability]\label{thm:classindiff-trans}
	$\sponge_{\phif}[\pad,r,c]$ calling a random function $ \phif $ is $(q,\eps)$-indifferentiable from a random oracle, Eq.~\eqref{eq:random-oracle}, for \textnormal{classical} adversaries for any $q<2^c$ and $\eps =8\frac{q(q+1)}{2^{c+1}}$.
\end{thm}

\begin{proof}
The proof proceeds in six games that we show to be indistinguishable.
We start with the real world: the public interface corresponding to the internal function $\phif$ is a random transformation and the private interface is $\sponge_{\phif}$.
Then in a series of games we gradually change the environment of the adversary to finally reach the ideal world, where the public interface is simulated by the simulator and the private interface is a random oracle $\Ho$.
The simulators used in different games of the proof are defined in Alg.~\ref{alg:sponge-clas-sim-tra}, the index of the simulator corresponds to the game in which the simulator is used. Explanations of the simulators follow.
\begin{algorithm}
	\caption{Classical $\Sim_2,\cboxed{darkgreen}{\Sim_3},\cboxed{red}{\Sim_4},\cboxed{blue}{\mathsf{I}_6}\symbolindexmark{\Sim}$, functions}\label{alg:sponge-clas-sim-tra}
\DontPrintSemicolon
\SetKwInOut{Input}{Input}\SetKwInOut{Output}{Output}
\SetKwInOut{State}{State}
\State{current sponge graph $G$}
\Input{$s\in \bits{r+c}$}
\Output{$\phif(s)$}
\BlankLine
\If(\tcp*[f]{new query}){$s \textnormal{ has no outgoing edge}$\label{step:if1}}{ 
\If(\tcp*[f]{$\hat{s}$-rooted, no saturation}){$\hat{s}\in\setr \wedge \setr \cup\sete \neq \bits{c}$\label{step:if2}}{
$\hat{t}\sampuni \bits{c}$, $\cboxed{darkgreen}{\cboxed{red}{\textnormal{ \textbf{if} $\hat{t}\in\setr\cup\sete$, set $\Bad=1$\symbolindexmark{\Bad}}}}$, $ \cboxed{blue}{\hat{t}\sampuni\bits{c}\setminus(\setr\cup\sete)} $ \;
Construct a path to $s$: $p:=\funPath(s,G)$\;
\If{$\exists x: p=\pad(x)$\label{step:if3}}{
$\bar{t}\sampuni\bits{r}$\;
$\cboxed{red}{\cboxed{blue}{\bar{t}:=\Ho(x)}}$\label{line:Ho}\;
}
\Else{
$\bar{t}\sampuni\bits{r}$\;
}
$t:=\bar{t}\|\hat{t}$\;
}
\Else{
$t\sampuni\bits{r+c}$\;
}
Add an edge $(s,t)$ to $\edges$.\;
}
Set $t$ to the vertex at the end of the edge starting at $s$\;
Output $t$\;
\end{algorithm}

\parag{Game 1}
We start with the real world where the distinguisher $\advA$ has access to a random function $\phif:\bits{r+c}\to\bits{r+c}$ and $\sponge_{\phif}$ using this random function.
The formal definition of the first game is the event
\begin{equation}
\GAME{1}:= \left( b=1:b\gets \advA[\sponge_{\phif},\phif]\right).
\end{equation}

\parag{Game 2}
In the second game we introduce the simulator $\Sim_2$---defined in Alg.~\ref{alg:sponge-clas-sim-tra}---that lazy-samples the random function $\phif$. In Alg.~\ref{alg:sponge-clas-sim-tra} we define all simulators of this proof at once, but note that the behavior of $\Sim_2$ is not influenced by any of the conditional ``if'' statements (in lines~\ref{step:if1}, \ref{step:if2}, and \ref{step:if3}), because in the end, the output state $t$ is picked uniformly from  $\bits{r+c}$ anyway. The definition of the second game is
\begin{equation}
\GAME{2}:= \left( b=1:b\gets \advA[\sponge_{\Sim_2},\Sim_2]\right).
\end{equation}
Because the simulator $\Sim_2$ perfectly models a random function and we use the same function for the private interface we have
\begin{align}
\abs{\PR[\GAME{2}]-\PR[\GAME{1}]}=0.
\end{align}

\parag{Game 3}
In the next step we modify $\Sim_2$ to $\Sim_3$. The game is then
\begin{equation}
\GAME{3}:= \left( b=1:b\gets \advA[\sponge_{\Sim_3},\Sim_3]\right).
\end{equation}
We made a single change in $\Sim_3$ compared to $\Sim_2$, we introduce the ``bad'' event $\Bad$ that marks the difference between algorithms. We use this event as the bad event in Lemma~\ref{lem:clas-game-play}.
With such a change of the simulators we can use Lemma~\ref{lem:clas-game-play} to bound the difference of probabilities:
\begin{align}
\abs{\PR[\GAME{3}]-\PR[\GAME{2}]} \leq \PR[\Bad=1].
\end{align}
It is quite easy to bound $\PR[\Bad=1]$ as it is the probability of finding a collision or preimage of the root in the set $\bits{c}$ having made $q$ random samples. Therefore we have that
\begin{align}
\PR[\Bad=1]\leq  f_{\coll}(q),
\end{align}
where $f_{\coll}$ is defined in Eq.~\eqref{eq:fcoll}. The bound is not necessarily tight as not all queries are made to rooted nodes.

\parag{Game 4}
In this step we introduce the random oracle $\Ho$ but only to generate the outer part of the output of $\phif$. The game is defined as
\begin{equation}
\GAME{4}:= \left( b=1:b\gets \advA[\sponge_{\Sim_4},\Sim^{\Ho}_4]\right).
\end{equation}
We observe that if $\Bad=0$ the outputs are identically distributed.
\begin{nrclaim}\label{claim:claim1}
Given that $\Bad=0$ the mentioned games are the same:
\begin{align}
\abs{\PR[\GAME{4}\mid \Bad=0]-\PR[\GAME{3}\mid \Bad=0]} =0.
\end{align}
\end{nrclaim}
\begin{proof}
Note that the inner part is distributed in the same way in both games if $\Bad=0$, so we only need to take care of the outer part of the output. The problem might lie in the outer part, as we modify the output from a random sample to $\Ho(x)$. If $\Bad=0$ then $\hat{t}$ is not rooted and has no outgoing edge, also the whole graph $G$ does not contain two paths leading to the same supernode. Hence, $x$ was not queried before and is uniformly random. This reasoning is made more formal in Lemma~1 and Lemma~2 of \cite{bertoni2007sponge}.
\end{proof}

The two games are identical-until-bad, this implies that the probability of setting $ \Bad $ to one in both games is the same $ \PR[\Bad=1:\GAME{3}] =\PR[\Bad=1:\GAME{4}]$. Together with the above claim we can derive the advantage:
\begin{align}
&\abs{\PR[\GAME{4}]-\PR[\GAME{3}]}  \overset{\textnormal{Claim~\ref{claim:claim1}}}{=} \Bigg\vert \PR[\GAME{4}\mid \Bad=0]\nonumber\\
&\cdot \underset{=0}{\underbrace{\left( \PR[\Bad=1:\GAME{3}] -\PR[\Bad=1:\GAME{4}]]\right) }} \nonumber\\
&+ \underset{\leq 1}{\underbrace{\PR[\GAME{3}\mid \Bad=1]}}\PR[\Bad=1]  +\underset{\leq 1}{\underbrace{\PR[\GAME{4}\mid \Bad=1]}}\PR[\Bad=1] \Bigg\vert\\
&\leq 2\PR[\Bad=1].\label{eq:class-adv-gam43}
\end{align}

\parag{Game 5}
In this stage of the proof we change the private interface to contain the actual random oracle. The simulator is the same as before and the game is
\begin{equation}
\GAME{5}:= \left( b=1:b\gets \advA[\Ho,\Sim^{\Ho}_4]\right).
\end{equation}
Conditioned on $ \Bad=0 $, the outputs of the simulator in Games 4 and 5 act in the same way and are consistent with $\Ho$. 

Note that the inner states are generated by the same pseudocode and the outer states are distributed in the same way.
Moreover conditioned on $ \Bad=0 $ the probabilities of $ \advA $ outputting $ 1 $ are the same.
To calculate the adversary's advantage in distinguishing between the two games we can follow the proof of Lemma~\ref{lem:almost-id-gms-dist1}, with $ \Ho\setminus R_1 $ replaced by $ \GAME{5} $, $ \G\setminus R_2 $ replaced by $ \GAME{4} $, and event $ \Find $ replaced by $ \Bad=1 $. As the derivation of Lemma~\ref{lem:almost-id-gms-dist1} uses no quantum mechanical arguments and the assumption holds---the games are identical conditioned on $ \Bad=0 $---the bound holds:
\begin{align}
\abs{\PR[\GAME{5}]-\PR[\GAME{4}]}\leq 4 \PR[\Bad=1]\leq 4 f_{\coll}(q).
\end{align}

\parag{Game 6}
In the last game we use $\mathsf{I}_6$ (we call it  $ \mathsf{I} $ for \emph{ideal}, that is the world we reach in the last step of the proof), a simulator that does not check for bad events and samples from the ``good'' subset of $ \bits{c} $. The game is 
\begin{equation}
\GAME{6}:= \left( b=1:b\gets \advA[\Ho,\mathsf{I}^{\Ho}_6]\right)
\end{equation}
and the advantage is  
\begin{align}
\abs{\PR[\GAME{6}]-\PR[\GAME{5}]} \leq \PR[\Bad=1]\leq  f_{\coll}(q).
\end{align}
following Lemma~\ref{lem:clas-game-play}.
as the only difference is in code but not outputs.
We included this last game in the proof because $\mathsf{I}_6$ is clearly a simulator that might fail only if $G$ is saturated but this does not happen if $q< 2^c$. Collecting and adding all the differences yields the claimed $\eps = 8 f_{\coll}(q)$. 
\end{proof}

\subsection{Quantum Indifferentiability of Sponges with Random Functions}
In this subsection we prove quantum indifferentiability of the sponge construction with a uniformly random internal function.

In the quantum indifferentiability simulator we want to sample the outer part of inputs of $ \phif $ and the inner part separately, similarly to the classical one. To do that correctly in the quantum case though we need to maintain two databases: one responsible for the outer part and the other for the inner part. We denote them by $ \overline{D} $ and $ \widehat{D} $ respectively. 

At line~\ref{line:Ho} of the classical simulator we replace the lazy sampled outer state by the output of the random oracle. In the quantum case we want to do the same. Unlike in the classical case we cannot, however, save the input-output pairs of the random oracle $\Ho$ that were sampled to generate the sponge graph, as they contain information about the adversary's query input. An attempt to store this data would effectively measure the adversary's state and render our simulation distinguishable from the real world. To get around this issue we reprepare the sponge graph at the beginning of each run of the simulator. To prepare the sponge graph we query $ \Ho $ on all necessary inputs to $ \hat{\phif} $,  i.e.\ on the inputs that are consistent with a path from the root to a rooted node. This is done gradually by iterating over the length of the paths. We begin with the length-0 paths, i.e. with all inputs in the database $\widehat D$ where the inner part is the all zero string. If the outer part of such an input (which is not changed by the application of $\funPath$) is equal to a padding of an input, that input is queried to determine the outer part of the output of $\phif$, creating an edge in the sponge graph. We can continue with length-1 paths. For each entry of the database $\widehat D$, check whether the input register is equal to a node in the current partial sponge graph. If so, the entry corresponds to a rooted node. Using the entry and the edge connecting its input to the root, a possible padded input to $\sponge$ is created using $\funPath$. If it is a valid padding, $\Ho$ is queried to determine the outer part of the output of $\phif$, etc.

In the proof we will make use of the result from Lemma~\ref{lem:prfind-preim-coll}. Let us denote the bound on inner collisions by
\begin{align}
	f_{\coll}^Q(q) := \frac{7 q (q+1)}{2^c},
\end{align}
which is valid for $ q\in O\left(2^{c/3}\right) $.

The main statement of this section is stated below. Noting the distinguishing bound that we prove, we would like to highlight that our result is most probably tight. Roughly, a quantum algorithm for finding inner-collisions in a sponge construction (such a collision would allow to distinguish a sponge from a random oracle) with a random internal function uses $O(|\setc|^{1/3})$ queries. The distinguishing complexity coming from our bounds, stated without limiting the range of $ q $ for them to apply in Lemma~\ref{lem:prfind-preim-coll}, is the matching $\Omega(|\setc|^{1/3})$.
\begin{thm}[$ \sponge $ with functions, quantum indifferentiability]\label{thm:quantumindiff-trans}
	$\sponge_{\phif}[\pad,r,c]$ calling a random function $ \phif $ is $(q,\eps)$-indifferentiable from a random oracle, Eq.~\eqref{eq:random-oracle}, for \textnormal{quantum} adversaries for any $ q\in O\left(2^{c/3}\right) $ and $\eps=  56\frac{q(q+1)}{2^c} + \sqrt{7\frac{q(q+1)^2}{2^c}}$.
\end{thm}

\begin{proof}
Even though we allow for quantum accessible oracles, the proof we present is very similar to the classical case. The proof follows the same structure, the biggest difference is in the simulators that use the compressed oracle to lazy-sample appropriate answers.

We denote by $\Uni_G$ the unitary that acting on $ \ket{0} $ constructs $G$ including edges consistent with queries held by the quantum compressed database from register $ D $. Similarly we define $\Uni_{\setr\cup \sete}$ to temporarily create a description of the set of supernodes that are rooted or have an outgoing edge.

In Alg.~\ref{alg:sponge-quant-sim-tra} we describe the simulators we use in this proof. In the quantum simulators we also make use of the graph representation of sponges. Note however that in a single query we only care about the graph before the query. Due to that fact we can apply the compressed oracle defined in Alg.~\ref{alg:generalCFO} and additionally analyzed in Lemma~\ref{lem:prfind-preim-coll}.
Eq.~\eqref{eq:prfind-coll} provides a bound of the probability of $ \Find $ (as defined in Section \ref{sec:comp-o2h}) in the case of compressed oracles and relations relevant for the sponge construction.

It is important to note that the ''IF'' statements are in fact quantum controlled operations.
In line~\ref{line:puncture-funcs} we apply a punctured compressed oracle controlled on the input and the database; To correctly perform this operation we postpone the measurement to after uncomputing of $G$ and $\setr\cup\sete$ in line~\ref{line:uncomp-funcs}. This procedure is also discussed in the end of Section~\ref{sec:comp-o2h}.

\begin{algorithm}[hbt]
	\caption{Quantum $\boxed{\Sim_2},\cboxed{darkgreen}{\Sim_3},\cboxed{red}{\Sim_4}\symbolindexmark{\Sim}$, functions }\label{alg:sponge-quant-sim-tra}
\DontPrintSemicolon
\SetKwInOut{Input}{Input}\SetKwInOut{Output}{Output}\SetKwInOut{State}{State}
\State{Quantum compressed database register $D$}
\Input{$\ket{s,v}\in \Hil_{\bits{r+c}}^{\otimes 2}$}
\Output{$\ket{s,v\xor\phif(s)}$}
\BlankLine
Locate input $s$ in $ \overline{D} $ and $\widehat{D}$ \label{line:simfun-locate} \tcp*[r]{Using the correct $ \Samp $}
Apply $\Uni_{\setr\cup \sete}\circ \Uni_G$ to register $\widehat{D}$ and two fresh registers\;
\If(\tcp*[f]{$\hat{s}$-rooted, no saturation}){$\hat{s}\in\setr  \; \wedge \; \setr \cup\sete \neq \bits{c}$\label{line:lazysamp-beg-qtr}}{
Apply $\cboxed{black}{\CStO_{\bits{c}}^{X\widehat{Y}\widehat{D}(s)}}$, $\cboxed{darkgreen}{\cboxed{red}{(\CStO_{\bits{c}}\setminus(\setr\cup\sete))^{X\widehat{Y}\widehat{D}(s)}}}$, result: $ \hat{t} $\label{line:puncture-funcs} \tcp*[r]{The red oracle is punctured!}
Construct a path to $s$: $p:=\funPath(s,G)$\;
\If{$\exists x: p=\pad(x)$\label{step:qif3}}{
$\cboxed{black}{\cboxed{darkgreen}{\textnormal{Apply }\CStO_{\bits{r}}^{X\overline{Y}\,\overline{D}(s)}}}$, result: $ \bar{t} $\;
Write $x$ in a fresh register $X_H$, $\cboxed{red}{\textnormal{apply }\Ho^{XX_H\overline{Y}\,\overline{D}(s)}}$, uncompute $ x $ from $ X_H $, result: $ \bar{t}$\;

}
\Else{
Apply $\CStO_{\bits{r}}^{X\overline{Y}\,\overline{D}(s)}$, result: $ \bar{t} $\;
}
$t:=(\bar{t},\hat{t})$, the value of registers $(\overline{D}^Y(s),\widehat{D}^Y(s))$\;
}
\Else{
Apply $\CStO_{\bits{r+c}}^{XY\overline{D}(s)\widehat{D}(s)}$, result: $ t $\label{line:lazysamp-end-qtr}\;
}
Uncompute $G$ and $\setr\cup\sete$\label{line:uncomp-funcs}\;
Output $\ket{s,v\xor t}$\;
\end{algorithm}
An illustration of the simulators in the quantum case is depicted in Fig.~\ref{fig:games-trans}.
\begin{figure}[hbt]
\begin{center}
\includegraphics[width=0.75\textwidth]{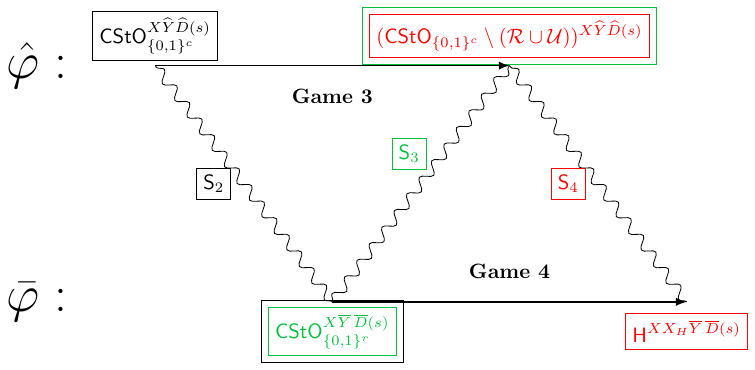}
\caption{Schematics of the simulators defined in Alg.~\ref{alg:sponge-quant-sim-tra}, horizontal arrows signify the change introduced in the labeled game.\label{fig:games-trans}}
\end{center}
\end{figure}

\parag{Game 1}
We start with the real world where the distinguisher $\advA$ has quantum access to a random function $\phif:\bits{r}\times\bits{c}\to\bits{r}\times\bits{c}$ and the $\sponge_{\phif}$ construction using this random function. The definition of the first game is
\begin{equation}
\GAME{1}:= \left( b=1:b\gets \advA[\sponge_{\phif},\phif]\right).
\end{equation}

\parag{Game 2}
In the second game we introduce the simulator $\Sim_2$, defined in Alg.~\ref{alg:sponge-quant-sim-tra}. This algorithm is essentially a compressed random oracle, the only difference are the if statements, note that the behavior of $\Sim_2$ is not influenced by any of the conditional ``if'' statements (in lines~\ref{line:lazysamp-beg-qtr}, and \ref{step:qif3}), because in the end, the output state $t$ is picked uniformly from  $\bits{r+c}$ anyway. The game is defined as:
\begin{equation}
\GAME{2}:= \left( b=1:b\gets \advA[\sponge_{\Sim_2},\Sim_2]\right).
\end{equation}
Because the simulator $\Sim_2$ perfectly models a quantum random function and we use the same function for the private interface we have
\begin{align}
\abs{\PR[\GAME{2}]-\PR[\GAME{1}]}=0.
\end{align}

\parag{Game 3}
In the next step we modify $\Sim_2$ to $\Sim_3$. The game is then
\begin{equation}
\GAME{3}:= \left( b=1:b\gets \advA[\sponge_{\Sim_3},\Sim_3]\right).
\end{equation}
With such a change of the simulators we can use Thm.~\ref{thm:comp-o2h} to bound the difference of probabilities. $\Sim_3$ measures the relation of being an element of $\setr\cup\sete$. This relation is equivalent to $ R_{\preim}\cup R_{\coll} $. The distinguishing advantage is
\begin{align}
\abs{\PR[\GAME{3}]-\PR[\GAME{2}]} \leq  \sqrt{(q+1)\PR[\Find:\advA[\sponge_{\Sim_3},\Sim_3]]}.
\end{align}
Using Lemma~\ref{lem:prfind-preim-coll} we have that
\begin{align}
&\PR[\Find:\advA[\sponge_{\Sim_3},\Sim_3]]\leq f_{\coll}^Q(q).
\end{align}

\parag{Game 4}
In this step we introduce the random oracle $\Ho$ but only to generate the outer part of the output of $\phif$. The game is defined as
\begin{equation}
\GAME{4}:= \left( b=1:b\gets \advA[\sponge_{\Sim_4},\Sim^{\Ho}_4]\right).
\end{equation}
Thanks to the classical argument we have that $\Sim_4$ and $\Sim_3$ are identical until bad, as in Def.~\ref{def:almost-id-or1}. Then we can use Lemma~\ref{lem:almost-id-gms-dist1} to bound the advantage of the adversary
\begin{align}
\abs{\PR[\GAME{4}]-\PR[\GAME{3}]} \leq 4\PR[\Find:\advA[\sponge_{\Sim_3},\Sim_3]]\leq 4 f_{\coll}^Q(q) .
\end{align}

\parag{Game 5}
In this stage of the proof we change the private interface to contain the actual random oracle. In this game the simulator is still $\Sim_4$, the definition is as follows:
\begin{equation}
\GAME{5}:= \left( b=1:b\gets \advA[\Ho,\Sim^{\Ho}_4]\right)
\end{equation}
and the advantage is
\begin{align}
\abs{\PR[\GAME{5}]-\PR[\GAME{4}]} \leq 4\PR[\Find:\advA[\sponge_{\Sim_4},\Sim^{\Ho}_4]] \leq 4 f_{\coll}^Q(q) .
\end{align}
Conditioned on $\neg\Find$, the outputs of the private interface are the same, then the games are identical-until-bad and we can use Lemma~\ref{lem:almost-id-gms-dist1} to bound the advantage of the adversary.

As long as $ \Find $ does not occur and the graph is not saturated the adversary cannot distinguish the simulator from a random function except for the distinguishing advantage that we calculated. Saturation certainly does not occur for $q<2^c$ as the database in every branch of the superposition increases by at most one in every query.
Collecting the differences between games yields the claimed $\eps=8f_{\coll}^Q(q)+\sqrt{(q+1)f^Q_{\coll}(q)}$.
\end{proof}

\section{Conclusions}
We develop a tool that allows for easier translation of classical security proofs to the quantum setting. Our technique shows that given the right proof structure it is relatively easy to prove stronger security notions valid in the quantum world.

It remains open to what degree classical security implies quantum security. An important open problem is specifying features of classical cryptographic constructions that allows constructions to retain their security properties in the quantum world. More concretely, tackling the problem of indifferentiability of other constructions will provide more evidence and possibly lead towards a general answer.

Another open problem is to find a way to quantum lazy sample random permutations. An almost completely new approach has to be devised to tackle this problem as our correctness theorem only applies to local distributions.

\iftsc
\else
\section{Acknowledgments}
The authors thank Gorjan Alagic, Andreas H\"ulsing and Dominique Unruh for enlightening discussions about the superposition oracle technique. Furthermore, the authors thank Dominique Unruh for sharing a draft of \cite{Unruh-forthcoming}. The authors were supported by a NWO VIDI grant (Project No. 639.022.519). We would also like to thank the anonymous reviewers for their insightful comments.
\fi

\tocsectionstar{References}
\printbibliography[heading=none]

\begin{appendix}
	
	\section{Full Proof of Theorem~\ref{thm:indist-cfofo-gen}}\label{sec:full-correctness}
	
	\begin{proof}[Proof of Theorem~\ref{thm:indist-cfofo-gen}]

		We will show that
		\begin{equation}\label{eq:lazy-psifo-psicfo}
			\ket{\Psi_{\FO}}_{AF} = \Dec^D_{\distrD}\ket{\Psi_{\CFO}}_{AD},
		\end{equation}
		where $\ket{\Psi_{\FO}}_{AF}$ is the state resulting from the interaction of $\advA$ with $\FO$ and $\ket{\Psi_{\CFO}}_{AD}$ is the state resulting from the interaction of $\advA$ with $\CFO_{\distrD}$. The state $ \ket{\Psi_{\FO}}_{AF} $ is generated by applying $ \prod_{i=1}^q \Uni_i\circ\FO $  to the $ \ket{\psi_{0}}_A\ket{0^M}_F $, where the $ \ket{\psi_0}_A $ is the initial state of the adversary. In the case of the compressed oracle the state $ \ket{\Psi_{\CFO}}_{AD} $ is generated by applying $ \prod_{i=1}^q \Uni_i\circ\CFO $  to the  $ \ket{\Psi_{0}}_A\ket{(\perp,0)^q}_D $, where $ (\perp,0)^q $ denotes $ q $ pairs $ (\perp,0) $.
		
		We can focus on the state equality from Eq.~\eqref{eq:lazy-psifo-psicfo} because if they are indeed equal, then any adversary's measurement on $ \ket{\Psi_{\FO}}_{AF} $ will yield the output $ b=1 $ with the same probability as on $ \Dec^D_{\distrD}\ket{\Psi_{\CFO}}_{AD} $.
		
		Let us call a database state
		\begin{align}\label{eq:lazy-database-def}
			\ket{\De(\vec{x},\vec{\ize})}:=\ket{x,\eta}_{XY}\ket{x_1,\ize_1}_{D_1}\cdots\ket{x_s,\ize_s}_{D_s}\cdots \ket{\perp,0}_{D_{q}},
		\end{align}
		where $ \vec{x}:=(x_1,x_2,\dots,x_s) $ and $ \vec{\ize}:=(\ize_1,\ize_2,\dots,\ize_s) $ well-formed, if no $ x_i $ in $ \vec{x} $ is $\perp$ and no $ \ize_i $ in $ \vec{\ize} $ is zero.
		
		To prove Eq.~\eqref{eq:lazy-psifo-psicfo} we show that 
		\begin{equation}
			\FO\circ\Dec_{\distrD}\ket{\De(\vec{x},\vec{\ize})}=\Dec_{\distrD} \circ \CFO_{\distrD}\ket{\De(\vec{x},\vec{\ize})}.\label{eq:fodec-deccfo}
		\end{equation}
		This is sufficient for the proof of the theorem as $  \ket{\Psi_{\FO}} $ is generated by a series of the adversary's unitaries intertwined with oracle calls. If we show that $ \FO = \Dec_{\distrD} \circ \CFO_{\distrD}\circ \Dec_{\distrD}^{\dagger}$, when acting on well-formed databases, then everything that happens on the oracle's register side can be compressed.
		Note that as we start from the empty oracle state and only apply the oracle to the oracle register, the database will always be well-formed.
		
		We study the action of $ \Dec_{\distrD} $ on the state in Eq.~\eqref{eq:lazy-database-def}.  To write the output state we need to name the matrix elements of the sampling unitary: $ \left( \Samp_{\distrD}(\mathcal{X}) \right)_{f\vec{\ize}}= a_{f \vec{\ize}}(\mathcal{X})$, the column index consists of a vector of size $M$ with exactly $s$ non-zero entries: $ \vec{\ize}= (0,\dots, 0,\ize_1,0\dots,0, \ize_2,0,\ldots) $. The decompressed state is
		\begin{align}\label{eq:Upsilon}
			\ket{\Upsilon(\vec{x},\vec{\ize})}_F:=& \Dec_{\distrD}\ket{\De(\vec{x},\vec{\ize})} = \sum_{\phi\in\mathcal{F} }\frac{1}{\sqrt{N^{M}}}\sum_{f\in\mathcal{F}} \omega_N^{\phi \cdot f} \; a_{f \vec{\ize}}(\mathcal{X}) \; \ket{\phi_0}_{F(0)}\cdots  \ket{\phi_{M-1} }_{F(M-1)},
		\end{align}
		where $ \phi\cdot f=\sum_{x\in\mathcal{X}}\phi_x f(x)\mod N $ and by $ f(x) $ we denote row number $ x $ of the function truth table $ f $.
		
		Using the fact that $ \Samp_{\distrD} $ is defined for a product distribution, as in Def.~\ref{def:local-samp}, we have that $\Samp_{\distrD}(\mathcal{X})=\Samp_{\distrD}(\mathcal{X}\setminus \{x\})\circ \Samp_{\distrD}(x)$ and we can focus our attention on some fixed $ x $: isolate register $ F(x) $ with amplitudes depending only on $ x $. Let us compute this state after application of $ \FO $, note that $ \FO $ only subtracts $ \eta $ from $ \F(x) $:
		\begin{align}
			\begin{split}\label{eq:Upsilon-split}
				\FO\ket{x,\eta}_{XY} & \ket{\Upsilon(\vec{x},\vec{\ize})}_F = \ket{x,\eta}_{XY}\sum_{\phi',f'\in \mathcal{F}(\mathcal{X}\setminus\{x\}) }\frac{1}{\sqrt{N^{M-1}}} \; \omega_N^{\phi' \cdot f'} \; a_{f' \vec{\ize}'}(\mathcal{X}\setminus\{x\}) \\ 
				&\cdot\ket{\phi_0}_{F(0)}\cdots \left( \sum_{\zeta,z\in[N]} \frac{1}{\sqrt{N}} \; \omega_N^{\zeta\cdot z} \;
				a_{z \ize_x}(x) \; \ket{\zeta-\eta}_{F(x)} \right)\cdots  \ket{\phi_{M-1} }_{F(M-1)},
			\end{split}
		\end{align}
		where $ \vec{\ize}'\in\mathcal{Y}^{M-1} $ denotes the vector of $ \ize_i $ without the row with index $ x $. Note that $ \ize_x=0 $ if $ x $ was not in $ \vec{x} $ before decompression and $ \ize_x\neq 0 $ otherwise.
		
		The harder part of the proof is showing that the right hand side of Eq.~\eqref{eq:fodec-deccfo} actually equals the left hand side that we just analyzed. Let us inspect $ \ket{\De(\vec{x},\vec{\ize})} $ after application of the compressed oracle
		\begin{align}\label{eq:Deprims}
			&\CFO_{\distrD}\ket{x,\eta}_{XY}\ket{\De(\vec{x},\vec{\ize})}_D = \ket{x,\eta}_{XY} \nonumber\\
			&\cdot \left( \sum_{\tilde{\ize}_x\neq 0} \alpha(x,\eta,\ize_x,\tilde{\ize}_x) \; \ket{\De'_{\ADD/\UPD}}_D + 
			\alpha(x,\eta,\ize_x,0) \; \ket{\De'_{\REM/\NOT}}_D  \right),
		\end{align}
		where $ \tilde{\ize}_x $ is the new value of $ \De^Y(x) $ and $ \ize_x $ is the old content of the database. By $\De'_{\ADD/\UPD}$ we denote the database $ \De(\vec{x},\vec{\ize}) $ with entry $ \tilde{\ize}_x\neq 0 $, it corresponds to $ x $ being added or updated. By $\De'_{\REM/\NOT}$ we denote the database where $ \tilde{\ize}_x=0 $, meaning $ x $ was removed  from $ \De $ or nothing happened. The function $\alpha(\cdot)$ denotes the corresponding amplitudes. 
		
		Before we proceed with decompression of the above state let us calculate the amplitudes $ \alpha $. Again using the definition of $ \Samp_{\distrD} $ we describe the action of the compressed oracle on a single $ x $ step by step.  Below we denote by $ \funRem $ removing $ \ize=0 $ from $ \De $ and by $ \funSub $ subtraction of $ \eta $ from database register $ D^Y $. We start with a database containing $ (x,\ize_x) $, which we can always assume due to line~\ref{line:x-alwaysin-D} in Alg.~\ref{alg:generalCFO}. In the case that $ x $ was not already in $ \De $ we have $ \ize_x = 0 $, otherwise it is the value defined in previous queries.  The simplification we make is to describe $ \CFO_{\distrD}$ acting on a single-entry database. We do not lose generality by that as the only thing that changes for $ q $ larger than one is maintaining proper sorting and padding, which can be easily done (see Appendix~\ref{sec:algorithms} for details). 
		The calculation of $ \CFO_{\distrD}$ on a basis state follows:
		\begin{align}
			& \ket{x,\eta}_{XY}\ket{x,\ize_x}_D  \overset{\Samp_{\distrD}}{\mapsto} \ket{x,\eta}_{XY} \sum_{z\in[N]} \; a_{z\ize_x}(x) \; \ket{x,z}_D  \\
			\overset{\QFT^{D^Y}_N}{\mapsto} &\ket{x,\eta}_{XY} \sum_{z\in[N]}a_{z\ize_x}(x)\sum_{\zeta\in[N]}\frac{1}{\sqrt{N}} \; \omega_N^{\zeta\cdot z} \; \ket{x,\zeta}_D \\
			\overset{\funSub}{\mapsto} &\ket{x,\eta}_{XY} \sum_{z,\zeta \in[N]}a_{z\ize_x}(x)\frac{1}{\sqrt{N}} \; \omega_N^{\zeta\cdot z} \; \ket{x,\zeta-\eta}_D \\
			\overset{\QFT^{\dagger D^Y}_N}{\mapsto}  &\ket{x,\eta}_{XY} \sum_{z,\zeta \in[N]}a_{z\ize_x}(x)\frac{1}{\sqrt{N}} \; \omega_N^{\zeta\cdot z} \sum_{z'\in[N]}\frac{1}{\sqrt{N}} \; \bar{\omega}_N^{z'\cdot(\zeta-\eta)} \; \ket{x,z'}_D \\
			= &\ket{x,\eta}_{XY} \sum_{z \in[N]} a_{z\ize_x}(x) \underset{=\bar{\omega}_N^{-z\cdot\eta} \; \delta(z',z)}{\underbrace{\sum_{z',\zeta \in[N]}\frac{1}{N} \; \omega_N^{\zeta\cdot z} \; \bar{\omega}_N^{z'\cdot(\zeta-\eta)}}} \; \ket{x,z'}_D \label{eq:qft-sub-qft} \\
			\overset{ \Samp_{\distrD}^{\dagger D}(x) }{\mapsto}  &\ket{x,\eta}_{XY} \sum_{z \in[N]}a_{z\ize_x}(x) \; \omega_N^{z\cdot\eta} \sum_{\tilde{\ize}_x\in[N]}\bar{a}_{z\tilde{\ize}_x}(x) \; \ket{x,\tilde{\ize}_x}_D \\
			= &\ket{x,\eta}_{XY}  \sum_{\tilde{\ize}_x\in[N]}  \underset{:=\alpha(x,\eta,\ize_x,\tilde{\ize}_x)}{\underbrace{  \sum_{z \in[N]} a_{z\ize_x}(x) \; \omega_N^{z\cdot \eta} \; \bar{a}_{z\tilde{\ize}_x}(x) }}  \; \ket{x,\tilde{\ize}_x}_D \\
			\overset{ \funRem^D }{\mapsto}  &\ket{x,\eta}_{XY}  \left( \sum_{\ize\in[N]\setminus\{0\}} \alpha(x,\eta,\ize_x,\tilde{\ize}_x) \; \ket{x,\tilde{\ize}_x}_D + \alpha(x,\eta,\ize_x,0) \; \ket{\perp,0}_D \right).
		\end{align}
		In the above equations we have defined $\alpha$ as
		\begin{align}
			\alpha(x,\eta,\ize_x,\tilde{\ize}_x) :=  \sum_{z\in[N]} a_{z\ize_x}(x) \; \bar{a}_{z\tilde{\ize}_x}(x) \; \omega_N^{z\cdot\eta}.
		\end{align}
		
		After decompressing the state from Eq.~\eqref{eq:Deprims}, the resulting database state will be $ \sum_{\tilde{\ize}_x\neq 0} \alpha(x,\eta,\ize_x,\tilde{\ize}_x) \; \ket{\Upsilon(\De'_{\ADD/\UPD})} +  \alpha(x,\eta,\ize_x,0) \; \ket{\Upsilon(\De'_{\REM/\NOT})}_D $, where we overload notation of $ \ket{\Upsilon(\vec{x},\vec{\ize})} $ to denote that $ (\vec{x},\vec{\ize})$ consists of values in the respective databases.
		We can write down this state in more detail using Eq.~\eqref{eq:Upsilon-split}:
		\begin{align}\label{eq:Upsilon-split-CFO}
			&\Dec_{\distrD} \circ \CFO_{\distrD}\ket{x,\eta}_{XY}\ket{\De(\vec{x},\vec{\ize})}_D \nonumber\\
			&= \sum_{\phi',f'\in \mathcal{F}(\mathcal{X}\setminus\{x\}) }\frac{1}{\sqrt{N^{M-1}}} \; \omega_N^{\phi' \cdot f'} \; a_{f' \vec{\ize}'}(\mathcal{X}\setminus\{x\}) \; \ket{\phi_0}_{F(0)} \cdots \nonumber \\ 
			&\cdot \left( \sum_{\tilde{\ize}_x\neq 0}\alpha(x,\eta,\ize_x,\tilde{\ize}_x) \sum_{\zeta,z\in[N]} \frac{1}{\sqrt{N}}  \omega_N^{\zeta\cdot z}  a_{z \tilde{\ize}_x}(x)  \ket{\zeta}_{F(x)} \right.\nonumber\\
			&\left.+\alpha(x,\eta,\ize_x,0) \sum_{\zeta,z\in[N]} \frac{1}{\sqrt{N}}  \omega_N^{\zeta\cdot z}  a_{z 0}(x)  \ket{\zeta}_{F(x)}\right) \cdots  \ket{\phi_{M-1} }_{F(M-1)} .
		\end{align}
		In the above equation we notice that 
		\begin{align}
			&\sum_{\tilde{\ize}_x\neq 0}\alpha(x,\eta,\ize_x,\tilde{\ize}_x) \sum_{\zeta,z\in[N]} \frac{1}{\sqrt{N}}  \omega_N^{\zeta\cdot z}  a_{z \tilde{\ize}_x}(x)  \ket{\zeta}_{F(x)} \nonumber\\
			&+\alpha(x,\eta,\ize_x,0) \sum_{\zeta,z\in[N]} \frac{1}{\sqrt{N}}  \omega_N^{\zeta\cdot z}  a_{z 0}(x)  \ket{\zeta}_{F(x)} \nonumber\\
			&=\sum_{\zeta,z\in[N]} \frac{1}{\sqrt{N}} \; \omega_N^{\zeta\cdot z} \sum_{\tilde{\ize}_x\in[N]}\alpha(x,\eta,\ize_x,\tilde{\ize}_x) \; a_{z \tilde{\ize}_x}(x) \; \ket{\zeta}_{F(x)}
		\end{align}
		which comes from the fact that $ \Samp_{\distrD} $ is a unitary and $ \sum_{j\in[N]} a_{ij} \bar{a}_{kj}=\delta_{ik}$ and therefore we have
		\begin{align}
			&\sum_{\tilde{\ize}_x\in[N]}\alpha(x,\eta,\ize_x,\tilde{\ize}_x) \; a_{z \tilde{\ize}_x}(x) \nonumber\\
			& = \sum_{z'\in[N]} \underset{=\delta_{z',z}}{\underbrace{\sum_{\tilde{\ize}_x\in[N]}\bar{a}_{z'\tilde{\ize}_x}(x) \; a_{z\tilde{\ize}_x}(x)}} a_{z'\ize_x}(x) \; \omega_N^{z'\cdot\eta} = a_{z\ize_x}(x) \; \omega_N^{z\cdot\eta}.
		\end{align}
		Together with changing the variable $ \zeta\mapsto \zeta -\eta$ and observing Eq.~\eqref{eq:Upsilon-split} we derive the claimed identity:
		\begin{align}
			&\Dec_{\distrD} \circ \CFO_{\distrD}\ket{x,\eta}_{XY}\ket{\De(\vec{x},\vec{\ize})}_D \nonumber\\
			&=\FO \;\ket{x,\eta}_{XY} \ket{\Upsilon(\vec{x},\vec{\ize})}= \FO\circ\Dec_{\distrD}\; \ket{x,\eta}_{XY}\ket{\De(\vec{x},\vec{\ize})}_D.
		\end{align}
	\end{proof}

	\section{Full Proof of Theorem~\ref{thm:comp-o2h}}\label{sec:full-o2h}

	\begin{proof}[Proof of Theorem~\ref{thm:comp-o2h}]
		The proof works almost the same as the proof of Theorem~1 of \cite{ambainis2018quantum}. Let us state the analog of Lemma~5 from \cite{ambainis2018quantum}.
		
		For the following lemma let us first define two algorithms. 
		Let $\advA^{\Ho}(z)$ be a unitary quantum algorithm with oracle access to $\Ho$ with query depth $d$. Let $Q$ denote the quantum register of $\advA$ and $D$ the database of the compressed oracle $\Ho$. We also need a ``query log'' register $L$ consisting of $d$ qubits. 
		
		Let $\advB^{\Ho,R}(z)$ be a unitary quantum algorithm acting on registers $Q$ and $L$ and having oracle access to $\Ho$.
		First we define the following unitary 
		\begin{equation}
			\V_{R,i}\ket{D}_D\ket{l_1,l_2,\dots,l_d}_L:= \begin{cases}
				\ket{D}_D\ket{l_1,l_2,\dots,l_d}_L & \textnormal{ if } R(\ket{D}_D)=0 \\
				\ket{D}_D\ket{l_1,\dots,l_i\xor 1,\dots,l_d}_L & \textnormal{ if } R(\ket{D}_D)=1
			\end{cases},
		\end{equation}
		where $R(\ket{D}_D)$ denotes the outcome of the projective binary measurement on $D$. The unitary exists for all relations. One can just coherently compute $R(D)$ into an auxiliary register, apply CNOT from that register to $L_i$ and then uncompute $R(D)$. If the relation is efficiently computable, then so is the unitary.
		We define $\advB^{\Ho,R}(z)$ as:
		\begin{itemize}
			\item Initialize the register $L$ with $\ket{0^d}$.
			\item Perform all operations that $\advA^{\Ho}(z)$ does.
			\item For all $i$, after the $i$-th query of $A$ apply the unitary $\V_R$ to registers $D,L$.
		\end{itemize}
		
		Let $\ket{\Psi_{\advA}}$ denote the final state of $\advA^{\Ho}(z)$, and $\ket{\Psi_{\advB}}$ the final state of $\advB^{\Ho,R}(z)$. Let $\tilde{P}_{\textnormal{find}}$ be the probability that a measurement of $L$ in the computational basis in the state $\ket{\Psi_{\advB}}$ returns $l\neq 0^d$, i.e.\ $\tilde{P}_{\textnormal{find}}:=\norm{\mathbbm{1}^{Q,D}\otimes (\mathbbm{1}^L-\ket{0^d}_L\bra{0^d})\ket{\Psi_{\advB}}}^2$.
		
		To deal with relation $ R_1 $ we consider algorithms with all measurements postponed to the end of their operation; Instead of performing the actual measurement we save the outcome into a fresh quantum register---with $ \V_R $ as in Alg.~\ref{alg:measureR}, note that prior to the measurement this fresh register can hold a superposition. Moreover we postpone the measurement of the auxiliary register until the very end of the run of the quantum algorithm. The coherent evaluation of $ R_1 $ happens in both algorithms. In addition, the proof below does not make use of the particular form of the unitaries that are applied between the measurements of $R_2$, so the evaluation of $R_1$ can be absorbed into the compressed oracle unitary.
		
		\begin{lemm}[Compressed oracle O2H for pure states]
			Fix a joint distribution for $\Ho, R, z$. 
			Consider the definitions of algorithms $ \advA $ and $ \advB $ and their quantum states, then
			\begin{equation}
				\norm{\ket{\Psi_{\advA}}\otimes\ket{0^d}_L-\ket{\Psi_{\advB}}}^2\leq (d+1)\tilde{P}_{\textnormal{find}}.
			\end{equation}
		\end{lemm}
		\begin{proof}
			This lemma can be proved in the same way as Lemma~5 of \cite{ambainis2018quantum}.
			Here we omit some details and highlight the most important observation of the proof.
			
			First define $\advB_{\textnormal{count}}$ that works in the same way as $\advB$ but instead of storing $L$, the log of queries with $D$ in relation,  it keeps \emph{count}---in register $C$---of how many times a query resulted in $R(\ket{D}_D)=1$. The state that results from running $\advB_{\textnormal{count}}$ is $\ket{\Psi_{\advB_{\textnormal{count}}}}=\sum_{i=0}^d \ketun{\Psi_{\advB_{\textnormal{count}}}^{i}}\ket{i}_C$ and similarly $\ket{\Psi_{\advB}}=\sum_{l\in\bits{d}} \ketun{\Psi_{\advB}^{l}}\ket{l}_L$, where $\ketun{\Psi}$ denotes a sub-normalized state.
			We can observe that $\ket{\Psi_{\advA}}=\sum_{i=0}^d\ketun{\Psi_{\advB_{\textnormal{count}}}^{i}}$. As $\tilde{P}_{\textnormal{find}}$ is the probability of measuring at least one bit in the register $L$ of $\advB$, or counting at least one fulfilling of $R$ in $C$, we have that $\ketun{\Psi_{\advB}^{0^d}}=\ketun{\Psi_{\advB_{\textnormal{count}}}^{0}}$. From the definition we also have $\tilde{P}_{\textnormal{find}}=1-\norm{\ketun{\Psi_{\advB_{\textnormal{count}}}^{0}}}^2$. Using the above identities we can calculate the bound
			\begin{align}
				& \norm{\ket{\Psi_{\advB}}-\ket{\Psi_{\advA}}\otimes\ket{0^d}_L}^2=\norm{\sum_{i=1}^d \ketun{\Psi_{\advB_{\textnormal{count}}}^{i}} }^2 + \tilde{P}_{\textnormal{find}}  \overset{\triangle}{\leq} \left( \sum_{i=1}^d\norm{\ketun{\Psi_{\advB_{\textnormal{count}}}^{i}} }\right)^2 + \tilde{P}_{\textnormal{find}}\nonumber\\
				& \overset{\textnormal{J-I}}{\leq} d \underset{=\tilde{P}_{\textnormal{find}}}{\underbrace{\sum_{i=1}^d\norm{\ketun{\Psi_{\advB_{\textnormal{count}}}^{i}} }^2}} + \tilde{P}_{\textnormal{find}}=  (d+1)\tilde{P}_{\textnormal{find}},
			\end{align}
			where $\triangle$ denotes the triangle inequality and J-I denotes the Jensen's inequality. It is apparent that introducing $\advB_{\textnormal{count}}$ gave us a more coarse-grained look at the initial algorithm $\advB$, resulting in a tighter bound.
		\end{proof}
		
		The rest of the proof of the theorem follows the same reasoning as the proof of Lemma~6 in \cite{ambainis2018quantum} with the modifications shown in the above lemma. Using bounds on fidelity (Lemma~3 and Lemma~4 of \cite{ambainis2018quantum}) and monotonicity and joint concavity of fidelity (from Thm.~9.6 and Eq.~9.95 of \cite{nielsen2002quantum}) one can generalize the results to the case of arbitrary mixed states.
	\end{proof}

	\section{Second Proof of Lemma~\ref{lem:prfind-preim-coll}}\label{sec:full-proof}

	\begin{proof}[Proof of Lemma~\ref{lem:prfind-preim-coll}]
				In Lemma~\ref{lem:prfind-preim-coll} we prove a bound on the probability of finding a database fulfilling the relation of collision or a preimage of 0. This event of finding is denoted by $ \Find $. This relation is crucial in the proof of quantum indifferentiability of the sponge construction. 
				
				The first observation of the proof is that the probability of $\Find$ is the sum of probabilities that after the $ i $’th query we find a database that fulfills the relation given that we did not find such database in any previous query. Hence, the proof focuses on calculating this probability for any $ i $ and then performing the sum. 
				
				It is in general challenging to calculate such probability, and especially challenging to write out the joint state of the adversary and the oracle after $ i $ queries to the punctured oracle. Our solution to this challenge is to define an auxiliary state, called the \emph{good} state $ \ket{\Psi^{\Good}_i} $. This is an auxiliary state of the adversary and the oracle register that is easier to handle from the true state $ \ket{\Phi_i} $ resulting from the interaction of $ \advA $ with the punctured oracle.
				
				In a hybrid argument we introduce a sum over differences between the actual state and the good state. This is the focal point of our proof, if we find this difference, then we can work with the good state and calculate the bound on $ \Find $ much easier. Technically the most difficult part of our proof is bounding the norm of the difference of the actual and the good states, it is the topic of section~\ref{sec:bound-step} and Lemma~\ref{lem:lazy-step}.
				
				The second important technical part is calculating the norm of finding a database that fulfills the relation in the good state after a query. Thankfully, after the analysis of the first problem we mentioned it is a relatively easy task.

		Punctured oracles are defined in Definition~\ref{def:punctured}. We start the proof by specifying some operations involved in that definition. 
		
		\parag{Introduction}
		We define a ``lazy'' approach to calculating the number of non-empty entries in $ D $. In this unitary we focus on using the ordered structure of $ D^X $. 
		We use the phase oracle instead of the standard oracle; in detailed calculations that we do later on in the proof, $ \CPhO $ is easier to deal with than $ \CStO $.
		
		Let us define $ \Queries $, a unitary that outputs the size of a database. It acts on an auxiliary register $ S $ and is controlled on $ D $. This unitary acts exactly like Alg.~\ref{alg:generalCFO} in lines ~\ref{line:queries1} and \ref{line:queries2}: it counts the number of non-padding ($ x\neq \perp $) entries.
		
		The full description of the measurement involves using an auxiliary register $ J $---note Def.~\ref{alg:measureR} measuring a relation---with a bit stating whether the database fulfills the relation. Then the actual measurement is a computational basis measurement of register $ J $. 
		The measurement that we apply after $ \CPhO_{\mathcal{Y}} $, in line~\ref{line:measure-relation} of Alg.~\ref{alg:measureR} is 
		\begin{align}
			&\Juni_{R}:=\id\otimes \ket{1}_J\bra{1},\\
			&\overline{\Juni}_{R}:=\id\otimes \ket{0}_J\bra{0}.
		\end{align}
		
		In the following we focus on the punctured oracle just prior to measurement $ \Juni_R $. A unitary that omits the last step of Alg.~\ref{alg:measureR} in $ \CPhO_{\mathcal{Y}}\setminus R_{\preim}\cup R_{\coll} $ acts on registers $ ADJ $, we define it as
		\begin{align}
			\CPhO_{\mathcal{Y}}\setminus \V_{R} := \Queries^{\dagger} \circ\V_R	\circ\Queries\circ	\CPhO_{\mathcal{Y}},
		\end{align}
		where the unitary $ \V_R $ checks whether the queried values in registers $ D $ fulfill the relation $ R $---in our case it is the collision and preimage relations from Eqs.~\eqref{eq:def-rel-coll}, \eqref{eq:def-rel-preim}---and saves the single bit answer to register $ J $.

		We proceed by rephrasing the definition of $ \PR[\Find:\advA[\CPhO_{\mathcal{Y}}\setminus R_{\preim}\cup R_{\coll}]] $, after that we treat the part specific to our relation. We follow Eq.~\eqref{eq:def-find} to analyze the probability of $ \Find $:
		\begin{align}
			&\PR[\Find:\advA[\CPhO_{\mathcal{Y}}\setminus R_{\preim}\cup R_{\coll}]] = 1- \norm{\left(\prod_{i=q}^{1} \overline{\Juni}_{R}\Uni_i\CPhO_{\mathcal{Y}}\setminus \V_{R} \right) \ket{\Psi_0}\ket{0}_{J} }^2  \\
			&= 1- \norm{\left(\prod_{i=q-1}^{1} \overline{\Juni}_{R}\Uni_j\CPhO_{\mathcal{Y}}\setminus \V_{R}\right)\ket{\Psi_0}\ket{0}_{J} }^2\nonumber\\
			& + \norm{\Juni_{R}\Uni_q\CPhO_{\mathcal{Y}}\setminus \V_{R} \left(\prod_{i=q-1}^{1}\overline{\Juni}_{R}\Uni_j\CPhO_{\mathcal{Y}}\setminus \V_{R}\right)\ket{\Psi_0}\ket{0}_{J} }^2 =\dots =\\
			&= \sum_{i=1}^q \norm{\Juni_{R}\Uni_i\CPhO_{\mathcal{Y}}\setminus \V_{R} \underset{:=\Uni_{i-1}\ket{\Phi_{i-1}}}{\underbrace{ \left( \prod_{j=i-1}^{1}\overline{\Juni}_{R}\Uni_j\CPhO_{\mathcal{Y}}\setminus \V_{R}\right) \ket{\Psi_0}\ket{0}_{J}}}}^2  \\
			&= \sum_{i=1}^q \norm{\Juni_{R}\Uni_i\CPhO_{\mathcal{Y}}\setminus \V_{R} \Uni_{i-1}\ket{\Phi_{i-1}}}^2, \label{eq:prob-find-CPhOminR}
		\end{align}
		where $ \ket{\Psi_0} $ is the initial state of the adversary. Note that in the definition
		\begin{align}
			\ket{\Phi_{i-1}}:= \Uni_{i-1}^{\dagger}\left( \prod_{j=i-1}^{1}\overline{\Juni}_{R}\Uni_j\CPhO_{\mathcal{Y}}\setminus \V_{R}\right) \ket{\Psi_0}\ket{0}_{J}
		\end{align}
		we use $ [\Uni_{i-1},\overline{\Juni}_{R}]=0 $\footnote{The commutator of two operators (matrices) is defined as $ [\mathsf{A},\mathsf{B}]:=\mathsf{A}\mathsf{B}-\mathsf{B}\mathsf{A} $.}.
		Here, the second and third equations follow from the fact that $\|\ket v\|^2=\|\mathsf{P}\ket v\|^2+\|(\mathbbm{1}-\mathsf{P})\ket v\|^2$ for all $\ket v$ and projectors $\mathsf{P}$.
		
		In what follows we analyze $ \norm{\Juni_{R}\Uni_i\CPhO_{\mathcal{Y}}\setminus \V_{R} \Uni_{i-1}\ket{\Phi_{i-1}}}^2 $. Our approach is to propose a state $ \ket{\Psi_{i-1}^{\Good}} $, close to the original $ \ket{\Phi_{i-1}} $, for which bounding $ \norm{\Juni_{R}\Uni_i\CPhO_{\mathcal{Y}}\setminus \V_{R}\Uni_{i-1} \ket{\Psi_{i-1}^{\Good}}\ket{0}_{J}}^2 $ is easy. The intuition behind $ \ket{\Psi_{i-1}^{\Good}} $ is to have a superposition over databases that do not contain $ y=0 $ and are collision free for the queried values.
		
		\parag{The good state}
		To define the good state we specify the set of bad databases $ D\in R $.  For the relation $ R_{\preim}\cup R_{\coll} $ we have
		\begin{align}\label{eq:def-bad-sett}
			&\setb(s)   := [N]^s\setminus \left\{(y_1,\dots,y_s)\in[N]^s: \textnormal{all $ y_i $ are distinct and }\neq 0\right\},\\
			&\setb(1\mid D):=\{y\}_{y\in D^Y}\cup\{0\}.
		\end{align}
		The second set defined above is the subset of the codomain of the sampled function corresponding to the new value creating a collision or being a preimage of $ 0 $.
		To better understand $ \setb(1\mid D) $ let us assume $ D\not \in R $ and $ x $ is some input $ \not\in D^X $. Then $  \setb(1\mid D) $ is the set of $ y $ such that $ D\cup\{(x,y)\}\in R  $. 
		We also define a coefficient $ b(s) $ defined as
		\begin{align}\label{eq:b-setb}
			b(s):=\abs{\setb(1\mid D) }, \textnormal{ where } D\not \in \setb(s-1),
		\end{align} 
		where we use the fact that $ \abs{\setb(1\mid D) } $ depends only on the size of $ D $ and not the actual contents of it.
		We define $ \setb(1\mid D)  $ in a way specific to $ R_{\coll}\cup R_{\preim} $ but the definition can be easily extended to other relations. As examples consider $ R_{\preim} $, then $ b(s)=1 $, there is just one value $ y=0 $ that causes a fresh query to be in relation; For $ R_{\coll} $ we have $ b(s)=s-1 $, the new $ y $ can be any of the previously queried values to make $ D $ fulfill the relation. Finally for our relation $ R_{\preim}\cup R_{\coll} $ we have $ b(s)=s $, database $ D $ consists of $ s-1 $ distinct values that are distinct from $ 0 $, matching any of them or $ 0 $ causes $ D^Y\cup\{y\} $ to be in $ \setb(s) $.
		Throughout the rest of this proof we do not evaluate $ b(s) $, which makes it is easier to reuse the proof for other relations.

		In what follows we write $ \vec{x} $ to denote all the previous inputs asked by the adversary and $ (x,\eta) $ is the last query. The state $ \ket{\Psi_{i,R}^{\Good}}_{AD} $ corresponds to the adversary's state just after the $ i $-th query and before the application of $ \Uni_i $. The size of the database $ s $ depends on whether the new query $ x $ was added to, updated, or removed from the database, it equals $ \abs{\vec{x}\cup\{x\}} $, $ \abs{\vec{x}} $, or $ \abs{\vec{x}\setminus\{x\}} $ respectively. After $ i $ queries $ s $ can range from $ 0 $ to $ i $ and the joint state of $ \advA $ and the oracle can be a superposition over different database sizes. We denote the outputs given to $ \advA $ by $ \vec{y}:= (y_1,\dots,y_s)$.
		When we use set operations on vectors we mean a set consisting of entries of $ \vec{x} $, there are no repetitions in the vector as this is an invariant of the oracle. By $ D(\perp) $ we denote the part of the database containing empty entries. Adversary's work register is denoted by $ A^W $ and its contents by $ \psi(x,\eta,\vec{x},\vec{\eta},w) $, where $ w $ can be any value of finite size.
		We define the good state as:
		\begin{align}\label{eq:def-good-state}
			&\ket{\Psi_{i,R}^{\Good}}_{AD}:=  \sum_{x,\eta,\vec{x},\vec{\eta},w}\alpha_{x,\eta,\vec{x},\vec{\eta},w}\ket{x,\eta}_{A^{XY}}\ket{\psi(x,\eta,\vec{x},\vec{\eta},w)}_{A^W} \nonumber \\
			&\sum_{\vec{y}\not\in\setb(s) } \frac{1}{\sqrt{(N-b(1))(N-b(2))\cdots(N-b(s))}} \omega_N^{\vec{\eta}\cdot \vec{y}} \ket{(x_1,y_1),\dots,(x_s,y_s)}_{D(\vec{x})} \nonumber \\
			& \sum_{y_{s+1},\dots,y_q\in[N]}\frac{1}{\sqrt{N^{q-s}}}\ket{(\perp,y_{s+1}),\dots,(\perp,y_q)}_{D(\perp)}.
		\end{align} 
		In case we have added $ x $ to $ D $, the database above contains $ (x,y_j) $. In the rest of the proof we omit the subscript $ R $, however note that $ \ket{\Psi_{i}^{\Good}} $ does indeed depend on $ R $.
	
		Another way to define the good state is to consider the joint state of the adversary and the non-punctured oracle  just after the $ i $-th query. The good state is then this state after a projection of register $ D $ with $ \overline{\Juni}_{R} $. Normalization of the projected state comes from multiplying each branch corresponding to a given size of the database by an appropriate $\sqrt{ \frac{N^{s}}{(N-b(1))\cdots(N-b(s))} }$ factor. The reason why the good state is normalized is that for a fixed set of queries we can think of definining it as $ \advA $ interacting with the normalized database register using $ \PhO $ instead of $ \CPhO $. This intuition works for every branch of the superposition separately. Now combining all branches together also gives a normalized state, because they origin from a valid interaction of a unitary adversary with $ \CPhO $ (as mentioned in the beginning of this section).

	\parag{Final Bound}
To calculate the probability of measuring $ R $, Eq.~\eqref{eq:prob-find-CPhOminR} implies
\begin{align}
	\PR[\Find]  \leq 
	\sum_{i=1}^q \norm{\Juni_{R}\Uni_i\CPhO_{\mathcal{Y}}\setminus \V_{R} \Uni_{i-1}\ket{\Phi_{i-1}}}^2.
\end{align}
We use the good state to bound the elements of the sum in the following way:
\begin{align}\label{eq:lazy-Jnorm}
	&\norm{\Juni_{R}\Uni_i\CPhO_{\mathcal{Y}}\setminus \V_{R} \Uni_{i-1}\ket{\Phi_{i-1}}}\nonumber\\
	&
	\leq \norm{ \ket{\Phi_{i-1}}-\ket{\Psi^{\Good}_{i-1}}}+
	\norm{\Juni_{R}\Uni_i\CPhO_{\mathcal{Y}}\setminus \V_{R} \Uni_{i-1}\ket{\Psi^{\Good}_{i-1}}}
	.
\end{align}

Next we bound the two norms in Eq.~\eqref{eq:lazy-Jnorm}. First we bound the distance of the good state from the state resulting from the interaction with the non-punctured oracle $ \ket{\Phi_i}_{ADJ} $. We simplify this task with the following derivation:	
\begin{align}
	&\norm{\ket{\Psi_i^{\Good}}_{AD}\ket{0}_{J}- \ket{\Phi_i}_{ADJ}}\nonumber\\
	&=\norm{\ket{\Psi_i^{\Good}}_{AD}\ket{0}_{J} -\overline{\Juni}_{R}  \CPhO_{\mathcal{Y}}\setminus \V_{R} \Uni_{i-1}\ket{\Phi_{i-1}}_{ADJ}}\\
	&\leq \norm{\ket{\Psi_i^{\Good}}_{AD}\ket{0}_{J}- \overline{\Juni}_{R}  \CPhO_{\mathcal{Y}}\setminus \V_{R}\Uni_{i-1}\ket{\Psi_{i-1}^{\Good}}_{AD}\ket{0}_{J}} \nonumber \\
	&+ \norm{\overline{\Juni}_{R}  \CPhO_{\mathcal{Y}}\setminus \V_{R}\Uni_{i-1}\ket{\Psi_{i-1}^{\Good}}_{AD}\ket{0}_{J}- \overline{\Juni}_{R}  \CPhO_{\mathcal{Y}}\setminus \V_{R}\Uni_{i-1}\ket{\Phi_{i-1}}_{ADJ}} \\
	&\leq \eps_{\textnormal{step}}(i) + \norm{\ket{\Psi_{i-1}^{\Good}}_{AD}\ket{0}_{J}- \ket{\Phi_{i-1}}_{ADJ}}\leq \sum_{j=1}^i\eps_{\textnormal{step}}(j), \label{eq:lazy-step-derivation}
\end{align}
where we use the triangle inequality and recursively get rid of all queries made by $ \advA $.
The definition of a single step is
\begin{align}\label{eq:lazy-step}
	\eps_{\textnormal{step}}(j):= \norm{\ket{\Psi_j^{\Good}}_{AD}\ket{0}_{J}-  \overline{\Juni}_{R} \CPhO_{\mathcal{Y}} \setminus \V_{R} \Uni_{j-1} \ket{\Psi_{j-1}^{\Good}}_{AD} \ket{0}_{J} }_2.
\end{align}

To calculate the bound on $ \eps_{\textnormal{step}}(j) $ we first calculate how a query affects the good state. The full calculations are presented in section~\ref{sec:lazy-good-query}. Using these findings we prove Lemma~\ref{lem:lazy-step} in section~\ref{sec:bound-step} that states a bound on the norm of the difference of the good and original states.

We define the second part in Eq.~\eqref{eq:lazy-Jnorm} as
\begin{align}
	\eps_{\Find}(i):= \norm{\Juni_{R}\Uni_i\CPhO_{\mathcal{Y}}\setminus \V_{R}\Uni_{i-1}\ket{\Psi^{\Good}_{i-1}}}.
\end{align}
Using the techniques developed to bound $ \eps_{\textnormal{step}}(j) $, we bound $ \eps_{\Find}(i) $ in section~\ref{sec:lazy-bound-find} and state the bounds in Lemma~\ref{lem:lazy-eps-find-bound}.

The final bound is 
\begin{align}
	&\pr{\Find:\advA[\CPhO_{\mathcal{Y}}\setminus R ]}  \leq \sum_{i=1}^q  	\left(\sum_{j=1}^{i-1} \eps_{\textnormal{step}}(j)  + \eps_{\Find}(i)\right)^2,
\end{align}
with Lemma~\ref{lem:lazy-step} and Lemma~\ref{lem:lazy-eps-find-bound} we get the final bound:
\begin{align}
	&\pr{\Find:\advA[\CPhO_{\mathcal{Y}}\setminus R ]} \nonumber\\
	& \leq \sum_{i=1}^q \Bigg(  
	\sum_{j=1}^{i-1} \max_{s\leq j-1}
	\left(
	3\frac{b(s)}{N} +\frac{b(s)}{\sqrt{N(N-b(s))}} + \frac{b(s+1)}{N}  
	\right)\nonumber\\
	&+\max_{s\leq i-1} \left(\sqrt{\frac{b(s+1)}{N}}+\frac{b(s)^{3/2}}{N\sqrt{N-b(s)}}+\frac{\sqrt{b(s)(N-b(s))}}{N}  \right)
	\Bigg)^2 \\
	& \leq \sum_{i=1}^q \Bigg(  
	\sum_{j=1}^{i-1} \max_{s\leq j-1}
	\left(
	5\frac{b(s+1)}{\sqrt{N(N-b(q))}} 
	\right)
	+\max_{s\leq i-1} \left(
	2\sqrt{\frac{b(s+1)}{N}}+\frac{b(s)^{3/2}}{N\sqrt{N-b(q)}}
	\right)
	\Bigg)^2 \\
	& \leq \sum_{i=1}^q \Bigg(  
	\sum_{j=1}^{i-1} 
	5\frac{b(j)}{\sqrt{N(N-b(q))}} 
	+ 	2\sqrt{\frac{b(i)}{N}}+\frac{b(i)^{3/2}}{N\sqrt{N-b(q)}}
	\Bigg)^2 \label{eq:lazy-bound-prfind-1D}
\end{align}
In the above bound we use the facts that $ b(s) $ is a monotonously growing function of $ s $. For our relation $ R_{\coll}\cup R_{\preim} $ we know that $ b(s)=s $. To get a simple bound we note that for real-valued functions $ \sum_{j=1}^{i} f(j) \leq \int_{1}^{i} \mathrm{d}j f(j) $. 
Moreover $ b(s)\leq b(q) $, which we use in the denominator.

Simplifying the above bound and performing the sums we get the claimed result.
	\end{proof}

\subsection{The Good State After a Query}\label{sec:lazy-good-query}
To prove the main technical lemmas of this section we need to analyze how a single query to the oracle affects the good state.

To prove Lemma~\ref{lem:lazy-step} we analyze how far apart the state $ \ket{\Psi_{i-1}^{\Good}} $ is after a query from $ \ket{\Psi_{i}^{\Good}} $. To achieve this goal we inspect in detail the state $ \CPhO_{\mathcal{Y}}\setminus \V_{R}\Uni_{i-1}\ket{\Psi_{i-1}^{\Good}}_{AD}\ket{0}_{J} $. We distinguish different modes of operation: $ \ADD $ when the queried $ x $ is added to $ D $, $ \UPD $ when $ x $ was already in $ D $ and is not removed from the database, $ \REM $ when we remove $ x $ from $ D $, and $\NOT$ where register $ A^Y $ is in state $\ket 0$ . These modes correspond to different branches of superposition in $ \CPhO_{\mathcal{Y}}\setminus \V_{R}\Uni_{i-1}\ket{\Psi_{i-1}^{\Good}}_{AD}\ket{0}_{J} $. We write
\begin{align}
	\Uni_{i-1}\ket{\Psi_{i-1}^{\Good}}_{AD}\ket{0}_{J} =\ket{\xi_{i-1}(\ADD)}+\ket{\xi_{i-1}(\UPD)}+\ket{\xi_{i-1}(\REM)}+\ket{\xi_{i-1}(\NOT)}
\end{align}
and analyze the action of  $\CPhO_{\mathcal{Y}}\setminus \V_{R}$ on the above states separately.

For $\ket{\xi_{i-1}(\NOT)}$ there is no change to the state.
Adding a new entry to a database results in setting the register corresponding to $ x $ to  $ \sum_{y_{s+1}\in[N]}\frac{1}{\sqrt{N}}\omega_{N}^{\eta y_{s+1}} \ket{x,y_{s+1}}$, just like expected from a phase oracle for the uniform distribution. After applying $\Queries^{\dagger}\circ \V_R\circ\Queries $ the state is:
\begin{align}
	& \ADD: \CPhO_{\mathcal{Y}}\setminus \V_R\ket{\xi_{i-1}(\ADD)}\ket{0}_J =
	\sum_{x,\eta,\vec{x},\vec{\eta},w}\alpha_{x,\eta,\vec{x},\vec{\eta},w}\ket{x,\eta}_{A^{XY}}\ket{\psi(x,\eta,\vec{x},\vec{\eta},w)}_{A^W} \nonumber\\
	&\sum_{\vec{y}\not\in \setb(s)}\frac{1}{\sqrt{(N-b(1))\cdots(N-b(s))}}\omega_N^{\vec{\eta}\cdot \vec{y}}\ket{(x_1,y_1),\dots,(x_s,y_{s})}_{D(\vec{x})}\nonumber \\
	&\left( 
	\sqrt{\frac{N-b(s+1)}{N}}
	\underset{\ket{\Psi^{\Good}_{i}(\ADD,s)}}{\underbrace{\sum_{y_{s+1}\not\in\setb(1\mid D(\vec{x})) }\frac{1}{\sqrt{N-b(s+1)}}\omega_N^{\eta y_{s+1}} \ket{x,y_{s+1}}}} \ket{0}_{J} \right.\nonumber\\
	&\left.+ \sqrt{\frac{b(s+1)}{N}}\sum_{y_{s+1}\in\setb(1\mid D(\vec{x})) }\frac{1}{\sqrt{b(s+1)}}\omega_N^{\eta y_{s+1}} \ket{x,y_{s+1}} \ket{1}_{J}	\right)	\nonumber\\
	&\sum_{y_{s+2},\dots,y_q\in[N]}\frac{1}{\sqrt{N^{q-s-1}}}\ket{(\perp,y_{s+2}),\dots,(\perp,y_q)}_{D(\perp)},\label{eq:add-x-to-ss}
\end{align}
where the appropriate position of register $ J $ is after $ D $. By $ \ket{\Psi_{i}^{\Good}(\ADD;s)} $ we mean a state equal to the above state but with just the underlined part in the parentheses. We add $ s $ as the argument to specify the size of the database.

For  $\ket{\xi_{i-1}(\UPD)}$ and $\ket{\xi_{i-1}(\REM)}$, we treat the updated $ x $ as the last one in $ D $, this does not have to be true but it simplifies notation. Note that we want the corresponding $ y_s $ to depend on previous queries but not the other way around, this assumption is without loss of generality as there is no fixed order for $ \sum_{\vec{y}} $. The empty register is moved to the back of $ D $, we do not write it out for simplicity but still consider it done.
\begin{align}
	& \UPD/\REM:  \CPhO_{\mathcal{Y}}\left(\ket{\xi_{i-1}(\UPD)}+\ket{\xi_{i-1}(\REM)}\right)\nonumber\\
	&=\sum_{x,\eta,\vec{x},\vec{\eta},w}\alpha_{x,\eta,\vec{x},\vec{\eta},w}\ket{x,\eta}_{A^{XY}}\ket{\psi(x,\eta,\vec{x},\vec{\eta},w)}_{A^W}\nonumber \\
	&\sum_{\vec{y}\not\in \setb(s-1)}\frac{1}{\sqrt{(N-b(1))(N-b(2))\cdots(N-b(s-1))}}\omega_N^{\vec{\eta}\cdot \vec{y}}\ket{(x_1,y_1),\dots,(x_{s-1},y_{s-1})}_{D(\vec{x}\setminus\{x\})}  \nonumber \\
	&\left( \sum_{y_s\not\in\setb(1\mid D(\vec{x}\setminus\{x\}))} \frac{1}{\sqrt{N-b(s)}} \omega_N^{(\eta_s+\eta)y_s} \ket{x,y_s}_{D(x)} \right. \nonumber\\
	&-\frac{1}{\sqrt{N(N-b(s))}} \sum_{y_s\not\in\setb(1\mid D(\vec{x}\setminus\{x\}))} \omega_N^{(\eta_s+\eta)y_s}\sum_{y'_s\in[N]}\frac{1}{\sqrt{N}} \ket{x,y'_s}_{D(x)} \nonumber\\
	&+\left.\frac{1}{\sqrt{N(N-b(s))}} \sum_{y_s\not\in\setb(1\mid D(\vec{x}\setminus\{x\}))} \omega_N^{(\eta_s+\eta)y_s}\sum_{y'_s\in[N]}\frac{1}{\sqrt{N}} \ket{\perp,y'_s}_{D(x)}\right) \nonumber\\
	& \sum_{y_{s+1},\dots,y_q\in[N]}\frac{1}{\sqrt{N^{q-s}}}\ket{(\perp,y_{s+1}),\dots,(\perp,y_q)}_{D(\perp)}. \label{eq:upd/rem-x-to-ss}
\end{align}
Whether we are in the branch $ \UPD $ or $ \REM $ depends on whether $ \eta=-\eta_s $ or not. 

When the database is updated we have the following state after the query:
\begin{align}
	& \UPD:  \CPhO_{\mathcal{Y}}\setminus\V_R\ket{\xi_{i-1}(\UPD)}\ket{0}_J
	=\sum_{x,\eta,\vec{x},\vec{\eta},w}\alpha_{x,\eta,\vec{x},\vec{\eta},w}\ket{x,\eta}_{A^{XY}}\ket{\psi(x,\eta,\vec{x},\vec{\eta},w)}_{A^W}\nonumber \\
	&\sum_{\vec{y}\not\in \setb(s-1)}\frac{1}{\sqrt{(N-b(1))\cdots(N-b(s-1))}}\omega_N^{\vec{\eta}\cdot \vec{y}}\ket{(x_1,y_1),\dots,(x_{s-1},y_{s-1})}_{D(\vec{x}\setminus\{x\})}  \nonumber \\
	&\left(\underset{\ket{\Psi_j^{\Good}(\UPD;s)} }{\underbrace{ \sum_{y_s\not\in\setb(1\mid D(\vec{x}\setminus\{x\}))} \frac{1}{\sqrt{N-b(s)}} \omega_N^{(\eta_s+\eta)y_s} \ket{x,y_s}_{D(x)}}} \ket{0}_J \right. \nonumber\\
	&- \underset{{\color{red} \ket{\Psi_{i,1}^{\Bad}(\UPD;s)} }}{\underbrace{\frac{1}{\sqrt{N(N-b(s))}} \sum_{y_s\in\setb(1\mid D(\vec{x}\setminus\{x\}))} \omega_N^{(\eta_s+\eta)y_s}\sum_{y'_s\in[N]}\frac{1}{\sqrt{N}} \ket{\perp,y'_s}_{D(x)}}}\ket{0}_J
	\nonumber\\
	&+\underset{ {\color{red} \ket{\Psi_{i,2}^{\Bad}(\UPD;s)}} }{\underbrace{\frac{1}{N} \sum_{y_s\in\setb(1\mid D(\vec{x}\setminus\{x\}))} \omega_N^{(\eta_s+\eta)y_s}\sum_{y'_s\not\in\setb(1\mid D(\vec{x}\setminus\{x\}))}\frac{1}{\sqrt{N-b(s)}} \ket{x,y'_s}_{D(x)} }}\ket{0}_J
	\nonumber\\
	&\left.		
	+\sqrt{\frac{b(s)}{N^2(N-b(s))}} \sum_{y_s\in\setb(1\mid D(\vec{x}\setminus\{x\}))} \omega_N^{(\eta_s+\eta)y_s}\sum_{y'_s\in\setb(1\mid D(\vec{x}\setminus\{x\}))}\frac{1}{\sqrt{b(s)}} \ket{x,y'_s}_{D(x)} \ket{1}_J	\right) \nonumber\\
	& \sum_{y_{s+1},\dots,y_q\in[N]}\frac{1}{\sqrt{N^{q-s}}}\ket{(\perp,y_{s+1}),\dots,(\perp,y_q)}_{D(\perp)}. \label{eq:upd-x-to-ss}
\end{align}
In the above state we have simplified the sum $  \sum_{y_s\not\in\setb(1\mid D(\vec{x}\setminus\{x\}))} =- \sum_{y_s\in\setb(1\mid D(\vec{x}\setminus\{x\}))}$. Register $ J $ is supposed to be placed after $ D $, for the sake of presentation though, we put it in the middle.
By $ \ket{\Psi_{i}^{\Good}(\UPD;s)} $, $ \ket{\Psi_{i,1}^{\Bad}(\UPD;s)} $, and $ \ket{\Psi_{i,2}^{\Bad}(\UPD;s)} $ we mean the whole state with just the underlined states in the parentheses  equals the given state. We add $ s $ as the argument to specify the size of the database.

After removing an element from the database we have:
\begin{align}
	& \REM:  \CPhO_{\mathcal{Y}}\setminus\V_{R}\ket{\xi_{i-1}(\REM)}\ket{0}_J
	=\sum_{x,\eta,\vec{x},\vec{\eta},w}\alpha_{x,\eta,\vec{x},\vec{\eta},w}\ket{x,\eta}_{A^{XY}}\ket{\psi(x,\eta,\vec{x},\vec{\eta},w)}_{A^W}\nonumber \\
	&\sum_{\vec{y}\not\in \setb(s-1)}\frac{1}{\sqrt{(N-b(1))\cdots(N-b(s-1))}}\omega_N^{\vec{\eta}\cdot \vec{y}}\ket{(x_1,y_1),\dots,(x_{s-1},y_{s-1})}_{D(\vec{x}\setminus\{x\})}  \nonumber \\
	&\left(\sqrt{\frac{N-b(s)}{N}} \underset{\ket{\Psi^{\Good}_{i}(\REM,s)}}{\underbrace{\sum_{y_s\in[N]}\frac{1}{\sqrt{N}} \ket{\perp,y_s}_{D(x)} }}\ket{0}_J \right. \nonumber\\
	&+ \underset{{\color{red} \ket{\Psi^{\Bad}_{i}(\REM;s)} }} {\underbrace{\frac{b(s)}{N} \sum_{y_s\not\in\setb(1\mid D(\vec{x}\setminus\{x\}))} \frac{1}{\sqrt{N-b(s)}}  \ket{x,y_s}_{D(x)}}} \ket{0}_J  \nonumber\\
	&\left.-\frac{\sqrt{b(s)(N-b(s))}}{N} \sum_{y_s\in\setb(1\mid D(\vec{x}\setminus\{x\}))}\frac{1}{\sqrt{b(s)}} \ket{x,y_s}_{D(x)} \ket{1}_J	\right) \nonumber\\
	& \sum_{y_{s+1},\dots,y_q\in[N]}\frac{1}{\sqrt{N^{q-s}}}\ket{(\perp,y_{s+1}),\dots,(\perp,y_q)}_{D(\perp)}. \label{eq:rem-x-to-ss}
\end{align}

\subsection{Bound on $ \eps_{\textnormal{step}} $}\label{sec:bound-step}
We want to show that after any query, $ \ket{\Phi_i}_{ADJ} $ is close to $ \ket{\Psi_i^{\Good}}_{AD}\ket{0}_{J} $. 
One way of looking at the lemma below is from the perspective of an adversary searching for inputs that provide outputs of a random function that are in $ R $. Normally this task does not involve a punctured oracle but a regular one. We show here the error introduced by puncturing the oracle; The two states that we consider come from projecting with $ \overline{\Juni}_R $ either the state after interacting with a non-punctured oracle or the state after interacting with a punctured oracle (given $ \neg\Find $).
This intuition, however, is not crucial for our proof, as we focus solely on punctured oracles.
\begin{lemm} \label{lem:lazy-step}
	For states defined in the preceding sections we have
	\begin{align}
		&\norm{\ket{\Psi_i^{\Good}}_{AD}\ket{0}_{J}- \ket{\Phi_i}_{ADJ}}\leq \sum_{j=1}^i \eps_{\textnormal{step}}(j)\nonumber\\
		& \leq \sum_{j=1}^i  \max_{s\leq j-1}
		\left(
		3\frac{b(s)}{N}	+\frac{b(s)}{\sqrt{N(N-b(s))}}  +\frac{b(s+1)}{N} 
		\right). 
	\end{align}
\end{lemm}

\begin{proof}
		We are going to prove the statement by recursion over the number of queries made by the adversary. The exact derivation  is shown in Equation~\eqref{eq:lazy-step-derivation}.		
		We are going to prove the statement by recursion over the number of queries made by the adversary. 
		
	In the following we calculate $ \eps_{\textnormal{step}}(j) $ defined in Equation~\eqref{eq:lazy-step}. For $ i=0 $ the statement is true, as $ \ket{\Psi_0^{\Good}}\ket{0}_{J} = \ket{\Phi_0}=\ket{\Psi_0}\ket{0}_{J} $. 
	
	From Eqs.~\eqref{eq:add-x-to-ss}, \eqref{eq:upd-x-to-ss}, and \eqref{eq:rem-x-to-ss} we know how querying works for $ \ket{\Psi^{\Good}_{j-1}} $, now we distinguish two types of errors compared to $ \ket{\Psi^{\Good}_j}\ket{0}_J $: an additive error of adding a small-weight state to the original one and a multiplicative error where one branch of the superposition is multiplied by some factor.
	
	The additive error includes all states of small-weight states multiplied by $ \ket{0}_{J} $ with the superscript $ \Bad $.
	In the branches of the superposition where we add a new entry to the database we see that we recover $ \ket{\Psi^{\Good}_{j}} \ket{0}_{J}$  after multiplying a branch of $ \CPhO_{\mathcal{Y}}\setminus\V_{R}\Uni_{j-1}\ket{\Psi^{\Good}_{j-1}}\ket{0}_{J} $ by $ \sqrt{\frac{N-b(s+1)}{N}}  $ (Eq.~\eqref{eq:add-x-to-ss}) or by $ \sqrt{\frac{N-b(s)}{N}} $ (Eq.~\eqref{eq:rem-x-to-ss}).
	
	Our approach to the rest of the proof consists of first dealing with the additive and later with the multiplicative error. To this end let us define $ \ket{\psi^{\times}_{j}}_{ADJ} $ as the state $\overline{\Juni}_{R} \CPhO_{\mathcal{Y}}\setminus\V_{R} \Uni_{j-1}\ket{\Psi^{\Good}_{j-1}}\ket{0}_{J} $ with all branches classified as the additive error excluded. By ``classified as the additive error'' we mean states with superscript $ \Bad $ and highlighted in red in Equations~(\ref{eq:add-x-to-ss}, \ref{eq:upd-x-to-ss}, \ref{eq:rem-x-to-ss}).
	The new state is defined as
	\begin{align}
		&\ket{\psi^{\times}_{j}}_{ADJ} := \left(\sum_{s} \ket{\Psi^{\Good}_{j}(\NOT;s)} \right. 
		+ \sqrt{\frac{N-b(s+1)}{N}}   \ket{\Psi^{\Good}_{j}(\ADD;s)}\nonumber\\
		& +  \ket{\Psi^{\Good}_{j}(\UPD;s)}  \left.+\sqrt{\frac{N-b(s)}{N}}  \ket{\Psi^{\Good}_{j}(\REM;s)} \right)\ket{0}_J,
	\end{align}
	where the states above correspond to branches of superposition where we do nothing ($ \NOT $, for $ \eta=0 $), add an entry, update the database, and remove an entry from $ D $.
	Bounding the difference of the states is done as follows
	\begin{align}
		&\norm{\ket{\Psi^{\Good}_{j}} \ket{0}_{J}-\overline{\Juni}_{R}\CPhO_{\mathcal{Y}}\setminus\V_{R}\Uni_{j-1}\ket{\Psi^{\Good}_{j-1}}\ket{0}_{J} }\nonumber \\
		&\leq \norm{\ket{\Psi^{\Good}_{j}} \ket{0}_{J}-  \ket{\psi^{\times}_{j}}_{ADJ}}+\norm{ \ket{\psi^{\times}_{j}}_{ADJ}-\overline{\Juni}_{R}\CPhO_{\mathcal{Y}}\setminus\V_{R}\Uni_{j-1}\ket{\Psi^{\Good}_{j-1}}\ket{0}_{J} }. \label{eq:bound-add-or-mul}
	\end{align}
	The second term above is just the norm of all states amplifying the additive error---we call them the bad states.
	
	We bound the additive error $ \| \ket{\psi^{\times}_{j}}_{ADJ} -\overline{\Juni}_{R} \CPhO_{\mathcal{Y}} \setminus\V_{R} \Uni_{j-1}\ket{\Psi^{\Good}_{j-1}}\ket{0}_{J} \| $ by first splitting the three cases underlined above:		
	\begin{align}
		\norm{\ket{\Psi^{\Bad}_{j}}}  \leq \norm{\ket{\Psi^{\Bad}_{j,1}(\UPD)}} +\norm{\ket{\Psi^{\Bad}_{j,2}(\UPD)}} +\norm{\ket{\Psi^{\Bad}_{j}(\REM)}}, \label{eq:lazy-three-errors}
	\end{align}
	where $ \ket{\Psi^{\Bad}_{j}} $ is the sum of all three bad states, the bound follows from the triangle inequality.
	
	Calculating all of the three norms above is done by first focusing on a particular interface that is queried and by focusing on particular sizes of databases:
	\begin{align}\label{eq:lazy-sum-betas-bound}
		\norm{\ket{\Psi^{\Bad}_{j}}} = \sqrt{\sum_{s=0}^j \abs{\beta(s)}^2 \norm{\ket{\Psi^{\Bad}_{j}(s)}}^2 },
	\end{align}
	where $ \beta(s) $ is the amplitude of the good state projected to states with the specified parameter: For a projector $ \Puni_{s} $ to adversaries that query databases of size $ s $ we have $ \beta(s):=\Puni_{s} \ket{\Psi^{\Good}_j} $ and $ \ket{\Psi^{\Bad}_{j}(s)}:=\Puni_{s} \ket{\Psi^{\Bad}_{j}} $.

		\parag{Additive errors}
	Dealing with additive errors, we begin with the $ \UPD $ branch. In the bad states in the $ \UPD $ case, Eq.~\eqref{eq:upd-x-to-ss}, we need to take special care of $ \sum_{y_{s}\in\setb(1\mid D(\vec{x}\setminus\{x\}))}\omega_{N}^{(\eta_{s}+\eta)y_{s}} $; This is a a complex number that depends on $ \eta_{s} $, so it enters the norm in a non-trivial way. 
	The first step is a change of variables: Instead of summing over elements of of the bad state we sum over $ y_{s}\in[b(s)] $ and change $ y_s $ in the expression to $\setb(1\mid D(\vec{x}\setminus\{x\}))(y_{s})$, by which we denote the $ y_{s} $-th element of $ \setb(1\mid D(\vec{x}\setminus\{x\})) $. Note that there is a natural order in the bad set, as $ \mathcal{Y}=[N] $.
	
	Given the change of variables we can use the triangle inequality to focus on the norm of a state with a single phase factor $ \omega_{N}^{(\eta_{s}+\eta)\setb(1\mid D(\vec{x}\setminus\{x\}))(y_{s})} $, instead of the whole sum:
	\begin{align}\label{eq:lazy-y-bad-bound}
		&\norm{\ket{\Psi^{\Bad}_{j}(\UPD;s)}} \leq  \sum_{y_{s}\in[b(s)]} \norm{\ket{\Psi^{\Bad}_{j}(\UPD;s,\setb(1\mid D(\vec{x}\setminus\{x\}))(y_{s}))}},
	\end{align}
	where we omit the index of the $ \UPD $ errors because the techniques here work in almost the same way for both states. The input $ D(\vec{x}\setminus\{x\}) $ should not be treated as an actual argument of the state, we still consider the superposition over different inputs, we just mean that in the state $ \ket{\Psi^{\Bad}_{j}(\UPD;s)} $ we change the variable $ y_s $.	
	In what follows we denote the state on the right hand side of the above equation by  $ \ket{\Psi^{\Bad}_{j}(\UPD;s,\setb'(y_{s}))} $. 
	
	Now we focus on the state with a fixed $ \setb'(y_{s}) $, we bound the norm of this state.
	\begin{nrclaim}\label{claim:lazy-setb-y-norm}
		For all $ y_{s}\in[b(s)] $
		\begin{align}\label{eq:lazy-sety-norm}
			&\norm{\ket{\Psi^{\Bad}_{j,1}(\UPD;s,\setb(1\mid D(\vec{x}\setminus\{x\}))(y_{s}))}} \leq \frac{1}{\sqrt{N(N-b(s))}} \;\textnormal{ and}\\
			&\norm{\ket{\Psi^{\Bad}_{j,2}(\UPD;s,\setb(1\mid D(\vec{x}\setminus\{x\}))(y_{s}))}} \leq \frac{1}{N}.
		\end{align}
	\end{nrclaim}
	\begin{proof}
		Our idea for the proof is to first show that the norm of a good state in the $ \UPD $ branch with a modified sum over $ y_s $ is not greater than $ 1 $. Then to prove that the norm of $ \ket{\Psi^{\Bad}_{j}(\UPD;s,\setb(1\mid D(\vec{x}\setminus\{x\}))(y_{s}))} $ multiplied by the corresponding right hand side of Eq.~\eqref{eq:lazy-sety-norm} equals the norm of the good state we mentioned earlier.
		
		We start by defining two states:
		\begin{align}
			& \sum_{x,\eta,\vec{x},\vec{\eta},w} \alpha_{x,\eta,\vec{x},\vec{\eta},w}\ket{x,\eta}_{A^{XY}} \ket{\psi(x,\eta,\vec{x},\vec{\eta},w)}_{A^W} \nonumber \\
			&\sum_{\vec{y}\not\in \setb(s-1)}\frac{1}{\sqrt{(N-b(1))\cdots(N-b(s-1))}}\omega_N^{\vec{\eta}\cdot \vec{y}}\ket{(x_1,y_1),\dots,(x_{s-1},y_{s-1})}_{D(\vec{x}\setminus\{x\})}  \nonumber \\
			& \sum_{y_{s+1},\dots,y_q\in[N]}\frac{1}{\sqrt{N^{q-s}}}\ket{(\perp,y_{s+1}),\dots,(\perp,y_q)}_{D(\perp)}\nonumber\\
			&\otimes \begin{cases}
				\sum_{y_s\in\setb(1\mid D(\vec{x}\setminus\{x\}))} \frac{1}{\sqrt{b(s)}} \omega_N^{(\eta_s+\eta)y_s} \ket{x,y_s}_{D(x)} =: \ket{\overline{\Psi}^{\Good}_{j}(\UPD;s)} \\ 
				\sum_{y_s\in[N]} \frac{1}{\sqrt{N}} \omega_N^{(\eta_s+\eta)y_s} \ket{x,y_s}_{D(x)} =: \ket{\widetilde{\Psi}^{\Good}_{j}(\UPD;s)}
			\end{cases} .
			\label{eq:lazy-bar/tilde-good}
		\end{align}
		The first one, $ \ket{\overline{\Psi}^{\Good}_{j}(\UPD;s)}  $ is the one that we use in the last step of the proof, as described in the previous paragraph. The second one will be used to show that the norm of $ \ket{\overline{\Psi}^{\Good}_{j}(\UPD;s)}  $ is bounded by $ 1 $.
		
		One more introductory statement that we need to prove is that $ \norm{\ket{\widetilde{\Psi}^{\Good}_{j}(\UPD;s)}}\leq 1 $. To this end let us remind ourselves that the good state is a state interacting with the not-punctured oracle for $ j $ queries, projected to databases that are not in $ R $, and normalized. Let us consider a projection that just omits register $ D(x) $ when bringing $ D $ to be not in $ R $. Using this latter projection on a state interacting with the not-punctured oracle results in the state $\ket{\widetilde{\Psi}^{\Good}_{j}(\UPD;s)}$. Hence $ \norm{\ket{\widetilde{\Psi}^{\Good}_{j}(\UPD;s)}}\leq 1 $, just like $ \norm{\ket{\Psi^{\Good}_{j}(\UPD;s)}}\leq 1 $. The inequality comes from excluding a single branch of the superposition in $\ket{\widetilde{\Psi}^{\Good}_{j}(s)}$.

		The fact that the state with $ \sum_{y_{s}\in[N]} $  is sub-normalized is important because now we can bound the norm of $ \ket{\overline{\Psi}^{\Good}_{j}(\UPD;s)} $. Having in mind that $ \sum_{y_{s}\in\setb(1\mid D(\vec{x}\setminus\{x\})) }= \sum_{y_{s}\in[N]}-\sum_{y_{s}\not\in\setb(1\mid D(\vec{x}\setminus\{x\})) }$ we see that
		\begin{align}
			&b(s)\norm{\ket{\overline{\Psi}^{\Good}_{j}(\UPD;s)}}^2 \nonumber\\
			&=  N\norm{ \ket{\widetilde{\Psi}^{\Good}_{j}(\UPD;s)} }^2 -  (N-b(s)) \norm{  \ket{\Psi^{\Good}_{j}(\UPD;s)} }^2 \leq b(s),
		\end{align}
		hence $ \norm{\ket{\overline{\Psi}^{\Good}_{j}(\UPD;s)}}^2\leq 1 $. 

			Now that we know that $ \ket{\overline{\Psi}^{\Good}_{j}(\UPD;s)} $ is sub-normalized we show that
			\begin{align}
				\norm{\ket{\overline{\Psi}^{\Good}_{j}(\UPD;s,\setb'(y{s}))}}\leq \frac{1}{\sqrt{b(s)}} .
			\end{align} 
			To prove this bound, consider measuring register $ D_a(x) $ of $ \ket{\overline{\Psi}^{\Good}_{j}(\UPD;s)}  $ in the computational basis. The probability of getting any outcome $ y_{s} $ is necessarily $ \frac{1}{b(s)} $, as the outputs of the oracle are uniformly random. The post-measurement state, for an outcome $ y_{s} $, is $\sqrt{b(s)} \cdot\ket{\overline{\Psi}^{\Good}_{j}(\UPD;s,\setb'(y_{s}))} $. Naturally, norm of this post-measurement state is at most $ 1 $.
			
		Now we can use the state $ \ket{\overline{\Psi}^{\Good}_{j}(\UPD;s,\setb'(y_{s}))} $ to analyze the norm of $ \ket{\Psi^{\Bad}_{j}(\UPD;s,\setb'(y_{s}))}$. First let us inspect the norm squared of the bad state:		
		\begin{align}\label{eq:lazy-psi-bad-norm2}
			&\norm{\ket{\Psi^{\Bad}_{j}(\UPD;s,\setb'(y_{s}))}}^2 =  \sum_{x,\eta,\vec{x},\vec{\eta}',\vec{\eta},w',w} \sum_{\eta_{s}',\eta_{s}} \bar{\alpha}'_{x,\eta,\vec{x},\vec{\eta}',\eta_{s}',w'}  \alpha'_{x,\eta,\vec{x},\vec{\eta},\eta_{s},w}\nonumber\\
			& \braket{\psi(x,\eta,\vec{x},\vec{\eta}',\eta'_{s},w')}{\psi(x,\eta,\vec{x},\vec{\eta},\eta_{s},w)} \nonumber \\
			&\sum_{\vec{y}\not\in \setb(s-1) } \frac{1}{(N-b(1))\cdots(N-b(s))}\bar{\omega}_{N}^{\vec{\eta}'\cdot \vec{y}} \omega_{N}^{\vec{\eta}\cdot \vec{y}} \nonumber\\
			& \frac{1}{N^2(N-b(s))}\bar{\omega}_{N}^{(\eta'_{s}+\eta) \setb'(y_{s})} \omega_{N}^{(\eta_{s}+\eta) \setb'(y_{s})} \underset{=\nu}{\underbrace{\sum_{y'_{s}\in[\nu]}   }},
		\end{align}
		where $ \nu=N $ for $ \ket{\Psi^{\Bad}_{j,1}(\UPD;s,\setb'(y_{s}))} $ and $ \nu=N-b(s) $ for $ \ket{\Psi^{\Bad}_{j,2}(\UPD;s,\setb'(y_{s}))} $ (in the second case the sum goes over $ y'_s\not \in\setb(1\mid D(\vec{x}\setminus\{x\})) $ instead of $ [\nu] $).
		It is easy to notice, that the only difference between Eq.~\eqref{eq:lazy-psi-bad-norm2} and norm squared of $ \ket{\overline{\Psi}^{\Good}_{j}(\UPD;s,\setb'(y_{s}))} $ lies in the factor $ \frac{\nu}{N^2(N-b(s))} $. This factor in the modified good state equals $ \frac{1}{b(s)} $.
		This observation implies that 
		\begin{align}\label{eq:lazy-psi-is-psioverline}
			\norm{ \ket{\Psi^{\Bad}_{j}(\UPD;s,\setb'(y_{s}))}}= \sqrt{\frac{b(s)\cdot\nu}{N^2(N-b(s))} } \norm{\ket{\overline{\Psi}^{\Good}_{j}(\UPD;s, \setb'(y_{s}))} }.
		\end{align}
		Together with the bound on the norm in the left hand side this proves the claimed bounds.
	\end{proof}
	
	Claim~\ref{claim:lazy-setb-y-norm}, together with the bound from Eq.~\eqref{eq:lazy-y-bad-bound} gives us:
	\begin{align}\label{eq:lazy-upd-bounds-s}
		&\norm{\ket{\Psi^{\Bad}_{j,1}(\UPD;s)}} \leq   
		\frac{b(s)}{\sqrt{N(N-b(s))}}, \\
		&\norm{\ket{\Psi^{\Bad}_{j,2}(\UPD;s)}} \leq   \frac{b(s)}{N} .
	\end{align}
	
	The bounds from Eq.~\eqref{eq:lazy-upd-bounds-s} in Eq.~\eqref{eq:lazy-sum-betas-bound} give us the bound on the additive error in the $ \UPD $ branch. 
	The additive error for the $ \REM $ branch ( $ \ket{\Psi^{\Bad}_{j}(\REM)} $ in Eq.~\eqref{eq:rem-x-to-ss} ) is much easier to calculate: As register $ D(x) $ is normalized and all the rest of the state is the same as $ \ket{\Psi^{\Good}_{j}(\REM)} $, the only error comes from the factor $ \frac{b(s)}{N} $. To calculate the norm of the state we can follow the analysis of Eq.~\eqref{eq:lazy-psi-bad-norm2}. Finally we get:
	\begin{align}\label{eq:lazy-additive-errors}
		&\norm{\ket{\Psi^{\Bad}_{j,1}(\UPD)}}  \leq \max_{s}\left(
		\frac{b(s)}{\sqrt{N(N-b(s))}} \right) , \\
		&\norm{\ket{\Psi^{\Bad}_{j,2}(\UPD)}}  \leq   \max_{s}\left(\frac{b(s)}{N} \right), \\
		& \norm{\ket{\Psi^{\Bad}_{j}(\REM)}} \leq \max_{s}\left( \frac{b(s)}{N} \right),
	\end{align}
	where $ s\leq j-1 $.
	
	\parag{Multiplicative errors}
	The multiplicative error is a factor that multiplies a part of the state $ \ket{\psi^{\times}_{j}}_{ADJ} $. Similarly as before we need to take care of the fact that the joint state of the adversary and the oracle is a sum over databases of different sizes and queries to different interfaces:
	\begin{align}\label{eq:lazy-state-sum-s}
		&\ket{\psi^{\times}_{j}}=\sum_{s}\ket{\psi^{\times}_{j}(s)},
	\end{align}
	where the states $ \ket{\psi^{\times}_{j}(s)} $ are orthogonal. The above is also true for $ \ket{\Psi^{\Good}_j}=\sum_{s}\ket{\Psi^{\Good}_j(s)} $.
	
	There are two sources of multiplicative errors, $ \ADD $ from Eq.~\eqref{eq:add-x-to-ss} and $ \REM $ from Eq.~\eqref{eq:rem-x-to-ss}, we split the two sources with the triangle inequality. We deal with both in the same way, just the final bound is different. 
	
	Let us write down the two parts, one affected by the error and the second not:
	\begin{align}
		&\ket{\Psi_{j}^{\Good}}_{AD}\ket{0}_{J}=\sum_{s}\alpha(s)\ket{\varphi_1(s)}+ \beta(s) \ket{\varphi_2(s)},\\
		&\ket{\psi^{\times}_{j}}_{ADJ}=\sum_{s}\alpha(s)\ket{\varphi_1(s)}+\sqrt{1-e} \beta(s)\ket{\varphi_2(s)}, \label{eq:lazy-varphi12}
	\end{align}
	where $ \sqrt{1-e} $ is the multiplicative error, in the case $ \ADD $ the error is $ e=\frac{b(s+1)}{N} $ and $ e=\frac{b(s)}{N} $ in the case $ \REM $.
	We know that $ \sum_{s}\abs{\alpha(s)}^2+\abs{\beta(s)}^2\leq1 $, because we excluded a single branch of the superposition, for $ \ADD $ and $ \REM $. This inequality implies  $ \sum_{s}\abs{\beta(s)}^2\leq 1 $. 
	We continue with the bound
	\begin{align}
		&\norm{	\ket{\psi^{\times}_{j}}_{ADJ}-\ket{\Psi_{j}^{\Good}}_{AD}\ket{0}_{J}}
		=\norm{\sum_{s} (1-\sqrt{1-e})\beta(s)\ket{\varphi_2(s)}}\\
		&= \sqrt{\sum_{s}  (1-\sqrt{1-e})^2\abs{\beta(s)}^2}\leq \max_{s}\{  1- \sqrt{1-e} \}  \leq \max_{s}\{   e \},
		\label{eq:lazy-bound-multip}
	\end{align}	
	Maximization is done over $ s\leq j-1 $.
	
	\parag{Bound on one step}
	From Eqs.~\eqref{eq:bound-add-or-mul}, \eqref{eq:lazy-additive-errors}, and \eqref{eq:lazy-bound-multip} (for the two sources of error) the bound on the single step is
	\begin{align}
		&\eps_{\textnormal{step}}(j) \leq  \max_{s\leq j-1}\left(
		\frac{b(s)}{\sqrt{N(N-b(s))}} + \frac{b(s)}{N}
		+2 \frac{b(s)}{N}
		+ \frac{b(s+1)}{N}
		\right)
		\label{eq:lazy-eps-bound}
	\end{align}	
	and the final bound is
	\begin{align}\label{eq:lazy-good-vs-phi-bound}
		&\norm{\ket{\Psi_i^{\Good}}_{AD}\ket{0}_{J}- \ket{\Phi_i}_{ADJ}} \leq \sum_{j=1}^i \max_{s\leq j-1}\left(
		\frac{b(s)}{\sqrt{N(N-b(s))}} + \frac{b(s)}{N}
		+2 \frac{b(s)}{N}
		+ \frac{b(s+1)}{N}
		\right)
	\end{align}	
\end{proof}

The bound from Lemma~\ref{lem:lazy-step} can be further simplified to 
	\begin{align}\label{eq:lazy-good-vs-phi-bound-simple}
	&\norm{\ket{\Psi_i^{\Good}}_{AD}\ket{0}_{J}- \ket{\Phi_i}_{ADJ}} \leq 5\sum_{j=1}^i \max_{s\leq j-1}\left(
	\frac{b(s+1)}{\sqrt{N(N-b(q))}}
	\right)	,
\end{align}
Where in the denominator we use $ b(s)\leq b(q) $, which is true for the relations considered in this paper.

\subsection{Bound on $ \eps_{\Find} $}\label{sec:lazy-bound-find}
Our task here is bounding the norm of $ \norm{\Juni_{R}\Uni_i\CPhO_{\mathcal{Y}}\setminus \V_{R}\Uni_{i-1}\ket{\Psi^{\Good}_{i-1}}} $. All states (among the states defined in section~\ref{sec:lazy-good-query}) that give non-zero contributions to this norm are the ones that we give the superscript $ \Find $, they contain $ \ket{1}_J $.

\begin{lemm}\label{lem:lazy-eps-find-bound}
	For states defined in preceding sections we have
	\begin{align}
		&\norm{\Juni_{R}\Uni_i\CPhO_{\mathcal{Y}}\setminus \V_{R}\Uni_{i-1}\ket{\Psi^{\Good}_{i-1}}}=\eps_{\Find}(i) \nonumber\\
		& \leq \max_{s\leq i-1} \left(
		\sqrt{\frac{b(s+1)}{N}}
		+\frac{b(s)^{3/2}}{N_a\sqrt{N-b(s)}}
		+\frac{\sqrt{b(s)(N-b(s))}}{N}  \right).
	\end{align}
\end{lemm}
\begin{proof}
	For all states multiplied by $ \ket{1}_J $ we start bounding the norm by splitting the norm by the size of the database, like in Eq.~\eqref{eq:lazy-sum-betas-bound}. Let us now go through the three important modes of operation, i.e.\ adding, updating, or removing from the database.
	
	\parag{The $ \ADD $ case}
	The bound on the norm of the state in $ R $ in this case is:
	\begin{align}
		 \norm{\Juni_{R}\Uni_i\CPhO_{\mathcal{Y}}\setminus \V_{R}\Uni_{i-1}\ket{\Psi^{\Good}_{i-1}(\ADD;s)}}\leq \max_{s} \sqrt{\frac{b(s+1)}{N}}.
	\end{align}
	This bound holds , because except for the factor in front of the state and register $ D(x) $ the state is just a good state (one from just before the query we analyze in Eq.~\eqref{eq:add-x-to-ss}). Moreover register $ D(x) $ is normalized (given the fact that $ \eta $ is explicit in the adversary's register).
	
	\parag{The $ \UPD $ case}
	In this case we have a bound of
	\begin{align}
		\norm{\Juni_{R}\Uni_i\CPhO_{\mathcal{Y}}\setminus \V_{R}\Uni_{i-1}\ket{\Psi^{\Good}_{i-1}(\UPD;s)}}
		\leq  \max_{s} \frac{b(s)^{3/2}}{N\sqrt{N-b(s)}},
	\end{align}
	where we follow the same reasoning as in the proof of Lemma~\ref{lem:lazy-step} and Claim~\ref{claim:lazy-setb-y-norm}. 
	
	\parag{The $ \REM $ case}
	Finally we have a bound of
	\begin{align}
		 \norm{\Juni_{R}\Uni_i\CPhO_{\mathcal{Y}}\setminus \V_{R}\Uni_{i-1}\ket{\Psi^{\Good}_{i-1}(\REM;s)}}\leq \max_{s} \frac{\sqrt{b(s)(N-b(s))}}{N} 
	\end{align}
	and to get it we follow the same reasoning as for the $ \ADD $ case.
	
We use these bounds and the triangle inequality to bound the second term in Eq.~\eqref{eq:lazy-Jnorm}:
\begin{align}\label{eq:lazy-J-good-bound}
	&\norm{\Juni_{R}\Uni_i\CPhO_{\mathcal{Y}}\setminus \V_{R}\Uni_{i-1}\ket{\Psi^{\Good}_{i-1}}} \nonumber\\
	&\leq \max_{s\leq i-1} \left(
	\sqrt{\frac{b^3(s)}{N^2(N-b(s))}}  + \frac{\sqrt{b(s)(N-b(s))}}{N} 	+\sqrt{\frac{b(s+1)}{N}} \right)\\
	&\leq \max_{s\leq i-1} \left(
	\sqrt{\frac{b^3(s)}{N^2(N-b(s))}}+2\sqrt{\frac{b(s+1)}{N}} \right).
\end{align}

\end{proof}

\subsection{Other Relations}
For $ R_{\coll} $ we use eq.~\eqref{eq:lazy-eps-bound} with $ b(i)=i-1 $ instead of $ b(i)=i $. The bound on the probability of the event $ \Find $ is
\begin{align}\label{eq:prfind-coll}
	\PR[\Find:\advA[\CStO_{\mathcal{Y}}\setminus R_{\coll}]]\leq 
	\frac{2 q^2}{N}+\frac{4 q^{7/2}}{N\sqrt{N-q}}+ 
	\frac{5 q^5}{N (N-q)}.
\end{align}

For $ R_{\preim} $ in eq.~\eqref{eq:lazy-eps-bound} we set a constant $ b(j)=1 $. The bound on the probability of $ \Find $ is then
\begin{align}\label{eq:prfind-preim}
	\PR[\Find:\advA[\CStO_{\mathcal{Y}}\setminus R_{\preim}]]\leq 
	\frac{9 q}{N}+	\frac{30 q^2 }{N\sqrt{N-1}}+\frac{25 q^3 }{ N(N-1) }.
\end{align}

\section{Additional Details on Quantum-Accessible Oracles}\label{sec:cfo-details}

\subsection{Example Non-Uniform Distributions}
The most important distribution that can be quantumly lazy sampled is the uniform distribution. It was first shown in \cite{zhandry2018record} how to do that. We present a lot of details and intuitions on this matter in the rest of this section. 

Let us say we want to efficiently simulate a quantum oracle for a random function $h:\bits{m}\to\bits{}$, such that $h(x)=1$ with probability $\lambda$. Then the adding function of the corresponding compressed oracle is $\forall x\in \bits{m}$:
\begin{align}
	\Samp_{\lambda}(x):= \left( \begin{array}{cc}
		\sqrt{1-\lambda} & \sqrt{\lambda} \\
		\sqrt{\lambda} & -\sqrt{1-\lambda}
	\end{array} \right),
\end{align}
independent from any previous queries.
This observation comes in useful in tasks like search in a sparse database.

\subsection{Uniform Oracles}\label{sec:uniform-oracles}
For ease of exposition, and to highlight the connection to the formalism in \cite{zhandry2018record}, we present a discussion of compressed oracles with \emph{uniform oracles} that model functions sampled uniformly at random from $\mathcal{F}:=\left\{f:\bits{m}\to\bits{n}\right\}\symbolindexmark{\bits}$. A complete formal treatment of the uniform case, including applications, can be found in \cite{Unruh-forthcoming}. 

We denote the uniform distribution over $\mathcal{F}$ by $\distrU\symbolindexmark{\distrU}$. The cardinality of the set of functions is $|\mathcal{F}|=2^{n2^m}$ and the truth table of any $f\in\mathcal{F}$ can be represented by $2^m$ rows of $n$ bits each. Uniform oracles are the most studied in the random-oracle model and are also analyzed in \cite{zhandry2018record}.

The transformation we use in the case of uniformly sampled functions is the Hadamard transform. The unitary operation to change between types of oracles is defined as $ \symbolindexmark{\HT} $
\begin{equation}
\HT_{n} \ket{x} := \frac{1}{\sqrt{2^n}}\sum_{\xi\in\bits{n}} (-1)^{\xi\cdot x}\ket{\xi}, \label{eq:ht-def} 
\end{equation}
where $\xi\cdot x$ is the inner product modulo two between the $n$-bit strings $\xi$ and $x$ viewed as vectors. In this section the registers $X,Y$ are vectors in the $n$-qubit Hilbert space $ (\mathbb{C}^{2})^{\otimes n} $.

In what follows we first focus on \emph{full} oracles, i.e.\ not compressed ones. We analyze in detail the relations between different pictures of the oracles: the Standard Oracle, the Fourier Oracle, and the intermediate Phase Oracle. Next we provide an explicit algorithmic description of the compressed oracle and discuss the behavior of the compressed oracle in different pictures.

For the QROM, usually the Standard Oracle is the oracle used. The initial state of the oracle is the uniform superposition of truth tables $f$ representing functions $f:\bits{m}\to\bits{n}$. The Standard Oracle acts as follows
\begin{equation}
\StO_{\distrU} \ket{x,y}_{XY} \frac{1}{\sqrt{|\mathcal{F}|}} \sum_{f\in\mathcal{F}}\ket{f}_F =  \frac{1}{\sqrt{|\mathcal{F}|}} \sum_{f\in\mathcal{F}}\ket{x,y \xor f(x)}_{XY}\otimes\ket{f}_F, \label{sto-bits-def}
\end{equation}
where instead of modular addition we use bitwise XOR  denoted by $\xor\symbolindexmark{\xor}$. Note that in the above formulation $\StO_{\distrU}$ is just a controlled XOR operation from the $x$-th row of the truth table to the output register $Y$. We add the subscript $\distrU$ to denote that in the case of uniform distribution we also fix the input and output sets to bit-strings and the operation the oracle performs is not addition modulo $N$ like we introduced it in the main body. The register $F$ contains vectors in $(\mathbb{C}^{2})^{\otimes n 2^m}$.

The Fourier Oracle that stores the queries of the adversary is defined as
\begin{equation}
\FO_{\distrU} \ket{x,\eta}_{XY} \ket{\phi}_F := \ket{x,\eta}_{XY} \ket{\phi \xor \chi_{x,\eta}}_F,\label{eq:fo-bits-def}
\end{equation}
where $\chi_{x,\eta}:=(0^n,\dots,0^n,\eta,0^n,\dots,0^n)$ is a table with $2^m$ rows, among which only the $x$-th row equals $\eta$ and the rest are filled with zeros. Note that initially the $Y$ register is in the Hadamard basis, for that reason we use Greek letters to denote its value.	

To model the random oracle we initialize the oracle register $F$ in the Hadamard basis in the all $0$ state $\ket{\phi}=\ket{0^{n2^m}}$. 

If we take the Standard Oracle again and transform the adversary's $Y$ register instead, again using $\HT$, we recover the commonly used Phase Oracle.
More formally, the phase oracle is defined as 
\begin{equation}
\PhO_{\distrU}:=(\id_m^X\otimes\HT^Y_n)\otimes\id^F_{n2^m} \,\circ\, \StO_{\distrU} \,\circ\, (\mathbbm{1}^X_m\otimes\HT^Y_n)\otimes\id^F_{n2^m},\label{eq:pho-def}
\end{equation}
where $ \id_n $ is the identity operator acting on $ n $ qubits.

Applying the Hadamard transform also to register $F$ will give us the Fourier Oracle 
\begin{align}
\FO_{\distrU} = (\mathbbm{1}^{XY})\otimes\HT_{n 2^m}^F \,\circ\, \PhO_{\distrU} \,\circ\, (\mathbbm{1}^{XY})\otimes\HT_{n2^m}^F \, .
\end{align}

The above relations show that we have a chain of oracles, similar to Eq.~\eqref{eq:chainQFT}:
\begin{align}
\StO_{\distrU} \xleftrightarrow{\HT_{n}^Y} \PhO_{\distrU} \xleftrightarrow{\HT_{n2^m}^F} \FO_{\distrU}.\label{eq:chainHT}
\end{align}

In the following paragraphs we present some calculations explicitly showing how to use the technique and helping understanding why it is correct.

\subsubsection{Full Oracles, Additional Details}
In this section we show detailed calculations of identities claimed in Section~\ref{sec:uniform-oracles}.
First we analyze the Phase Oracle, introduced in Eq.~\eqref{eq:pho-def}.
We can check by direct calculation that this yields the standard Phase Oracle, $ \symbolindexmark{\PhO}  $
\begin{equation}
\PhO_{\distrU} \ket{x,\eta}_{XY} \ket{f}_F = (-1)^{\eta\cdot f(x)}\ket{x,\eta}_{XY} \ket{f}_F.
\end{equation}
Including the full initial state of the oracle register, we calculate
\begin{align}
&\PhO_{\distrU} \ket{x,\eta}_{XY} \frac{1}{\sqrt{|\mathcal{F}|}} \sum_{f\in\mathcal{F}}\ket{f}_F \nonumber\\
&= (\mathbbm{1}^X_m\otimes\HT_{n}^Y)\otimes\id_{n2^m}^F \StO_{\distrU} \ket{x}_X\frac{1}{\sqrt{2^n}}\sum_{y}(-1)^{\eta\cdot y}\ket{y}_Y  \frac{1}{\sqrt{|\mathcal{F}|}} \sum_{f\in\mathcal{F}}\ket{f}_F  \\
&=  (\mathbbm{1}_m^X\otimes\HT_{n}^Y)\otimes\id_{n2^m}^F \ket{x}_X\frac{1}{\sqrt{2^n}}\sum_{y} \sum_{f\in\mathcal{F}}(-1)^{\eta\cdot y}\ket{y\xor f(x)}_Y  \frac{1}{\sqrt{|\mathcal{F}|}}\ket{f}_F \\
&=  \frac{1}{\sqrt{|\mathcal{F}|}} \sum_{f\in\mathcal{F}} \ket{x}_X \sum_{\zeta}  \underset{=\delta(\eta,\zeta)(-1)^{\zeta\cdot f(x)}}{\underbrace{\frac{1}{2^n}\sum_{y}(-1)^{\eta\cdot y}(-1)^{(y \xor f(x))\cdot \zeta}}} \ket{\zeta}_Y\ket{f}_F \\
&=  \frac{1}{\sqrt{|\mathcal{F}|}} \sum_{f\in\mathcal{F}} (-1)^{\eta\cdot f(x)} \ket{x}_{X}\ket{\eta}_Y \ket{f}_F.  
\end{align}

Applying the Hadamard transform also to register $F$ will give us the Fourier Oracle. In the following calculation we denote acting on register $F$ with $\HT_{n2^m}^{\otimes 2^m}$ by $\HT_{n2^m}^{ F}$.
\begin{align}
& \HT_{n2^m}^{F}\circ \PhO_{\distrU}\circ \HT_{n2^m}^{F} \ket{x,\eta}_{XY} \ket{0^{2^m n}}_F = \HT_{n2^m}^{F} \frac{1}{\sqrt{|\mathcal{F}|}} \sum_{f\in\mathcal{F}} (-1)^{\eta\cdot f(x)} \ket{x,\eta}\ket{f}_F \nonumber \\
&=  \frac{1}{|\mathcal{F}|} \sum_{\phi, f} (-1)^{\phi\cdot f} (-1)^{\eta\cdot f(x)} \ket{x,\eta}\ket{\phi}_F \nonumber \\
&= \sum_{\phi} \underset{=\delta(\phi_{x'},0^n)}{\underbrace{\frac{1}{2^{n(2^m-1)}}\sum_{f(x'\neq x)} (-1)^{\phi_{x'}\cdot f(x')}}} \underset{=\delta(\phi_{x},\eta)}{\underbrace{\frac{1}{2^n} \sum_{f(x)} (-1)^{\phi_{x}\cdot f(x)} (-1)^{\eta\cdot f(x)}}} \ket{x,\eta}\ket{\phi}_F \nonumber \\
&= \ket{x,\eta}\ket{0^{2^m n}\xor \chi_{x,\eta}}
\end{align}
where we write $f(x)$ and $\phi_{x}$ to denote the $x$-th row of the truth table $f$ and $\phi$ respectively.

\subsubsection{Compressed Oracles, Additional Details}
Let us state the input-output behavior of the compressed oracle $ \CFO_{\distrU} $ for uniform distributions.
The input-output behavior of $\CFO_{\distrU}$ is given by the following equation, $x_r$ is the smallest $x_i\in D^X$ such that $x_r \geq x$:
\begin{align}\label{eq:cfo-in-out}
& \CFO_{\distrU} \ket{x,\eta}_{XY} \ket{x_1,\eta_1}_{D_1}\cdots \ket{x_{q-1},\eta_{q-1}}_{D_{q-1}}\ket{\perp,0^n}_{D_q}=\ket{x,\eta}_{XY} \ket{\psi_{r-1}} \nonumber\\
&\otimes \begin{cases}
\ket{x_{r},\eta_{r}}_{D_{r}}   \cdots  \ket{x_{q-1},\eta_{q-1}}_{D_{q-1}} \ket{\perp,0^n}_{D_q}  & \text{if } \eta=0^n, \\
\ket{x,\eta}_{D_{r}}\ket{x_{r},\eta_{r}}_{D_{r+1}} \cdots  \ket{x_{q-1},\eta_{q-1}}_{D_{q}} & \text{if } \eta\neq0^n, x\neq x_r,\\
\ket{x_{r},\eta_{r}\xor\eta}_{D_{r}}   \cdots  \ket{x_{q-1},\eta_{q-1}}_{D_{q-1}} \ket{\perp,0^n}_{D_q} & \text{if } \eta\neq0^n, x= x_r,\\
& \text{ } \eta\neq\eta_r,\\
\ket{x_{r+1},\eta_{r+1}}_{D_{r}}\cdots  \ket{x_{q-1},\eta_{q-1}}_{D_{q-2}} \ket{\perp,0^n}_{D_{q-1}} \ket{\perp,0^n}_{D_q}  & \text{if } \eta\neq0^n, x= x_r,\\
& \text{ } \eta=\eta_r,
\end{cases} 
\end{align}
where $\ket{\psi_{r-1}}:=\ket{x_1,\eta_1}_{D_1}\cdots\ket{x_{r-1},\eta_{r-1}}_{D_{r-1}}$.

In the following let us change the picture of the compressed oracle to see how the Compressed Standard Oracle and Compressed Phase Oracle act on basis states.
Let us begin with the Phase Oracle, given by the Hadamard transform of the oracle database 
\begin{equation}
\CPhO_{\distrU}:=\id_{n+m}\otimes\HT_{n}^{D^Y}\circ  \CFO_{\distrU}\circ \id_{n+m}\otimes\HT_{n}^{D^Y},
\end{equation}\symbolindexmark{\CPhO}
where by $\HT_{n}^{D^Y}$ we denote transforming just the $Y$ registers of the database: $\HT_{n}^{D^Y}:= (\id_m\otimes \HT_{n})^{\otimes q}$.
Let us calculate the outcome of applying $\CPhO$ to a state for the first time, for simplicity we omit all but the first register of $D$
\begin{align}
&\CPhO_{\distrU} \ket{x,\eta}_{XY} \frac{1}{\sqrt{2^n}}\sum_{z\in\bits{n}}\ket{\perp,z}_D  = \id_{n+m}\circ\HT_{n}^{D^Y}\circ \CFO_{\distrU}\ket{x,\eta}_{XY} \ket{\perp,0^n}_D \\
&= \id_{n+m}\circ\HT_{n}^{D^Y} \left( (1-\delta(\eta,0^n)) \ket{x,\eta}_{XY} \ket{x,\eta}_D +\delta(\eta,0^n) \ket{x,\eta}_{XY} \ket{\perp,0^n}_D \right) \\
&= \frac{1}{\sqrt{2^n}}\sum_{z\in\bits{n}} \left( (1-\delta(\eta,0^n)) (-1)^{\eta\cdot z} \ket{x,\eta}_{XY}\ket{x,z}_D + \delta(\eta,0^n) \ket{x,0^n}_{XY} \ket{\perp,z}_{D} \right). \label{eq:cpho-calc}
\end{align}
If we defined the Compressed Phase Oracle from scratch we might be tempted to omit the coherent deletion of $\eta=0^n$. The following attack shows that this would brake the correctness of the compressed oracles:
The adversary inputs the equal superposition in the $X$ register $\frac{1}{\sqrt{2^m}}\sum_x\ket{x,0^n}_{XY}$, after interacting with the regular $\CPhO_{\distrU} $ the state after a single query is
\begin{equation}
\frac{1}{\sqrt{2^m}}\sum_x\ket{x,0^n}_{XY} \overset{\CPhO_{\distrU} }{\mapsto}  \frac{1}{\sqrt{2^m}}\sum_x\ket{x,0^n}_{XY}\frac{1}{\sqrt{2^n}}\sum_{z}\ket{\perp,z}_{D},
\end{equation}
but with a modified oracle that does not take care of this deleting, simply omits the term with $\delta(\eta,0^n)$, let us call it $\CPhO'_{\distrU} $, the resulting state is
\begin{equation}
\frac{1}{\sqrt{2^m}}\sum_x\ket{x,0^n}_{XY} \overset{\CPhO'_{\distrU} }{\mapsto} \frac{1}{\sqrt{2^m}}\sum_x\ket{x,0^n}_{XY} \frac{1}{\sqrt{2^n}}\sum_{z}\ket{x,z}_{D}.
\end{equation}
Performing a measurement of the $X$ register in the Hadamard basis distinguishes the two states with probability $1-\frac{1}{2^m}$.

Let us inspect the state after making two queries to the Compressed Phase Oracle
\begin{align}
&\CPhO_{\distrU}\ket{x_2,\eta_2}_{X_2Y_2}\CPhO_{\distrU}\ket{x_1,\eta_1}_{X_1Y_1} \frac{1}{2^n}\sum_{z_1,z_2\in\bits{n}}\ket{\perp,z_1}_{D_1}\ket{\perp,z_2}_{D_2} \nonumber \\
&= \ket{x_2,\eta_2}\ket{x_1,\eta_1} \frac{1}{2^n}\sum_{z_1,z_2}\left(  (-1)^{\eta_1\cdot z_1}\delta(\eta_2,0^n)(1-\delta(\eta_1,0^n))  \underset{=|\psi^{\NOT})}{\underbrace{ \ket{x_1,z_1}_{F_1}\ket{\perp,z_2}_{F_2}}} \right. \nonumber\\
& + \delta(\eta_2,0^n)\delta(\eta_1,0^n)\underset{=|\psi^{\NOT})}{\underbrace{\ket{\perp,z_1}_{F_1}\ket{\perp,z_2}_{F_2}}} \nonumber \\
&+ (-1)^{\eta_2\cdot z_1}(1-\delta(\eta_2,0^n))\delta(\eta_1,0^n)\underset{=|\psi^{\ADD})}{\underbrace{\ket{x_2,z_1}_{F_1}\ket{\perp,z_2}_{F_2}}}   \nonumber\\
&\left.+ (-1)^{\eta_1\cdot z_1}(-1)^{\eta_2\cdot z_2}(1-\delta(\eta_2,0^n))(1-\delta(x_1,x_2))(1-\delta(\eta_1,0^n))\underset{=|\psi^{\ADD})}{\underbrace{\ket{x_1,z_1}_{F_1}\ket{x_2,z_2}_{F_2}}}  \right. \nonumber\\
&\left. +(1-\delta(\eta_2,0^n))\delta(x_1,x_2)\delta(\eta_1,\eta_2)(1-\delta(\eta_1,0^n))\underset{=|\psi^{\REM})}{\underbrace{\ket{\perp,z_1}_{F_1}\ket{\perp,z_2}_{F_2}}}  \right. \nonumber\\
& +(1-\delta(\eta_2,0^n))\delta(x_1,x_2)(1-\delta(\eta_1,\eta_2))(1-\delta(\eta_1,0^n))\nonumber\\
&\left.\cdot (-1)^{(\eta_1\xor\eta_2)\cdot z_1} \underset{=|\psi^{\UPD})}{\underbrace{\ket{x_1,z_1}_{F_1}\ket{\perp,z_2}_{F_2}}}  \right),
\end{align}
where by the superscripts we denote the operation performed by $ \CPhO_{\distrU} $ on the compressed database. By $ \ADD $ we denote adding a new pair $ (x,\eta) $, by $ \UPD $ changing the $ Y $ register of an already stored database entry, $ \REM $ signifies removal of a database entry, and $ \NOT $ stands for doing nothing, that happens if the queried $ \eta=0^n $. 

\medskip
Let us discuss the Compressed Standard Oracle. We know that it is the Hadamard transform of the adversary's register followed by $\CPhO_{\distrU} $
\begin{equation}
\CStO_{\distrU} =  \id_m\otimes\HT_{n}^Y \circ\CPhO_{\distrU} \circ \id_m\otimes\HT_{n}^Y. \label{eq:csto-cpho}
\end{equation}
Let us present the action of $\CStO$ in the first query of the adversary
\begin{align}\label{eq:csto-calc}
& \CStO_{\distrU} \ket{x,y}_{XY} \frac{1}{\sqrt{2^n}}\sum_{z\in\bits{n}}\ket{\perp,z}_D \nonumber\\
& = \id_m\otimes\HT_{n}^Y \circ\CPhO_{\distrU}  \frac{1}{\sqrt{2^n}}\sum_{\eta\in\bits{n}} (-1)^{\eta\cdot y} \ket{x,\eta}_{XY} \frac{1}{\sqrt{2^n}}\sum_{z\in\bits{n}}\ket{\perp,z}_D \\
& = \id_m\otimes\HT_{n}^Y \frac{1}{\sqrt{2^n}}\sum_{\eta\in\bits{n}} \frac{1}{\sqrt{2^n}}\sum_{z\in\bits{n}}(-1)^{\eta\cdot y} \Bigg( (1-\delta(\eta,0^n))  (-1)^{\eta \cdot z} \ket{x,\eta}_{XY}\ket{x,z}_D \nonumber\\
&+ \delta(\eta,0^n) \ket{x,0^n}_{XY} \ket{\perp,z}_{D} \Bigg)\\
& = \frac{1}{2^n} \sum_{y',\eta} \frac{1}{\sqrt{2^n}}\sum_{z}(-1)^{\eta\cdot y}(-1)^{y'\cdot\eta} \Bigg( (1-\delta(\eta,0^n)) (-1)^{\eta \cdot z} \ket{x,y'}_{XY}\ket{x,z}_D \nonumber\\
&+ \delta(\eta,0^n)  \ket{x,y'}_{XY} \ket{\perp,z}_{D} \Bigg) \\
&= \sum_{y'} \frac{1}{\sqrt{2^n}}\sum_{z} \underset{= \delta(y',y\xor z)- \frac{1}{2^n}}{\underbrace{\frac{1}{2^n} \sum_{\eta\neq 0} (-1)^{\eta\cdot y}(-1)^{y'\cdot\eta} (-1)^{\eta \cdot z}}} \ket{x,y'}_{XY}\ket{x,z}_D \nonumber\\
&+ \sum_{y'} \frac{1}{\sqrt{2^n}}\sum_{z}\frac{1}{2^n} \ket{x,y'}_{XY}\ket{\perp,z}_D\\
&= \frac{1}{\sqrt{2^n}}\sum_{z} \left(\ket{x,y\xor z}_{XY}\ket{x,z}_D -\frac{1}{2^n}\sum_{y'} \ket{x,y'}_{XY}\ket{x,z}_D +\frac{1}{2^n}\sum_{y'}\ket{x,y'}_{XY}\ket{\perp,z}_D \right).
\end{align}
We would like to note that a similar calculation and resulting state is presented in \cite{hosoyamada2019tight}.

\subsection{Detailed Algorithm for Alg.~\ref{alg:generalCFO}: \texorpdfstring{$\CFO_{\distrD}$}{CFOD}}\label{sec:algorithms}
In Algorithm~\ref{alg:detail} we present the fully-detailed version of Algorithm~\ref{alg:generalCFO}. This algorithm runs the following subroutines:
\begin{itemize}
	\item $\funLocate$, Function~\ref{fun:Locate}: This subroutine locates the positions in $\De$ where the $x-$entry coincides with the $x-$entry of the query. The result is represented as $q$ bits, where $q_i = 1 \iff \De_i^X = x$. This result is then bitwise XOR'ed into an auxiliary register $L$.
	\item $\funAdd$, Function~\ref{fun:Add}: This subroutine adds queried $ x $ to the database and take care of appropriate padding. Here our padding is simply $ (0^m, 0^n) $.
	\item $\funUpd$, Function~\ref{fun:Upd}: This subroutine updates the database by subtracting $ \eta $ after a suitable basis transformation.
	\item $\funRem$, Function~\ref{fun:Rem}: This subroutine removes $ (x ,0) $ entries from the database and puts them to the back in the form of padding.
	\item $\funClean$, Function~\ref{fun:Clean}: This subroutine cleans the auxiliary registers setting them back to initial values.
	\item $\funLarger\symbolindexmark{\funLarger}$: This subroutine  determines whether one value is larger than a second value, it works on three registers, say $ D^XXA $ and flips the bit in $ A $ if the value of $ D^X $ is larger than the value in $ X $, so 
	\begin{equation}
		\funLarger^{D^XXA}\ket{u}_{D^X}\ket{v}_X\ket{a}_A=\ket{u}_{D^X}\ket{v}_X \begin{cases}
		\ket{a\xor 1}_A \textnormal{ if } u>v \\
		\ket{a}_A \textnormal{ otherwise}
		\end{cases}.
	\end{equation} 
	In \cite{sena2007quantum} an efficient implementation of $ \funLarger $ for $ u,v $ being bitstrings can be found.
\end{itemize}
In the $\funAdd$ and $\funRem$ subroutine the unitary $\Puni$ can be found. $\Puni$ permutes the database  such that a recently removed entry in the database is moved to the end of the database. Conversely $\Puni^{-1}$ permutes the database such that an empty entry is created in the database as to ensure the correct ordering of the $x-$entries after adding the query into this newly created empty entry:
\begin{align} \label{eq:puni}
\Puni\ket{x_1,...,x_q} \otimes \ket{y_1,...,y_n} := \ket{\sigma_n \circ ... \circ \sigma_1(x_1,...,x_q)} \otimes \ket{y_1,...,y_n} \, ,
\end{align}\symbolindexmark{\Puni}
where $\sigma_i$ is applied conditioned on $y_i = 1$ and $\sigma_i(x_1,...,x_n) := (x_1,...,x_{i-2},x_{i-1},x_{i+1},x_{i+2},...,x_q,x_i)$.

\begin{algorithm} \label{alg:detail}
\DontPrintSemicolon
\SetKwInOut{Input}{Input}\SetKwInOut{Output}{Output}
\Input{Unprepared database and adversary query: $\ket{x,\eta}_{XY}\ket{\De}_D$}
\Output{$\ket{x,\eta}_{XY}\ket{\De'}_{D}$}
\BlankLine
\caption{Detailed $\CFO_{\distrD}$}\label{alg:generalCFO-detailed}
\BlankLine
$\ket{a}_A = \ket{0\in\bits{}}_A$ \tcp*[r]{initialize auxiliary register $A$}
$\ket{l}_L = \ket{0^q\in\bits{q}}_L$ \tcp*[r]{initialize auxiliary register $L$}
$\ket{l}_L \mapsto \funLocate(\ket{x}_{X}\ket{\De}_D\ket{l}_L)$ \tcp*[r]{locate $x$ in the database}
\If(\tcp*[f]{if not located}){$l = 0^q$}{
	$\ket{a}_A \mapsto \ket{a \oplus 1}_A$ \tcp*[r]{save result to register A}
}
\If(\tcp*[f]{if not located}){$a = 1$}{
    $\ket{\De}_D\ket{l}_L \mapsto \funAdd(\ket{x}_{X}\ket{\De}_D$) \tcp*[r]{add $x-$entry to the database}
}
$\ket{\De^Y}_{D^Y} \mapsto \funUpd (\ket{\eta}_{Y}\ket{\De^Y}_{D^Y}\ket{l}_L)$ \tcp*[r]{update register $D^Y$} 
$\ket{\De}_D\ket{l}_L \mapsto \funRem (\ket{x}_X\ket{\De}_D\ket{l}_L)$ \tcp*[r]{remove a database entry if $\ize=0$}
$\ket{a}_A \mapsto \funClean(\ket{y}_Y\ket{\De^Y}_{D^Y}\ket{l}_L)$ \tcp*[r]{uncompute register $A$}
$\ket{l}_L \mapsto \funLocate(\ket{x}_{X}\ket{\De}_D \ket{l}_L)$ \tcp*[r]{uncompute register $L$}
\KwRet $\ket{x,\eta}_{XY}\ket{\De'}_{D}$\tcp*[r]{$\De'$ is the modified database}
\end{algorithm}

\begin{subalgorithm}
	\caption{$\funLocate$\symbolindexmark{\funLocate}}\label{fun:Locate}
	\DontPrintSemicolon
	\SetKwInOut{Input}{Input}\SetKwInOut{Output}{Output} 
	\Input{$\ket{x}_{X}\ket{\De}_D\ket{l}_L$}
	\Output{$\ket{x}_{X}\ket{\De}_D\ket{l'}_L$}
	Set $\ket{a}_A = \ket{0\in\mathcal{X}}_A$ \tcp*[r]{initialize auxiliary register $ A $}
	\For{$i = 1,...,q$}{
		\If(\tcp*[f]{locate entries in the database}){$\ize_i \neq 0$}{ 
			$\ket{a}_{A} \mapsto \ket{a + (\De_i^X-x)}_{A}$ \tcp*[r]{database entry $ - $ query}
			\If(\tcp*[f]{locate matches in the database}){$a_i \neq 0$}{ 
				$\ket{l_i}_{L_i} \mapsto \ket{l_i \oplus 1}_{L_i}$ \tcp*[r]{save the corresponding positions} 
			}
			$\ket{a}_{A} \mapsto \ket{a - (\De_i^X-x)}_{A}$\tcp*[r]{uncompute register $ A $}
		}
	}
	\KwRet $\ket{x}_{X}\ket{\De}_D\ket{l'}_R$ \tcp*[r]{$l'$ contains the position of $x$ in $\De$}
\end{subalgorithm}

\begin{subalgorithm}
	\caption{$ \funAdd $\symbolindexmark{\funAdd}}\label{fun:Add}
	\DontPrintSemicolon
	\SetKwInOut{Input}{Input}\SetKwInOut{Output}{Output}
		\Input{$\ket{x}_{X}\ket{\De}_D\ket{l}_L$}
		\Output{$\ket{x}_{X}\ket{\De'}_D\ket{l'}_L$}
		Set $\ket{a}_A = \ket{0^q\in\bits{q}}_A$\tcp*{initialize auxiliary register A}
		\For{$i = 1,...,q$}{
			$\ket{a_i}_{A_i} \mapsto \funLarger(\ket{\De^X_i}_{D_i^X}\ket{x}_{X}\ket{a_i}_{A_i})$ \tcp*{check if database entry $ > $ query}
			\If(\tcp*[f]{correct for empty entries}){$\De_i^{X} \neq \perp$}{ 
				$\ket{a_i}_{A_i} \mapsto \ket{a_i \oplus 1}_{A_i}$\;
			}
			\For(\tcp*[f]{flip all higher entries}){$j=i+1,...,q$}{
				$\ket{a_j}_{A_j} \mapsto \ket{a_j \oplus a_i}_{A_j}$ \tcp*[f]{so we're left with one position}
			}
		}
		$\ket{\De}_D \mapsto \Puni^{-1}(\ket{\De}_D \otimes \ket{a}_A)$\tcp*{permute D to create empty entry}\tcp*{$\Puni$ is defined in~\eqref{eq:puni}}
		\For{$i = 1,...,q$}{
			\If(\tcp*[f]{look for this empty entry}){$a_i = 1$}{ 
				$\ket{\De_i^X}_{D_i^X} \mapsto \ket{\De_i^X - x}_{D_i^X}$\tcp*[r]{add $x-$entry to the database}
				$\ket{l_i}_{L_i} \mapsto \ket{l_i \oplus 1}_{L_i}$\tcp*[r]{update location register}
			}
		}
		\If(\tcp*[f]{Non zero $x$ implies non zero $a$}){$x \neq 0$}{
			\For{$i = 1,...,q$}{
				\If(\tcp*[f]{if located}){$l_i = 1$}{
					$\ket{a_i}_{A_i} \mapsto \ket{a_i \oplus 1}_{A_i}$ \tcp*[r]{uncompute register $ A $}
				}
			}
		}
		\KwRet $\ket{x}_{X}\ket{\De'}_D\ket{l'}_L$ \tcp*[r]{$\De'$ is the modified database}\tcp*[r]{$l'$ is modified $l$}
\end{subalgorithm}

\begin{subalgorithm}
	\caption{$ \funUpd $\symbolindexmark{\funUpd}}\label{fun:Upd}
	\DontPrintSemicolon
	\SetKwInOut{Input}{Input}\SetKwInOut{Output}{Output}
	\Input{$\ket{\eta}_{Y}\ket{\De^Y}_{D^Y}\ket{l}_L$}
	\Output{$\ket{\eta}_{Y}\ket{\De'^Y}_{D^Y}\ket{l}_L$}
		Apply $\QFT_{N}^{D^Y}\Samp_{\distrD}$ \tcp*{transform to the Fourier basis}
		\For{$i = 1,...,q$}{
	    	\If(\tcp*[f]{if located}){$l_i = 1$}{
				$\ket{\Delta_i^Y}_{D_i^Y} \mapsto \ket{\Delta_i^Y - \eta}_{D_i^Y}$\tcp*{Update the Y register of entry}
	    	}
	    }
	    Apply $\Samp^{\dagger}_{\distrD}\QFT^{\dagger D^Y}_{N}$ \tcp*{transform back to the unprepared database}
	    \KwRet $\ket{\eta}_{Y}\ket{\De'^Y}_{D^Y}\ket{l}_L$ \tcp*{$\De'^Y$ is modified $Y$ register of the database}
\end{subalgorithm}

\begin{subalgorithm}
	\caption{$ \funRem $\symbolindexmark{\funRem}}\label{fun:Rem}
	\DontPrintSemicolon
	\SetKwInOut{Input}{Input}\SetKwInOut{Output}{Output}
	\Input{$\ket{x}_{X}\ket{\De}_D\ket{l}_L$}
	\Output{$\ket{x}_{X}\ket{\De'}_D\ket{l'}_L$}   
		Set $\ket{a}_A = \ket{0^q}_A$\tcp*{initialize auxiliary register $A$}
		Set $\ket{b}_B = \ket{0}_B$\tcp*{initialize auxiliary register $B$}
		\For{$i = 1,...,q$}{
			\If{$l_i = 1$}{
				\If(\tcp*[f]{if entry is incorrect}){$\ize_i = 0$}{
					$\ket{\De_i^X}_{D_i^X} \mapsto \ket{\De_i^X - x}_{D_i^X}$\tcp*{remove the entry}
					$\ket{b}_B \mapsto \ket{b \oplus 1}_B$ \tcp*{save that we have removed an entry}
				}
			}
		}
		\If(\tcp*[f]{if we removed an entry}){$b = 1$}{
			\For{$i = 1,...,q$}{
				$\ket{a_i}_{A_i} \mapsto \funLarger(\ket{x}_{X},\ket{\De_i^X}_{D_i^X},\ket{a_i}_{A_i})$ \tcp*{check if query $>$ database entry}
				\If(\tcp*[f]{Correct for x = 0}){$x = 0$}{
					\If(\tcp*[f]{correct for empty entries}){$\De_i^{Y} \neq 0$}{ 
						$\ket{a_i}_{A_i} \mapsto \ket{a_i \oplus 1}_{A_i}$\;
					}
				}
				\For(\tcp*[f]{flip all lower entries}){$j=i-1,...,1$}{
					$\ket{a_j}_{A_j} \mapsto \ket{a_j \oplus a_i}_{A_j}$ \tcp*{so we're left with only the removed position}
				}
				$\ket{l_i}_{L_i} \mapsto \ket{l_i \oplus a_i}_{L_i}$ \tcp*{correct for the removed entry}
			}
			$\ket{\De}_D \mapsto P(\ket{\De}_D \otimes \ket{a}_A)$\tcp*{permute D to move the empty entry}
			\For(\tcp*[f]{uncompute register A}){$i = q,...,1$}{
				\For(\tcp*[f]{by calculating the first position}){$j=q,...,i+1$}{
					$\ket{a_j}_{A_j} \mapsto \ket{a_j \oplus a_i}_{A_j}$\tcp*[f]{such that database entry > query}
				}
				\If(\tcp*[f]{as in the $\funAdd$ subroutine}){$\De_i^{Y} \neq 0$}{ 
					$\ket{a_i}_{A_i} \mapsto \ket{a_i \oplus 1}_{A_i}$\;
				}
				$\ket{a_i}_{A_i} \mapsto \funLarger(\ket{\De_i^X}_{D_i^X},\ket{x}_{X},\ket{a_i}_{A_i})$\; 
			}
		}
		$\ket{a}_A \mapsto \funLocate(\ket{x}_{X}\ket{\De}_D \ket{l}_A)$\;
		\If(\tcp*[f]{check if we have removed}){$A = 0^q$}{
			$\ket{b}_B \mapsto \ket{b \oplus 1}_B$ \tcp*{Uncompute register $B$}	
		}
		$\ket{a}_A \mapsto \funLocate(\ket{x}_{X}\ket{\De}_D \ket{l}_A)$ \tcp*[r]{uncompute register $A$}
		\KwRet $\ket{x}_{X}\ket{\De'}_D\ket{l'}_L$ \tcp*{$\De'$ is modified database}\tcp*{$l'$ is modified $l$}
\end{subalgorithm}

\begin{subalgorithm}
	\caption{$ \funClean $\symbolindexmark{\funClean}}\label{fun:Clean}
	\DontPrintSemicolon
	\SetKwInOut{Input}{Input}\SetKwInOut{Output}{Output}
	\Input{$\ket{\eta}_{Y}\ket{\De^Y}_D\ket{l}_L\ket{a}_A$}
	\Output{$\ket{\eta}_{Y}\ket{\De^Y}_D\ket{l}_L\ket{a'}_A$}   
		Set $\ket{b}_B = \ket{0\in\mathcal{Y}}_B$\tcp*{initialize auxiliary register $B$}
		Apply $\QFT_{N}^{D^Y}\Samp_{\distrD}$ \tcp*{transform to the Fourier basis}
		\For{$i = 1,...,q$}{
			\If{$l_i = 1$}{
				$\ket{b}_{B} \mapsto \ket{b + (\Delta_i^Y-\eta)}_{B}$ \tcp*[r]{database entry $ - $ query}
				\If(\tcp*[f]{locate matches in the database}){$b = 0$}{ 
					\If(\tcp*[f]{if we added}){$\eta \neq 0$}{ 
						$\ket{a}_A \rightarrow \ket{a \oplus 1}_A$
					}
				}
				$\ket{b}_{B} \mapsto \ket{b - (\Delta_i^Y-\eta)}_{B}$ \tcp*[r]{uncompute register $B$}
			}
		}
		Apply $\Samp^{\dagger}_{\distrD}\QFT^{\dagger D^Y}_{N}$\tcp*{transform back to the unprepared database}
		\KwRet $\ket{\eta}_{Y}\ket{\De^Y}_D\ket{l}_L\ket{a'}_A$ \tcp*{$a'$ is modified register $A$}
\end{subalgorithm}

\newpage

\section{Collapsingness of Sponges}\label{sec:collapsingness}
Collapsingness is a security notion defined in \cite{Unruh2016a}; It is a purely quantum notion strengthening collision resistance. It was developed to capture the required feature of hash functions used in cryptographic commitment protocols.

In this section we prove that quantum indifferentiability implies collapsingness.
We begin by introducing the notion of \emph{collapsing} functions.

For quantum algorithms $\advA$, $\advB$ with quantum access to $ \Ho $, consider the following games:
\begin{align}
	\Collapse{1}: {}\quad  & (S,M,h) \gets \advA^{\Ho}(), \,         m\gets \meas(M),  \, b\gets \advB^{\Ho}(S,M),\label{eq:collapse-game1} \\
	\Collapse{2}: {}\quad  & (S,M,h) \gets \advA^{\Ho}(),\,  \phantom{m\gets\meas(M),{}}\, b\gets \advB^{\Ho}(S,M).\symbolindexmark{\Collapse}
\end{align}
Here $S,M$ are quantum registers. $\meas(M)$ is a measurement of $M$ in the computational basis.
The intuitive meaning of the above games is that part $ \advA $ of the adversary prepares a quantum register $ M $ that holds a superposition of inputs to $ \Ho $ that all map to $ h $. Then she sends $ M $ along with the side information $ S $ to $ \advB $. The task of the second part of the adversary is to decide whether measurement $ \meas $ of the register $ M $ occurred or not.

We call an adversary $(\advA,\advB)$ \emph{valid} if and only if $\PR[\Ho(m)=h]=1$ when we run $(S,M,h)\gets \advA^{\Ho}()$ in $ \Collapse{1} $ from Eq.\eqref{eq:collapse-game1} and measure $M$ in the computational basis as $m$.

\begin{defi}[Collapsing \cite{Unruh2016a}]\label{def:collapse}
	A function $\Ho$ is \emph{collapsing} if for any valid quantum-polynomial-time adversary $(\advA,\advB)$
	\begin{align}
		\abs{\PR[b=1:\Collapse{1}]-\PR[b=1:\Collapse{2}]} <\eps,
	\end{align}
	where the \textnormal{collapsing-advantage} $ \eps $ is negligible.
\end{defi}
A more in-depth analysis of this security notion can be found in \cite{Unruh2016a,Unruh2016,czaj18sponge,fehr2018}.

It was shown in \cite{Unruh2016a} that if $ \Ho $ is a random oracle then is it collapsing:
\begin{lemm}[Lemma 37 \cite{Unruh2016a}]\label{lem:collapsing-ro}
	Let $ \Ho :\mathcal{X}\to\mathcal{Y} $ be a random oracle, then any valid adversary $ (\advA^{\Ho},\advB^{\Ho}) $ making $ q $ quantum queries to $ \Ho $ has collapsing-advantage $ \eps\in \cO{\sqrt{\frac{q^3}{\abs{\mathcal{Y}}}}} $.
\end{lemm}

In the rest of this section we state and prove that any function that is indifferentiable from a collapsing function is itself collapsing. In the context of sponges, together with thm.~\ref{thm:quantumindiff-trans}, we reprove the result of \cite{czaj18sponge} in a modular way that might come useful when indifferentiability of sponges with permutations is established.
\begin{thm}[Quantum indifferentiability preserves collapsingness]\label{thm:indiff2collaps}
	Let $ \sysC $ be a construction based on an internal function $f$, and let $\sysC$ be $ (q,\eps_{I}(q)) $-indifferentiable from an ideal function $\sysC_{\mathrm{ideal}}$ with simulator $S$. Assume further that $\sysC_{\mathrm{ideal}}$ allows for a collapsingness advantage at most $\eps_{\mathrm{coll}}(q)$ for a $q$-query adversary. Then $ \sysC $ is collapsing with advantage $ \eps_{\mathrm{coll}}(q_{\sysC}, q_f)= 2\eps_{I}(q_{\sysC}+q_f)+\eps_{\mathrm{coll}}(q_{\sysC}+\alpha q_f)  $, where $q_{\sysC}$ and $q_f$ are the number of queries to $ \sysC$ and $f$, respectively, and $\alpha$ is the number of queries simulator $\Sim$ makes (at most) to $\sysC_{\mathrm{ideal}}$ for each time it is queried.
\end{thm}
\begin{proof}
	Given a collapsingness distinguisher $\tilde{\advD}$ against $\sysC$ with advantage $\eps\ge \eps_{\mathrm{coll}}(q_{\sysC}+\alpha q_f)$ that makes $q_{\sysC}$ queries to $\sysC$ and $q_f$ queries to $f$, we build an indifferentiability distinguisher $\advD$ as follows. Chose $b\in\{0,1\}$ at random. Running $\tilde{\advD}$, if $b=0$ simulate $\Collapse{1}$, if $b=1$ simulate $\Collapse{2}$. Output 1 if $\tilde{\advD}$ outputs $b$, and 0 else.
	
	In the real world, we have that 
	\begin{align*}
		\PR[1\leftarrow\advD:\mathbf{Real}]&=\frac{1}{2}\left(\PR[0\leftarrow \tilde{\advD}^{\sysC,f}:\Collapse{1}]+\PR[1\leftarrow \tilde{\advD}^{\sysC,f}:\Collapse{2}]\right)\\
		&=\frac 1 2 +\frac 1 2 \left(\PR[1\leftarrow \tilde{\advD}^{\sysC,f}:\Collapse{2}]-\PR[1\leftarrow \tilde{\advD}^{\sysC,f}:\Collapse{1}]\right).
	\end{align*}
	In the ideal world, the distinguisher together with the simulator $S$ can be seen as a collapsingness distinguisher for $\sysC_{\mathrm{ideal}}$. Therefore we get
	\begin{align*}
		\PR[1\leftarrow \advD:\mathbf{Ideal}]&=\frac 1 2 +\frac 1 2 \left(\PR[1\leftarrow \tilde{\advD}^{\sysC_{\mathrm{ideal}},\Sim}:\Collapse{2}]-\PR[1\leftarrow \tilde{\advD}^{\sysC_{\mathrm{ideal}},\Sim}:\Collapse{1}]\right)
	\end{align*}
	and hence
	\begin{align*}
		&\Big|\PR[1\leftarrow \advD:\mathbf{Real}]-\PR[1\leftarrow \advD:\mathbf{Ideal}]\Big|\nonumber\\
		&=\frac 1 2 \Big|\PR[1\leftarrow \tilde{\advD}^{\sysC,f}:\Collapse{2}]-\PR[1\leftarrow \tilde{\advD}^{\sysC,f}:\Collapse{1}]\\
		&-\PR[1\leftarrow \tilde{\advD}^{\sysC_{\mathrm{ideal}},\Sim}:\Collapse{2}]+\PR[1\leftarrow \tilde{\advD}^{\sysC_{\mathrm{ideal}},\Sim}:\Collapse{1}]\Big|\\
		&\ge \frac 1 2\left(\eps- \eps_{\mathrm{coll}}(q_{\sysC}+\alpha q_f)\right).
	\end{align*}
\end{proof}

\end{appendix}

\newpage
\tocsectionstar{Symbol Index}\label{sec:notation}
\begin{center}
	\begin{longtable}{lp{3.3in}l}
		\InputIfFileExists{\jobname.sdx}{}{}
	\end{longtable}
\end{center}

\end{document}